\documentclass[aps,prxquantum,twocolumn,superscriptaddress,nofootinbib]{revtex4-2}

\usepackage{amsmath,amssymb,amsfonts,amsthm}
\usepackage{graphicx}
\usepackage{booktabs}
\usepackage{multirow}
\usepackage{hyperref}
\usepackage{listings}
\usepackage{xcolor}

\usepackage[normalem]{ulem}

\lstdefinestyle{diff}{
  basicstyle=\ttfamily\scriptsize,
  frame=single,
  xleftmargin=2pt,
  xrightmargin=2pt,
  breaklines=true,
  columns=flexible,
}

\newtheorem{theorem}{Theorem}
\newcommand{\FF}{\ensuremath{\mathbb{F}}}
\newcommand{\ZZ}{\ensuremath{\mathbb{Z}}}
\newcommand{\dblbrk}[1]{\ensuremath{\lbrack\!\lbrack #1 \rbrack\!\rbrack}}
\newcommand{\code}[1]{\dblbrk{#1}}
\newcommand{\FOM}{\ensuremath{\mathrm{FOM}}}

\begin{document}

\title{Evolutionary Discovery of Bivariate Bicycle Codes\\with LLM-Guided Search}

\author{Juan Cruz-Benito}
\affiliation{IBM Research, IBM T. J. Watson Research Center, NY, USA}
\author{Andrew W. Cross}
\affiliation{IBM Research, IBM T. J. Watson Research Center, NY, USA}
\author{David Kremer}
\affiliation{IBM Research, IBM T. J. Watson Research Center, NY, USA}
\author{Ismael Faro}
\affiliation{IBM Research, IBM T. J. Watson Research Center, NY, USA}

\date{\today}

\begin{abstract}
Quantum LDPC code discovery requires searching large algebraic design spaces while reliably certifying the parameters and equivalence classes of any candidates found.
We introduce an LLM-guided evolutionary workflow in which language models mutate Python programs that generate bivariate-bicycle and perturbed bivariate-bicycle code ansätze.
Across five campaigns, the system performed approximately 1{,}650 evolutionary iterations, screened about $2 \times 10^5$ candidate codes, and required ${\sim}140$ hours of computation and ${\sim}$US\$400 in LLM inference cost.
Candidate codes are evaluated through a staged validation pipeline combining $\mathrm{GF}(2)$ rank computation, distance estimation and certification, mixed-integer linear programming, BLISS Tanner-graph deduplication, decomposability analysis, and local-Clifford equivalence checks.
At block length $n \leq 360$, the workflow identifies 465 distinct candidate codes: 97 CSS bivariate-bicycle codes and 368 non-CSS perturbed variants.
The CSS search recovers known high-performing codes and finds new finite-length representatives, including an indecomposable $\code{288,16,12}$ code and higher-weight codes with up to $k = 50$ at distance $d = 8$.
The non-CSS search produces perturbed codes matching the gross-code figure of merit at $\code{144,12,12}$, along with additional high-distance candidates reported as certified values or upper bounds according to MILP status.
Overall, these results show that LLM-guided program evolution can serve as a practical tool for structured quantum-code discovery when paired with independent evaluation.
\end{abstract}

\maketitle

\section{Introduction}
\label{sec:intro}

Quantum low-density parity-check (qLDPC) codes are central to practical fault-tolerant quantum computing~\cite{breuckmann2021quantum}.
Among these, bivariate bicycle (BB) codes~\cite{bravyi2024high} are particularly attractive for near-term implementations: their weight-6 stabilizers and constant-depth syndrome extraction yield excellent rate--distance--threshold tradeoffs at practical block lengths ($n \lesssim 1000$).
The $\code{144,12,12}$ ``gross code'' achieves a pseudo-threshold of ${\sim}0.7\%$ with a $10\times$ reduction in qubit count versus the surface code~\cite{bravyi2024high}, and a baseline modular fault-tolerant architecture has been designed around it~\cite{yoder2025tour}.

Despite this progress, the landscape of BB codes at practical block lengths remains largely unexplored.
All weight-6 codes reported in~\cite{bravyi2024high} have $k \leq 12$; subsequent work reached $k = 16$ at $n = 150$~\cite{wang2024coprime}, while systematic enumerations~\cite{liang2025generalized,lin2024abelian} have catalogued weight-6 codes for specific polynomial forms and small block lengths. 
Whether higher encoding rates are achievable at weight~6---and at what cost to distance---has not been systematically investigated.
The search space is purely combinatorial: every trinomial pair over $\FF_2[x,y]/(x^\ell{-}1, y^m{-}1)$ defines a valid code, with $O(\ell^6 m^6)$ pairs per lattice ---reduced by symmetries such as translation invariance, but still vast ---and no gradient structure amenable to continuous optimization.

We address this with a large language model (LLM) guided evolutionary search that evolves \emph{generator ansätze}---Python programs producing candidate polynomial pairs for arbitrary lattice dimensions, without \emph{a priori} guarantee that all outputs are valid codes---rather than individual codes.
Inspired by FunSearch~\cite{romera2024funsearch} and AlphaEvolve~\cite{novikov2025alphaevolve}, our framework applies LLM-guided program evolution to quantum code discovery; the technical basis is FunSearch via the open-source OpenEvolve~\cite{openevolve} implementation.
The ansatz representation is key: a single mutation can express an algebraic pattern (e.g., ``use $x^{\ell/3}$'') that generalizes across all lattice dimensions, biasing the search toward algebraically regular families. 
Five evolution campaigns---three targeting CSS trinomial codes (weight-3 polynomials yielding weight-6 stabilizers), one exploring mixed-monomial polynomials (terms $x^a y^b$ with $a,b > 0$; Sec.~\ref{sec:families}), and one extending to non-CSS perturbed bivariate bicycle (PBB) codes---discover 465 distinct codes at $n \leq 360$, including CSS codes at $d = 12$--$14$ not present in prior catalogs and non-CSS codes matching the gross code's figure of merit.

MILP distance computation on all verified codes---exact where optimality is certified, otherwise yielding upper bounds---reveals an empirical regularity within the searched catalog: the weight-6 bivariate bicycle family is consistent with a rate--distance tradeoff (higher~$k$ accompanies lower~$d$, and vice versa; Sec.~\ref{sec:tradeoff}), though we do not prove this is a structural constraint of the construction.
The highest-$k$ weight-6 codes ($k = 40$--$54$) are equivalent to the hypergraph product (HGP) of two low-distance cyclic codes~\cite{eberhardt2024pruning,aydin2025cxc}---a known construction the evolution converged on independently---and universally have $d \leq 4$.
The $x/y$-swap codes (polynomials mixing both cyclic variables, e.g., $A = x^a{+}y^b{+}y^c$; Sec.~\ref{sec:families}) that achieve $d = 12$ are limited to $k \leq 24$, and the $k = 24$ endpoint is achieved only by a decomposable code (one whose Tanner graph disconnects into independent sub-codes---here, a direct sum of two gross codes; Sec.~\ref{sec:families}); higher-weight (weight-8) mixed-monomial codes access new $(k,d)$ combinations (e.g., $k = 50$ at $d = 8$) but do not escape this envelope.
All BB codes with $A = B$ have $d = 2$ exactly (Appendix~\ref{app:ab_trap})---a structural ``distance trap'' that BP-OSD fails to detect even at $1.5 \times 10^6$ trials, and that the pipeline's MILP verification identifies in under one second.

This CSS tradeoff raises a natural question: can non-CSS stabilizers access a more favorable rate--distance regime?
Standard BB codes are CSS by construction---their checks can be separated into X and Z type---which constrains the available stabilizer structures.
We introduce \emph{perturbed bivariate bicycle} (PBB) codes, which augment the X-type stabilizer block with Z-type support via two additional polynomials $(C, D)$, creating mixed stabilizers that couple X and Z (Sec.~\ref{sec:pbb}).
This non-CSS ansatz reproduces parameter regimes accessible to CSS codes through a structurally distinct stabilizer pattern: a $\code{144,12,12}$ PBB code matches the gross code's $\FOM$ with a non-CSS stabilizer structure.
However, the PBB ansatz produces non-CSS codes that match but do not exceed the gross code's $\FOM$ at $n = 144$. Whether non-CSS constructions over the same polynomial ring can escape the CSS rate--distance envelope remains open.

A second contribution is methodological.
The MILP analysis that grounds the tradeoff also reveals that BP-OSD systematically overestimates distance for high-rate codes ($k/n > 0.1$), with overestimates reaching $12\times$; 147 of 154 trinomial code representations had bounds tightened under a multi-decoder protocol.
For non-CSS codes, per-channel BP-OSD distance bounding~\cite{bravyi2024high} fails because the achievable logical cosets form a strict subspace of the full logical space---most randomly sampled cosets have no solution in any pure error channel---so the decoder returns trivial results; we introduce achievable-syndrome sampling (restricting to the achievable subspace per channel) to restore functionality.
These findings indicate that distance claims for high-$k$ BB codes require exact verification.

Our main contributions are:
\begin{enumerate}
\item An LLM-guided evolutionary pipeline with multi-stage verification---combining MAP-Elites program evolution, a $k$-then-distance evaluation cascade, trust filtering (Sec.~\ref{sec:trust}), MILP distance computation, BLISS~\cite{junttila2007engineering} Tanner-graph deduplication, and post-campaign checks (decomposability, Clifford equivalence)---that provides a reusable framework for quantum code discovery.

\item As an illustration of this framework, five campaigns discover 465 distinct codes---97 CSS (weight-6 and weight-8) and 368 non-CSS PBB---with encoding dimensions up to $k = 54$ (prior weight-6 maximum: $k = 16$).
The most practically relevant codes include the CSS $\code{288,16,12}$ ($d = 12$ exact, all shifts~$\leq 3$), the non-CSS $\code{360,12,{\leq}24}$ (highest $\FOM \leq 19.2$, with LER ${<}p$ at all tested rates), and the $\code{144,12,12}$ PBB (matching $\FOM = 12.0$ via mixed X/Z stabilizers).
Tanner graph analysis reveals that the apparent highest-$k$ CSS code at $d=12$, $\code{288,24,12}$, is a direct sum of two gross codes---demonstrating the pipeline's ability to detect and exclude such codes. 

\item Structural characterization of the BB code landscape via MILP distance computation on all codes (exact where certified, upper bounds otherwise), revealing a systematic rate--distance tradeoff (indecomposable weight-6 $d = 12$ codes limited to $k \leq 16$; weight-6 $k > 24$ implies $d \leq 4$; weight-8 codes access new $(k,d)$ points but do not escape the envelope), a proof that all $A = B$ codes have $d = 2$ regardless of check weight (Appendix~\ref{app:ab_trap}), and identification of four algebraic families (univariate/HGP, $x/y$-swap, mixed-monomial/higher-weight, non-CSS PBB) with distinct rate--distance profiles.

\item Systematic quantification of BP-OSD overestimation (up to $12\times$) via MILP ground truth, a multi-decoder verification protocol, and achievable-syndrome sampling for non-CSS codes---establishing a verification standard for high-rate BB code claims.
\end{enumerate}

\section{Related Work}
\label{sec:related}

Our work connects four research threads: BB codes within the qLDPC landscape, distance computation methodology, computational code discovery, and LLM-guided program synthesis.

\paragraph{BB codes and the qLDPC landscape.}
Quantum LDPC codes offer constant-rate encodings with low-weight checks as an alternative to the surface code, the most widely studied qLDPC code family~\cite{breuckmann2021quantum}. 
Asymptotic breakthroughs---fiber bundle codes~\cite{hastings2021fiber}, balanced product codes~\cite{breuckmann2021balanced}, and three independent proofs of asymptotically good quantum LDPC codes~\cite{panteleev2022asymptotically,leverrier2022quantum,dinur2023good}---established the theoretical frontier.
BB codes~\cite{bravyi2024high}, as abelian lifted product codes over $\ZZ_\ell \times \ZZ_m$~\cite{panteleev2022quantum}, are asymptotically bad ($d = O(\sqrt{n})$~\cite{postema2025existence}) but offer favorable rate--distance--threshold tradeoffs compared to other known qLDPC families at block lengths $n \lesssim 1000$, including hypergraph product codes.

Extensions include trivariate bicycle codes~\cite{voss2025multivariate}, coprime-index constructions reaching $k = 16$~\cite{wang2024coprime,rateadjustable2025}, covering-graph families at weight~8~\cite{symons2025covering}, and generalized toric codes on twisted tori~\cite{liang2025generalized}. 
On the logical gate front, fold-transversal gates~\cite{eberhardt2025logical}, qLDPC surgery~\cite{cross2024improved}, gauging~\cite{williamson2024gauging}, and homological measurement~\cite{ide2024homological} make high-$k$ codes practically valuable, since each additional logical qubit reduces overhead.
Hardware demonstrations include a BB code on a 32-qubit processor~\cite{wang2025demonstration}. 
Beyond CSS, Khesin and Lu~\cite{khesin2026mirror} introduced mirror codes---a non-CSS construction over abelian groups with notable weight-6 parameters $\code{60,4,10}$ and $\code{85,8,9}$. 

Despite this activity, the high-$k$ regime remains relatively unexplored at $n = 144$--$360$: the highest-$k$ published weight-6 codes have $k \leq 16$~\cite{bravyi2024high,wang2024coprime}, and the most comprehensive search~\cite{liang2025generalized} is restricted to a canonical 3-term polynomial form.

\paragraph{Distance computation.}
Computing the minimum distance of a stabilizer code is in general computationally hard~\cite{kapshikar2022hardness}.
BP-OSD~\cite{roffe2020decoding,panteleev2021degenerate}, used for distance bounding in~\cite{bravyi2024high}, is the most common heuristic, with the \texttt{ldpc} library~\cite{ldpc_library} exposing multiple BP variants and OSD methods; however, it provides only stochastic upper bounds.
Exact methods include mixed integer programming (MIP), Brouwer--Zimmermann enumeration, and SAT solvers; for a comprehensive comparison of heuristic and exact distance-finding algorithms, see Webster et al.~\cite{webster2026distance}.
The concurrent work of Webster et al.\ confirms that heuristic methods can significantly overestimate distance, observing that BP-OSD (min-sum, OSD-CS order~10) fails to find the correct distance for codes as small as $\code{48,5,10}$ without enhanced sampling strategies.
Our work complements theirs by quantifying BP-OSD overestimation specifically for bivariate bicycle codes across a 465-code catalog with MILP ground truth, revealing overestimates up to $12\times$ for high-rate codes.
While~\cite{bravyi2024high} complemented BP-OSD with MIP for select codes, many subsequent works~\cite{wang2024coprime,chengyu2026bayesian} report distance estimates using a single decoder configuration with modest trial counts---a gap our systematic multi-decoder verification protocol and comprehensive MILP analysis address.

\paragraph{Computational code discovery.}
Reinforcement learning has been applied to surface code optimization~\cite{nautrup2019optimizing}, Quantum Lego codes up to ${\sim}20$ qubits~\cite{su2025discovery}, code-and-encoder co-discovery~\cite{olle2024simultaneous}, weight reduction~\cite{he2025discovering}, and hypergraph product optimization~\cite{freire2025optimizing}.
Classical approaches include structured constructions such as univariate bicycle codes~\cite{rabeti2025list} and coprime-index codes~\cite{wang2024coprime}, as well as decoder-in-the-loop genetic optimization~\cite{elkelesh2019decoder}.
Chengyu et al.~\cite{chengyu2026bayesian} propose Bayesian optimization for BB codes, reporting a $\code{144,36,?}$ code without verified distance.
These approaches typically optimize individual codes at fixed parameters $(n,k)$.
Our approach differs: the LLM evolves a \emph{program} (ansatz) that generates candidates across multiple block lengths simultaneously, enabling the discovery of algebraic patterns that generalize. 

\paragraph{LLM-guided program synthesis.}

FunSearch~\cite{romera2024funsearch} introduced LLMs as mutation operators within evolutionary search, discovering cap set constructions and bin packing heuristics.
AlphaEvolve~\cite{novikov2025alphaevolve} extended the FunSearch approach from single functions to multi-file codebases, allowing LLM-guided mutations across interacting modules.
Open-source implementations include OpenEvolve~\cite{openevolve} (MAP-Elites, used here), CodeEvolve~\cite{codeevolve2026}, and ShinkaEvolve~\cite{lange2025shinkaevolve}.
In coding theory, Weindel and Heckel~\cite{weindel2025llmguided} apply FunSearch to deletion-correcting codes (classical), matching the maximum sizes of single-deletion codes and improving the best-known lower bounds for two-deletion codes at $n \in \{12, 13, 16\}$.
For quantum codes, the contemporaneous Cain et al.~\cite{cain2026shor} report new lifted-product codes obtained via an ``LLM-assisted heuristic computer search'', with the search itself deferred to an in-preparation companion (\textit{Atlas of atomic fault-tolerant quantum memory}, Caltech \& Oratomic) that had not appeared at the time of writing.
To our knowledge, ours may be the first full account in the published literature of LLM-guided program evolution applied to quantum code discovery.
We adapt the FunSearch paradigm with domain-specific evaluation and verification stages described in Secs.~\ref{sec:method}--\ref{sec:verification}.

\section{Preliminaries}
\label{sec:prelim}

\subsection{Bivariate bicycle codes}
\label{sec:bb}

A Calderbank--Shor--Steane (CSS) code~\cite{calderbank1996good,steane1996multiple} has parity checks that can be separated into X-type and Z-type. 
A bivariate bicycle (BB) code~\cite{bravyi2024high} is a CSS code defined over the quotient ring $R = \FF_2[x,y]/(x^\ell{-}1, y^m{-}1)$ by two weight-3 polynomials (trinomials) $A, B \in R$.
The parity check matrices are
\begin{equation}
H_X = \begin{pmatrix} A & B \end{pmatrix}, \quad
H_Z = \begin{pmatrix} B^\top & A^\top \end{pmatrix},
\end{equation}
where $A$ and $B$ denote $\ell m \times \ell m$ circulant matrices and $A^\top$ is the image under the involution $x \mapsto x^{-1}, y \mapsto y^{-1}$ (equivalently, the matrix transpose of the cyclic-shift representation).
The trinomial restriction yields weight-6 stabilizers, matching the original definition of~\cite{bravyi2024high}; some authors extend to higher-weight polynomials~\cite{panteleev2021degenerate,liang2025selfdual}, but the landscape of code parameters depends on check weight, so we separate claims accordingly.
The CSS condition $H_X H_Z^\top = AB + BA = 0$ over $\FF_2$ follows from the commutativity of~$R$.
The code parameters are $\code{n,k,d}$ with $n = 2\ell m$ physical qubits, $k = 2\ell m - 2\,\mathrm{rank}_{\FF_2}(H_X)$, and $d$ the minimum weight of a nontrivial logical operator.
Campaigns~1--3 search over trinomials, following~\cite{bravyi2024high}; Campaign~4 relaxes this restriction to 4--6-term polynomials (weight-8 to weight-12 stabilizers), including mixed monomials $x^a y^b$ ($a,b > 0$), to probe whether higher check weight unlocks better parameters.

\subsection{Perturbed bivariate bicycle codes}
\label{sec:pbb}

We introduce the perturbed bivariate bicycle (PBB) construction as a non-CSS ansatz that augments BB polynomial pairs with perturbation polynomials $(C, D)$, creating mixed stabilizers---without \emph{a priori} guarantee that the resulting code has good parameters.
A PBB code is defined by four polynomials $A, B, C, D \in R$ with stabilizer matrix
\begin{equation}
H = \begin{pmatrix} A & B & C & D \\ 0 & 0 & B^\top & A^\top \end{pmatrix},
\end{equation}
where the first block row has both X-type ($A,B$) and Z-type ($C,D$) support, creating mixed stabilizers, and the second is purely Z-type.
In the standard symplectic representation $H = (\mathbf{X} \mid \mathbf{Z})$, the first $\ell m$ columns encode X-support and the second $\ell m$ columns Z-support: the first block row maps to stabilizers with X-support from $(A,B)$ and Z-support from $(C,D)$.
All rows commute if and only if $(AC^\top + BD^\top) \bmod 2$ is symmetric over $\FF_2$ (this is the only nontrivial commutativity condition; commutativity between block~1 and block~2 is automatic from BB ring commutativity)---a constraint that, unlike the CSS case where commutativity of $R$ suffices, must be verified computationally for each tuple.
Empirically, ${\sim}10\%$ of random weight-3 4-tuples at lattice $(6,6)$ satisfy this constraint.

PBB codes reduce to CSS BB codes when $C = D = 0$; for $C, D \neq 0$, each stabilizer generator may have nontrivial support on both X and Z Pauli types.

A key validity question is whether PBB codes with $C, D \neq 0$ are genuinely non-CSS or merely CSS codes in disguise---reducible via row operations over $\FF_2$ or conjugation by single-qubit Clifford gates.
We apply three verification levels.
First, a row-operation check: if the Z-parts of blocks~1 and~2 coincide ($C = B^\top$ and $D = A^\top$), then row reduction eliminates the $(C, D)$ perturbation entirely, yielding a CSS code.
No top-$\FOM$ PBB code satisfies this condition.
Second, a Hadamard equivalence check derived independently in App.~\ref{app:lc} (a parity 2-coloring algorithm; inspired by~\cite{khesin2026mirror}'s mirror-code result but formulated and proven directly for general stabilizer codes): the supplied stabilizer generators are simultaneously brought to pure-X or pure-Z form by some single-qubit-Hadamard pattern $H_J$ if and only if no generator has Y support and the parity 2-coloring constraint graph is bipartite. We apply this at the level of the canonical block-circulant generators; the more permissive group-level question (a code's stabilizer group can decompose into pure-X and pure-Z subgroups even when the supplied generators do not) is addressed separately by the rank-condition test (Lemma~7.4 of~\cite{cross2025small}) in App.~\ref{app:lc}.
Of the 368 PBB codes, 158 have at least one stabilizer with Y support; of the remaining 210 (no Y support), 10 admit a bipartite constraint graph and are therefore Hadamard-equivalent to CSS---an explicit per-qubit $H$ pattern of weight $n/2$ renders every stabilizer pure-X or pure-Z, verified by direct construction.
Third, we computationally verify (Appendix~\ref{app:lc}) which further codes can be rendered CSS by per-qubit single-qubit Clifford conjugation, by exhaustively testing all 36 uniform per-block assignments and solving the exact $\mathrm{GF}(2)$ affine system for non-uniform $\{I,S\}$ and $\{H,HS\}$ patterns---two classes the Hadamard 2-coloring check does not cover.
Exactly one additional code passes: the $\code{36,4,6}$ ($\FOM = 4.0$) is LC-equivalent to a CSS code via uniform $S$ on both blocks.
Combining the three checks, 11 of the 368 PBB codes are CSS-equivalent under tested single-qubit Clifford patterns (10 via non-uniform Hadamard, 1 via uniform $S$); the remaining 357 are CSS-inequivalent under all tested LC patterns (Appendix~\ref{app:lc} details the residual coverage gaps).
The 10 Hadamard-CSS codes are all low-rate ($k \in \{4, 8\}$) and low-$\FOM$ (all $\leq 2.8$); none appears in Table~\ref{tab:pbb_best} or among the codes highlighted in Sec.~\ref{sec:results}.

The $A = B$ distance trap extends to the non-CSS setting (Appendix~\ref{app:ab_trap}): if $A = B$ in a PBB code, $\ell m$ weight-2 Z-type operators $(0|\mathbf{e}_i + \mathbf{e}_{i+\ell m})$ commute with all stabilizers---since the X-part $[A|A]$ has identical column blocks---placing them in the normalizer; MILP confirms $d = 2$ for all $A = B$ PBB codes in the catalog.

\subsection{Figure of merit}
\label{sec:fom}

We adopt the $\FOM = kd^2/n$, motivated by the Bravyi--Poulin--Terhal (BPT) bound $kd^2 = O(n)$~\cite{bravyi2010tradeoffs} and standard in the BB literature~\cite{bravyi2024high,liang2025generalized}.
The BPT bound applies to geometrically 2D-local stabilizer codes. Therefore, for qLDPC codes that are 2D-local, the FOM is bounded above by a constant as $n$ increases, but for general qLDPC codes this is not the case. For example, the rotated surface code has FOM = 1 whereas the gross code has FOM = 12. We use the FOM as a benchmark for comparing BB codes against each other.
The highest-$\FOM$ prior CSS BB codes achieve $\FOM = 12.0$ ($\code{144,12,12}$) and $\FOM \leq 19.2$ ($\code{360,12,{\leq}24}$, weight-6 CSS)~\cite{bravyi2024high}; the latter shares its $(n,k,d_{\rm upper})$ parameters with one of our PBB codes (Sec.~\ref{sec:results}, Table~\ref{tab:useful}), but the two are structurally distinct (Bravyi: weight-6 CSS; ours: weight-8 PBB).

\subsection{The $A = B$ distance trap}
\label{sec:ab_trap}

Certain pathological ansätze---notably $A = B$, which always yields $d = 2$ (Appendix~\ref{app:ab_trap})---serve as ``traps'' for the search.
The evolution initially converged on $A = B$ codes, which can achieve arbitrarily high~$k$ but provide no meaningful error correction.
Three pipeline mechanisms detect and avoid these traps: (i)~the $d/\sqrt{n}$ trust filter (Sec.~\ref{sec:trust}), (ii)~MILP verification revealing $d = 2$ in under one second, and (iii)~the LLM learning to diversify away from symmetric forms after receiving low-fitness feedback.
Notably, BP-OSD failed to detect $d = 2$ for $A = B$ codes even at $1.5 \times 10^6$ trials---reporting $d \leq 14$ for a $\code{144,32,2}$ code despite $\ell m = 72$ weight-2 X-type operators in the normalizer, of which $k/2 = 16$ are linearly independent nontrivial logicals (Appendix~\ref{app:ab_trap})---motivating our MILP verification pipeline.

\section{Method}
\label{sec:method}

\subsection{Evolutionary framework}
\label{sec:framework}

We employ OpenEvolve~\cite{openevolve}, an open-source framework for LLM-guided program synthesis using MAP-Elites~\cite{mouret2015illuminating}.
The evolutionary target is a \emph{generator ansatz} $G(\ell, m) \to \{(A_i, B_i)\}$ (or 4-tuples $(A_i, B_i, C_i, D_i)$ for non-CSS codes)---a program producing candidate polynomial tuples for any lattice dimensions, without guarantee that all outputs yield valid codes.
This allows a single evolved ansatz to discover codes across multiple block lengths simultaneously.

At each iteration, the LLM receives the current highest-fitness ansatz, domain knowledge about BB code algebra, and evaluation feedback, then proposes a targeted code diff---a modification to the generator ansatz rather than a full rewrite. 
Typical mutations include exponent adjustments, addition of new strategy branches, and control-flow restructuring.
The population is distributed across 4--6 islands with migration every 12--25 iterations (depending on campaign); the MAP-Elites archive indexes ansätze along two behavioral dimensions chosen to encourage both breadth (generalization across lattices) and depth (producing many high-quality codes): (i)~the number of lattices yielding codes with $k \geq 8$ and (ii)~the total count of such codes.
This prevents the search from collapsing to a single lattice-specific solution.
Campaign~4 uses different dimensions---polynomial term count and monomial structure (mixed-monomial, diagonal-mixed, or separated-variable; see Sec.~\ref{sec:families})---to steer the search toward structurally diverse ansätze.
The complete diff history is included in the public repository.

\begin{figure}[t]
\centering
\includegraphics[width=\columnwidth]{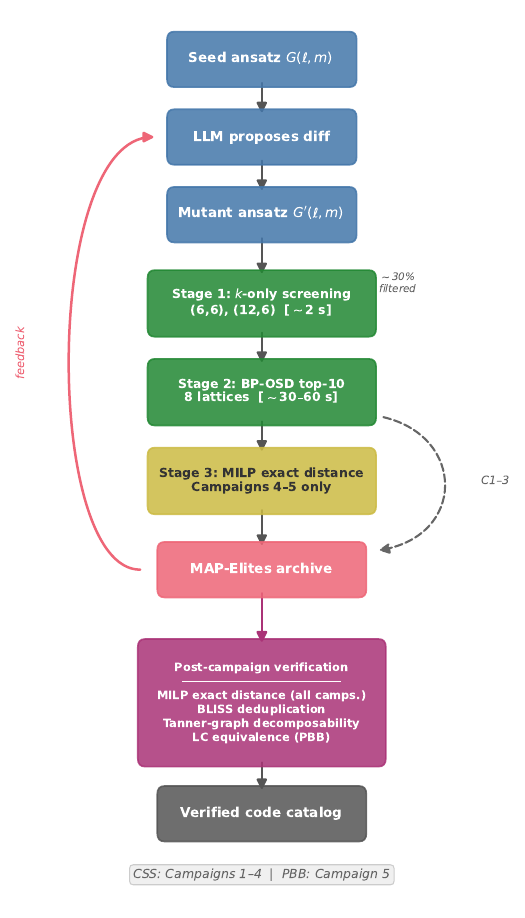}
\caption{Evolution pipeline.  An LLM proposes diffs to a seed generator ansatz
$G(\ell,m)$; each mutant is screened by a multi-stage cascade (Stage~1: $k$-only,
Stage~2: BP-OSD distance) and inserted into a MAP-Elites archive
that feeds back context to the LLM\@.  Campaigns~4--5 add in-loop MILP exact
distance (Stage~3, added in response to BP-OSD overestimation discovered in Campaigns~1--3; see Sec.~\ref{sec:bposd_findings}); Campaigns~1--3 use post-hoc MILP (dashed bypass).
Post-campaign verification (MILP distance computation $\to$ BLISS deduplication $\to$ Tanner-graph decomposability $\to$ LC equivalence for PBB) ensures all reported codes are genuine.}
\label{fig:pipeline}
\end{figure}

\subsection{Evaluation cascade}
\label{sec:cascade}

Each candidate ansatz passes through a multi-stage cascade.
\emph{Stage~1} (${\sim}2$\,s) screens the ansatz on two small lattices---$(6,6)$ and $(12,6)$---computing only~$k$ via $\FF_2$ rank; ansätze with no $k > 0$ code at both lattices are discarded (${\sim}30\%$ of mutants).
Both screening lattices have $3 \mid \ell$, which may bias the search toward divisibility-by-3 structures; the failure of all evolved ansätze at $(16,9)$ and $(18,8)$ (Sec.~\ref{sec:ablation}) may partly reflect this. 
\emph{Stage~2} evaluates surviving ansätze on 8 lattices spanning $n \in \{144, 288, 360\}$: $(12,6)$, $(6,12)$, $(12,12)$, $(24,6)$, $(15,12)$, $(30,6)$, $(16,9)$, $(18,8)$. 
For each lattice, all candidates are screened for~$k$, then the top~10 (by $k$ diversity) undergo BP-OSD with 1{,}000 OSD$_0$ trials (${\sim}30$--$60$\,s); promising candidates ($\FOM \geq 6$) receive three additional batches of 500 trials.
Campaigns~1--3 used BP-OSD for in-loop distance estimation (Stage~2); however, as we show in Sec.~\ref{sec:bposd_findings}, BP-OSD systematically overestimates distance for high-rate codes, rendering the fitness signal unreliable. This motivated the addition of MILP distance verification (\emph{Stage~3}) for Campaigns~4--5, ensuring that the evolutionary fitness reflects true code quality.
Campaign~5 uses 7 lattices ($n = 36$--$360$, all with $m \leq 6$) and the adaptive distance pipeline (Sec.~\ref{sec:noncss_pipeline}).

The cascade is designed to balance exploration speed against verification cost.
Stage~1 ($k$-only, ${\sim}2$\,s) eliminates ansätze that fail to produce any valid code, keeping the cost per rejected mutant low.
Stage~2 (BP-OSD, ${\sim}30$--$60$\,s) provides a distance estimate sufficient to rank candidates within the evolutionary loop, where exact verification would be prohibitively expensive (${\sim}10$\,min/code).
MILP exact verification is applied post-hoc (Campaigns~1--3) or in-loop for top candidates only (Campaigns~4--5), resolving the distance ambiguity that Sec.~\ref{sec:bposd_findings} shows BP-OSD cannot.
This staged approach allowed us to evaluate ${\sim}2 \times 10^5$ candidates while MILP-verifying only the ${\sim}400$ surviving codes.

\subsection{Campaign configurations}
\label{sec:campaigns}

Table~\ref{tab:campaigns} summarizes five campaigns totaling ${\sim}1{,}650$ iterations over ${\sim}140$ hours (${\sim}$US\$400 in LLM inference).
Campaigns~1--3 evolved trinomial CSS codes from the same seed---${\sim}300$ lines of Python encoding an $x/y$-swap pattern, perturbations of known codes, self-similar scaling, and small-exponent enumeration---differing in model configuration and population size.
Campaign~1 used Gemini~3 Flash Preview alone; Campaigns~2--3 used an ensemble of Claude~Opus~4.6 (Anthropic), GPT-5.2 (OpenAI), and Gemini~3 Pro Preview (Google) with equal selection weight.
Campaign~3's $10\times$ larger population (1{,}000 ansätze) enabled the discovery of 145 unique codes not found in Campaign~1, while Campaign~2 (population 100, 251 iterations) added none beyond Campaign~3's.
This is a negative result for the small-population ensemble configuration: at population~100 the 3-model ensemble produced no codes that the single-model Campaign~1 (same population, fewer iterations) had not already reached, and that Campaign~3 (same ensemble, $10\times$ population) did not subsume; the cost was still ${\sim}$US\$25.
We do not interpret this as evidence that ensemble mutation is unhelpful in general, only that at population~100 it does not pay off here.
Campaign~3 showed clear saturation: the last ${\sim}275$ of 500 iterations produced only one marginal improvement ($+1.2\%$). Fitness trajectories for Campaigns~1--3 and cross-campaign analysis are in the Supplemental Material.

Campaign~4 introduced MILP-in-the-loop verification after post-hoc analysis of Campaigns~1--3 revealed that BP-OSD distance estimates were inflated by up to $12\times$ (Sec.~\ref{sec:bposd_findings}), and extended the search to 4--6-term polynomials with mixed monomials ($x^a y^b$ terms), diagonal-mixed patterns ($x^a y^a$), and augmentations of known trinomials. 
It used an updated ensemble (Claude~Opus~4.6, GPT-5.3-Codex, Gemini~3.1 Pro Preview) with temperature $1.0$ (vs.\ $0.8$ for Campaigns~1--3 and~5) and the structural-diversity MAP-Elites features described in Sec.~\ref{sec:framework}.
We note that Campaign~4 simultaneously changed the MAP-Elites behavioral dimensions \emph{and} the polynomial-family scope (trinomial $\to$ 4--6-term mixed-monomial); no controlled ablation isolates which change is responsible for the 45 new distinct codes Campaign~4 contributed.
Campaigns~1--3 used temperature $0.8$, balancing syntactically valid code generation against mutation diversity; Campaign~4 increased to $1.0$ to encourage more radical mutations when exploring the larger space of 4--6-term polynomials.
These values were chosen based on a qualitative pilot observation (5~iterations each at temperatures $0.5, 0.8, 1.0, 1.2$, far too few to discriminate statistically): temperatures below $0.8$ produced mostly trivial edits, while $1.2$ frequently generated syntactically invalid programs.

Campaign~5 evolved non-CSS PBB 4-tuples $(A,B,C,D)$ on 7 lattices with $m \leq 6$, constrained by MILP runtime: the symplectic formulation is ${\sim}4\times$ slower per logical than the CSS formulation, making the $(12,12)$ and $(15,12)$ lattices (which produced the highest-$\FOM$ CSS codes) impractical. 
The seed generated 4-tuples using known CSS base pairs with random perturbation polynomials $C, D$.

All campaigns used \texttt{qldpc}~v1.0.1~\cite{qldpc} for code construction, \texttt{ldpc}~v2.2.0~\cite{ldpc_library} for BP-OSD, HiGHS~\cite{huangfu2018highs} via \texttt{scipy}~\cite{virtanen2020scipy} for MILP, and \texttt{python-igraph}~\cite{csardi2006igraph} with the BLISS algorithm~\cite{junttila2007engineering} for Tanner-graph canonical labeling.
Campaigns~1--3 ran on a single Apple M4 Max workstation with 36 GB of RAM; Campaigns~4--5 used a 64-core server with 251 GB RAM, Linux (RHEL 9).
LLMs were accessed via cloud APIs through a LiteLLM~\cite{litellm} proxy.

\begin{table}[t]
\centering
\caption{Evolution campaigns.  ``Ens.'' denotes a 3-model ensemble with equal selection weight; see text for model details.  ``Codes'' = BLISS-unique codes from that campaign not in prior catalogs; counts may overlap across campaigns, so the per-row sum exceeds the 97 distinct CSS codes (Sec.~\ref{sec:families}) after cross-campaign deduplication.  ``Cost'' = approximate per-campaign LLM inference spend (US dollars). The campaign rows sum to ${\sim}$US\$237; the ${\sim}$US\$400 total reported elsewhere additionally covers the ablation arms of Sec.~\ref{sec:ablation} (random search at $10^3$ and $10^4$/lattice plus the GA/GA-G arms over 5~seeds each, ${\sim}$US\$70 combined---random search uses no LLM but the comparable GA-G calls in cross-validation runs do) and exploratory/failed runs (re-tries on flaky API responses, abandoned config sweeps, ${\sim}$US\$90), neither of which is tabulated here.}
\label{tab:campaigns}
\begin{tabular}{lccccccc}
\toprule
Camp.\ & Type & Models & Iters & Pop.\ & Codes & Time & Cost \\
\midrule
1 & CSS & Flash & 100 & 100 & 9 & 5\,h & \$15 \\
2 & CSS & Ens.\ & 251 & 100 & 0$^*$ & 9.5\,h & \$25 \\
3 & CSS & Ens.\ & 500 & 1{,}000 & 145 & 21\,h & \$50 \\
4 & CSS & Ens.\ & 300 & 750 & 45 & 92\,h & \$47 \\
5 & PBB & Ens.\ & 500 & 200 & 368 & 11\,h & \$100 \\
\bottomrule
\multicolumn{8}{@{}p{\columnwidth}@{}}{\footnotesize $^*$Retroactively found, in cross-campaign deduplication, to overlap with Campaign~3's output.}
\end{tabular}
\end{table}

\section{Verification}
\label{sec:verification}

\subsection{BP-OSD and its limitations}
\label{sec:bposd_limits}

BP-OSD distance estimation~\cite{panteleev2021degenerate,roffe2020decoding,bravyi2024high} provides stochastic upper bounds on~$d$: for each trial, a random nontrivial logical coset representative is selected and the decoder seeks a low-weight operator in that coset (details in the Supplemental Material).
The minimum weight across $T$ trials is an upper bound ($d_{\text{true}} \leq d_{\text{reported}}$).
These bounds can be loose, particularly for high-rate codes ($k/n > 0.1$), where the logical operator space is large and low-weight representatives are sparse.

Bravyi et al.~\cite{bravyi2024high} used a single OSD$_0$ configuration; we extend this to a multi-decoder protocol with three configurations---OSD$_0$ with sum-product BP, OSD-CS$_{10}$ (combination-sweep OSD, order~10) with sum-product BP, and OSD-CS$_{10}$ with minimum-sum BP---each run for 10 batches of 5{,}000 trials, totaling 150{,}000 trials per code.
The motivation is that different configurations find different low-weight logicals.
The verified bound is the global minimum across all 30 batches: on the $\code{288,24,12}$ (MILP $d = 12$), OSD$_0$ never finds an operator below weight~24, while OSD-CS$_{10}$/minimum-sum finds weight~12.
All configurations use \texttt{ldpc} library~\cite{ldpc_library} defaults: $\mathrm{max\_iter} = n$ BP iterations, no damping, depolarizing error-rate placeholder $10^{-3}$.
For a systematic comparison of BP-OSD settings and alternative syndrome-decoder-based distance estimators, including the effect of DEM permutations and random stabilizer additions, see Webster et al.~\cite{webster2026distance}.

\subsection{MILP distance computation}
\label{sec:milp}

MILP is one of several exact methods for quantum code distance, alongside Brouwer--Zimmermann enumeration and SAT solvers; Webster et al.~\cite{webster2026distance} provide a comprehensive comparison.
Computing the minimum distance of a CSS code is NP-hard in general~\cite{kapshikar2022hardness}; we use MILP because it provides both upper bounds (incumbents) and optimality certificates (gap~$= 0$) with tunable timeout, and the low-weight check matrices of BB codes ($w = 6$--$12$) produce constraint matrices amenable to branch-and-bound solvers at $n \leq 360$.
Following~\cite{bravyi2024high}, we run mixed-integer linear programming on all 97 distinct CSS codes and all 368 non-CSS codes; distances are reported as exact only when all relevant logical MILP instances reach proven optimality, and otherwise as valid upper bounds.
The CSS formulation encodes the minimum-weight logical operator problem as a binary integer program: for each logical qubit, minimize $\sum_j x_j$ subject to mod-2 commutation and anticommutation constraints with slack variables.
Since a CSS code with $k$ logical qubits has $k$ independent X-type and $k$ independent Z-type logicals, it suffices to solve $2k$ MILP instances (one per logical generator); the code distance is the minimum weight across all $2k$ solutions.
We use HiGHS~\cite{huangfu2018highs} via \texttt{scipy}~\cite{virtanen2020scipy} with 120--600\,s per-logical timeouts (3{,}000\,s for Campaign~4 top codes, up to 14{,}400\,s for Campaign~5 deep verification).

The formulation was validated on the $\code{72,12,6}$ and $\code{144,12,12}$, both confirmed exactly (MIP gap~$= 0$ for every logical).
A dedicated optimality audit confirmed $d = 12$ exactly for the $\code{288,24,12}$ (48 logicals, 29\,min), the $\code{288,16,12}$ (32 logicals, 80\,min), and the gross code (24 logicals, 13\,min).
For the $\code{360,16,{\leq}14}$, the solver found weight-14 solutions for all 32 logicals and proved optimality for~7; since no logical returned a weight below~14, $d \leq 14$ is established.
For codes with $d = 2$ (including all $A = B$ codes), weight-2 solutions are found in under one second, establishing $d = 2$ exactly in combination with the lower bound from column weights.

The complete MILP formulations for both CSS (Hamming weight) and non-CSS (symplectic weight with the standard linear encoding of the per-qubit binary OR support) distance problems, along with validation details, are given in the Supplemental Material.
HiGHS reports either a proven-optimal solution (MIP gap~$= 0$, establishing $d$ exactly) or an incumbent solution with a lower bound (providing a valid upper bound $d \leq d_{\text{incumbent}}$ and a lower bound $d \geq d_{\text{LB}}$).
We report codes as ``exact'' only when all $2k$ logicals achieve MIP gap~$= 0$.

\subsection{Adaptive pipeline for non-CSS codes}
\label{sec:noncss_pipeline}

Computing distance for non-CSS codes is harder than for CSS codes because the minimum-weight logical operator has \emph{symplectic} weight (support on X, Z, or both), and the BP-OSD distance-bounding technique---which relies on decoding random syndromes---fails for non-CSS codes (Tier~3 below) because random binary vectors typically do not lie in $\mathrm{im}(H)$.
We developed a 3-tier adaptive pipeline:

\emph{Tier~1: Exhaustive weight-$w$ enumeration} ($d \leq 6$ at $n \leq 216$; $d \leq 4$ at $n > 216$ due to ${\sim}89$\,GB memory at $n = 360$).
For small codes, we enumerate all Pauli operators up to weight~$w$ by constructing a lookup table (syndrome column dictionaries and XOR lookups) mapping syndromes to operators, and checking for nontrivial logicals.
If a weight-$w$ logical is found, $d = w$ exactly.

\emph{Tier~2: MILP symplectic formulation.}
For larger codes, we solve an MILP minimizing symplectic weight (full formulation in the Supplemental Material): for each of $2k$ logical operators, minimize the combined X and Z support subject to commutation/anticommutation constraints.
Adaptive timeouts: 15\,s per logical at $n \leq 108$, 30\,s at $n \leq 216$, 60\,s at $n > 216$.
Partial results are valid upper bounds.

\emph{Tier~3: BP-OSD with achievable-syndrome sampling.}
The BP-OSD distance-bounding technique of~\cite{bravyi2024high} works by forming an effective matrix $H_{\mathrm{eff}} = \bigl(\begin{smallmatrix} H_{\mathrm{check}} \\ L \end{smallmatrix}\bigr)$ (check rows stacked above logical rows), setting the stabilizer syndrome to zero, choosing a random nontrivial logical coset~$\lambda$, decoding the syndrome $(0, \lambda)$, and recording the decoded weight.
For CSS codes, X and Z errors are decoded separately: the Z-distance, for instance, uses $H_X$ as check matrix and $L_X$ as logicals.
Although $H_X$ has $k/2$ redundant rows, this is harmless---the stabilizer syndrome is fixed at zero and the $k/2$ Z-logical operators span the full logical space within $\ker(H_X)$, so every nontrivial logical coset is achievable by a pure-Z error.
For non-CSS codes, the channels are no longer cleanly separated: when the technique is applied per channel (e.g., pure-X or pure-Z errors), the achievable logical cosets form a strict subspace of the full logical space, because some logicals inherently require mixed X/Z (Y-type) support.
We characterize ``achievable'' empirically per channel and code: for a given channel (the support pattern of allowed errors, e.g.\ X-only) and a given code, the achievable subspace is the image of the channel error space under the $H_{\mathrm{eff}}$ map, computed as a GF(2) null-space projection on the stabilizer matrix restricted to the channel.
Empirically, on the $\code{144,12,12}$ PBB code at $k = 12$, the Z-channel's achievable subspace has dimension $k$ within the $2k$-dimensional logical space, so $\sim$half of randomly sampled cosets have no preimage; the same holds for the Y-channel (we have not proved this is general).
Our fix is to compute the achievable subspace per channel and sample only from it, which restored BP-OSD from $\sim 0\%$ to $\sim 100\%$ decode success on every PBB code we tested.

The final reported distance is the upper bound $d \leq \min(d_{\text{enum}}, d_{\text{MILP}}, d_{\text{BP-OSD}})$.
Each code is assigned \emph{exact} (all $2k$ logicals solved to optimality) or \emph{trusted} (valid upper bound from $\geq 2$ independent methods).
To tighten distance bounds, 149 PBB catalog entries underwent deep MILP (timeouts up to 14{,}400\,s per logical, up to a 60-worker pool), producing 63 exact outcomes and identifying 33 downward distance corrections (22\% correction rate), with the largest being $d = 24 \to 16$ at $n = 360$.

\subsection{Trust boundaries}
\label{sec:trust}

During CSS Campaigns~1--4, a trust filter on $d/\sqrt{n}$ prevents low-distance high-$k$ codes from dominating fitness: $d/\sqrt{n} \leq 1.3$ is fully trusted, $\geq 2.0$ is discarded, with linear interpolation between.
This heuristic is imperfect---the $\code{360,40,2}$ passed with $d_{\text{BP}}/\sqrt{n} = 1.26$ despite true $d/\sqrt{n} = 0.11$---exposing a fundamental limitation: the filter operates on BP-OSD estimates that are themselves unreliable for exactly the codes it should detect.
Campaign~5 bypasses this filter entirely, using verified distances from the adaptive pipeline with codes having $d \leq 4$ rejected.

\section{Results}
\label{sec:results}

\subsection{Structural families}
\label{sec:families}

The five campaigns produced 225 unique CSS polynomial representations.
To determine the number of genuinely distinct codes, we performed permutation-equivalence analysis via BLISS canonical labeling~\cite{junttila2007engineering} of the colored Tanner graph.
We construct a colored bipartite graph for each code: qubit vertices (color~0), X-check vertices (color~1), and Z-check vertices (color~2), with edges from each check to its supporting qubits.
The BLISS algorithm computes a canonical form for this colored graph; while graph isomorphism is quasi-polynomial in the worst case~\cite{babai2016graph}, BLISS is efficient in practice for the structured Tanner graphs arising here. Two codes are declared \emph{permutation-equivalent} (related by a qubit relabeling that preserves X- and Z-stabilizer roles, with X-checks and Z-checks permitted to permute among themselves; for non-CSS codes the analogous decoupling of stabilizer X-supports and Z-supports under the 3-coloring described below) if and only if their colored-Tanner-graph canonical forms are identical.
This captures all automorphisms of the Tanner graph, including lattice symmetries invisible to polynomial-level comparisons (e.g., re-indexing $x \mapsto x^a$ when $\gcd(a, \ell) = 1$).
The 225 representations collapse to 97 distinct codes (2.3:1 redundancy, up to 14:1 for univariate codes); Campaign~4 also found additional representations of two Bravyi et al.\ codes, for 99 total equivalence classes.
An independent deduplication of Campaign~5's 720 tuple-distinct PBB codes---using the same colored Tanner graph approach, generalized to non-CSS stabilizers by representing each stabilizer with two separately-colored check vertices (one carrying its X-support edges, one carrying its Z-support edges, so the graph has three vertex colors total: qubits, stabilizer-X-supports, stabilizer-Z-supports), with a tying edge between each stabilizer's X-support and Z-support vertices to forbid independent permutations of the two check-vertex classes---yields 368 BLISS-unique non-CSS codes (1.96:1 ratio, increasing with lattice size). The complete CSS catalog (225 representations, 97 distinct codes) and non-CSS PBB catalog (368 codes) appear in the Supplemental Material; all data are also available at~\cite{cruzbenito2026qcode}.

Throughout this paper, codes are called \emph{distinct} if their colored Tanner graphs (defined above; vertex coloring distinguishes qubits from check-vertices, with the CSS variant separating X-checks from Z-checks and the non-CSS variant using per-stabilizer X- and Z-support vertices) have non-isomorphic BLISS canonical forms---a relation that is sound and complete for permutation equivalence under the respective coloring.
We do not attempt to decide broader notions of code equivalence: the classical code-equivalence problem has no known polynomial-time algorithm~\cite{petrank1997code} and lies in roughly the same complexity class as graph isomorphism (currently bounded by quasi-polynomial time~\cite{babai2016graph}); it is closely related to the security analysis of McEliece-style cryptosystems on hidden-structure code families~\cite{mceliece1978public,sendrier2000finding}.
The analogous stabilizer-code problem under broader transformations (local Clifford, full Clifford) likewise has no known efficient algorithm in general.
Related tooling uses the same graph-isomorphism primitive for adjacent questions: Sage's \texttt{sage.coding.codecan}~\cite{feulner2009canonical} computes canonical forms of classical linear codes via partition refinement, and \texttt{autqec}~\cite{sayginel2024fault} maps a stabilizer code to a derived binary linear code and runs BLISS to compute its \emph{automorphism group} (whence logical Clifford gates)---a complementary application of the same machinery to a different question (self-equivalence of a single code, rather than equivalence between two codes).
No canonical open-source solver for full stabilizer-code equivalence currently exists.
We therefore regard our 465-code ``distinct'' count as a \emph{conservative upper bound}---under any broader notion of code equivalence (coarser than colored-Tanner-graph isomorphism, e.g.\ local Clifford or full Clifford), the count cannot exceed 465 and may be smaller.
For non-CSS PBB codes we separately analyze the local-Clifford-equivalence layer in Sec.~\ref{sec:pbb} and App.~\ref{app:lc}, which collapses 11 of the 368 codes onto known CSS counterparts and leaves 357 CSS-inequivalent within the tested LC families.

Deduplication operates at two levels.
\emph{Within the evolution loop:} polynomial-level deduplication rejects candidates whose $(A, B)$ pair (or $(A, B, C, D)$ tuple for PBB) has already been evaluated, using a hash set of polynomial coefficient tuples for $O(1)$ lookup.
This is fast but does not catch codes that are equivalent under lattice symmetries (e.g., $x \mapsto x^a$).
\emph{Post-campaign:} BLISS canonical labeling of the colored Tanner graph (described above) identifies equivalences that polynomial-level deduplication cannot detect.

Four structural families were identified by post-hoc analysis of the 225 evolved polynomial representations.
Each has a distinct rate--distance profile.

\paragraph{Univariate codes (HGP of low-distance cyclic codes).}
The dominant motif has $A = f(y)$ and $B = g(x)$ (separated variables), equivalent to the hypergraph product of two cyclic codes~\cite{eberhardt2024pruning,aydin2025cxc}.
Of the 154 Campaigns~1--3 representations, 87 (57\%) are univariate, collapsing to just 12 distinct codes---the most redundant family.
These achieve the highest encoding dimensions: $k = 40$ at $n = 360$ (e.g., $A = 1{+}y{+}y^2$, $B = 1{+}x^5{+}x^{10}$), following the formula $k = 8\ell/3$.
MILP reveals a universal distance collapse: every univariate code in the catalog has $d \in \{2, 4\}$.
By the Tillich--Z\'emor formula $d = \min(d_1, d_2, d_1^\top, d_2^\top)$~\cite{tillich2014quantum}, the BB code distance is bounded by the smallest minimum distance among the four classical cyclic codes $\ker H_A$, $\ker H_B$, $\ker H_A^\top$, $\ker H_B^\top$ (where $d^\top$ denotes the distance of $\ker H^\top$; for palindromic check polynomials as here, $H^\top = H$ so $d^\top = d$ and the four-way min collapses to $\min(d_A, d_B)$).
For the $A = 1{+}y{+}y^2$, $B = 1{+}x^c{+}x^{2c}$ subfamily with $c = \ell/3$, the quotient $(x^\ell{-}1)/B = 1{+}x^c$ has weight~2, so $\ker H_B$ has distance~2 and $d_{\text{BB}} = 2$; for the remaining univariate codes in the catalog (where neither quotient has weight~2), the smallest of the four classical distances is~4, verified case-by-case via MILP.
Moreover, any weight-3 (odd-weight) check polynomial $f$ with $f(1) = 1$ produces a cyclic code in which all codewords have even Hamming weight (since the generator $g = (y^m{-}1)/\gcd(f,\,y^m{-}1)$ satisfies $g(1) = 0$), so all four component distances are even and the catalog skips $d = 3$ entirely---a regularity MILP confirms exhaustively across all lattice sizes.
Replacing the trinomial $1{+}y{+}y^2$ with the weight-2 binomial $1{+}y$ recovers the toric code with linearly growing distance, but such polynomials lie outside the trinomial search space of Campaigns~1--3.

The univariate subspace is small enough for exhaustive enumeration once recognized: at $(15,12)$ there are $\binom{11}{2} \times \binom{14}{2} = 5{,}005$ pairs, versus ${\sim}9 \times 10^{11}$ unrestricted trinomial pairs.
That the LLM independently converged on this well-characterized subspace---without any HGP-related prompt context---illustrates how the search naturally gravitates toward the highest-rate region.
The LLM's search trajectory provides structural insight: the dominant evolutionary path led to univariate/HGP codes, suggesting these are, in some sense, the ``easiest'' high-$k$ solutions; the $x/y$-swap codes that achieve $d \geq 12$ required different starting ansätze.
The complete univariate family table with all polynomial pairs, lattice embeddings, and BP-OSD vs.\ MILP distance comparisons appears in the Supplemental Material.

For the subfamily $A = 1{+}y{+}y^2$, $B = 1{+}x^c{+}x^{2c}$ with $c = \ell/3$, the BB code is a hypergraph product code $\mathrm{HGP}(H_B, H_A^\top)$ encoding $k = 8\ell/3$ logical qubits (independent of~$m$) for any $3 \mid \ell$, $3 \mid m$ (Theorem~\ref{lem:crt_k} in Appendix~\ref{app:crt}).
The result follows from the HGP dimension formula $k = k_1 k_2 + k_1^\top k_2^\top$~\cite{tillich2014quantum} and the polynomial divisibility $A(y) \mid y^m{-}1$, $B(x) \mid x^\ell{-}1$.
We have verified computationally that the formula holds for all 1{,}680 parameter combinations with $3 \mid \ell$, $3 \mid m$, $\ell m \leq 250$.

\paragraph{$x/y$-swap codes.}
Codes with $A = x^a {+} y^b {+} y^c$ and $B = y^f {+} x^g {+} x^h$ (mixed $x$ and $y$ in each polynomial) are the only trinomial family achieving $d \geq 6$.
The $\code{288,24,12}$ at lattice $(12,12)$---confirmed at $d = 12$ by MILP with proven optimality for all 48 logicals---achieves $\FOM = 12.0$.
However, its polynomials $A = x^6{+}y{+}y^2$, $B = y^3{+}x^2{+}x^4$ use only even $x$-exponents, and Tanner graph connectivity analysis reveals that the code \emph{decomposes}: the full stabilizer graph ($H_X$ and $H_Z$ combined) has two disconnected components, one on even $x$-indices and one on odd.
The disconnection of the Tanner graph suffices to establish a direct sum: if a logical operator is supported on both components, its restriction to a single component also commutes with all stabilizers (since stabilizers on the other component act trivially on this component's qubits), so the restriction is itself a logical operator and the code decomposes.

Under the re-indexing $x \to x/2$ (mapping the even-$x$ sublattice of $\ZZ_{12}$ to $\ZZ_6$), each component sits at $(\ell',m') = (6,12)$, related to the gross code's $(\ell,m) = (12,6)$ by the $x \leftrightarrow y$ swap (which preserves $k$ and $d$); the resulting check matrices are identical to those of the gross code up to that swap.
Thus $\code{288,24,12} \cong \code{144,12,12} \oplus \code{144,12,12}$: the code encodes 24 qubits but offers no error-correction advantage over two independent gross codes.
The per-qubit error rates $p_L$ confirm this, matching the gross code at every tested rate (Supplemental Table~XVI).
Algebraically, the polynomials use only even powers of $x$ in $\ZZ_{24}$, so all stabilizers act within two independent cosets (even and odd $x$-indices), each isomorphic to $\ZZ_{12}$.
This code also appears at $(24,6)$ with the same polynomials; BLISS confirms these are the \emph{same} direct sum (isomorphic Tanner graphs).
The decomposition demonstrates a key pipeline capability: the Tanner-graph connectivity analysis in our verification pipeline detected that this code offers no error-correction advantage over two independent gross codes---a degeneracy invisible to BP-OSD, which cannot detect decomposability.

The $\code{288,16,12}$ ($A = x^3{+}y{+}y^2$, $B = y^3{+}x{+}x^2$, all polynomial shifts $\leq 3$) has the shortest coupling range among MILP-verified $d = 12$ codes (maximum shift~3 out of $\ell = 12$, i.e., 25\% of the cyclic dimension) and, with the decomposition of the $\code{288,24,12}$, is the highest-$k$ indecomposable CSS code at $d = 12$.
At $n = 360$, the $\code{360,16,{\leq}14}$ achieves $\FOM \leq 8.7$ and is one of only two codes where BP-OSD matched the MILP bound.
Only three distinct $x/y$-swap codes reach $d \geq 12$, all with $k \leq 24$, but the $k = 24$ endpoint is decomposable; the indecomposable bound is $k \leq 16$.

\paragraph{Mixed-monomial and higher-weight codes.} 
Campaign~4 relaxed the trinomial restriction, discovering 45 additional distinct codes involving $x^a y^b$ terms ($a,b > 0$) not present in prior work.
These include both weight-6 codes (trinomials with mixed-monomial terms, e.g., the $\code{144,8,12}$) and higher-weight codes with 4--6-term polynomials (weight-8 to weight-12 stabilizers).
Post-hoc analysis of the verified outputs identifies three subfamilies: \emph{diagonal-mixed} ($x^a y^a$ terms, weight-6), \emph{mixed-monomial augmentation} (adding a mixed term to a known trinomial, weight-8), and \emph{multi-term pure} (4-term polynomials with separated variables, weight-8). 
The highest-$\FOM$ higher-weight codes are the $\code{288,50,8}$ (weight-8, MILP exact $d = 8$, $\FOM = 11.1$) and the $\code{360,20,{\leq}14}$ (weight-8, $\FOM \leq 10.9$, encoding 4 more qubits than the highest-$\FOM$ weight-6 code at $n = 360$). 
Campaign~4 also independently rediscovered the gross code through six mixed-monomial polynomial representations---BLISS-equivalent to the trinomial form---showing that the LLM can find structurally diverse paths to known good codes.

The highest-$k$ Campaign~4 codes ($k = 50$--$54$) share a factored-product structure: $A$ factors as a product of two binomials, e.g., $A = (1{+}x^3)(1{+}y^3) = 1{+}x^3{+}y^3{+}x^3 y^3$ for the $\code{144,54,4}$.
This generalizes the univariate structure to mixed monomials while inheriting its distance ceiling ($d \leq 4$).
A notable exception is the $\code{288,50,8}$, whose \emph{cross-factored} structure---$A = (1{+}x)(1{+}y^5)$, $B = (1{+}y)(1{+}x^5)$, with variable roles swapped between $A$ and $B$---achieves $d = 8$, the highest distance among factored-product codes.
Across the full catalog, same-variable factored codes universally have $d \leq 4$, while cross-factoring provides distance protection analogous to the $x/y$-swap mechanism in trinomials.

\paragraph{Non-CSS PBB codes.}
Motivated by the CSS rate--distance tradeoff, Campaign~5 tested whether the non-CSS PBB ansatz could access superior parameters by coupling X and Z stabilizers via perturbation polynomials $(C, D)$. 
The campaign evaluated 18{,}588 candidate 4-tuples across 7 lattices, yielding 368 distinct non-CSS codes spanning 78 $(n,k,d)$ parameter sets: 251 exact rows and 117 upper-bound rows (110 TRUSTED and 7 PARTIAL with $d/\sqrt{n} \geq 1.5$).
Table~\ref{tab:pbb_best} summarizes selected representative codes per lattice; full per-lattice rankings are in the Supplemental Material.

\begin{table*}[t]
\centering
\caption{Representative non-CSS PBB codes per lattice (Campaign~5), one per lattice except $(30,6)$ where both the top-FOM trusted code and an additional simulated high-distance code are shown; the full catalog and per-lattice rankings are in the Supplemental Material.
$d_{\text{MILP}}$: MILP distance (\textbf{bold} = proven exact).
``Wt'': max stabilizer weight (type-1 mixed; type-2 pure-Z have weight~6).
Within the tested LC families, eleven of the 368 PBB codes are CSS-equivalent under single-qubit Cliffords (Sec.~\ref{sec:pbb}, App.~\ref{app:lc}): 10 via non-uniform per-qubit Hadamard and 1 via uniform $S$ on both blocks (the $\code{36,4,6}$ at $(6,3)$ with $C = y{+}x^3 y$, $D = xy{+}x^4 y$, which is neither of the two $\code{36,4,6}$ representatives shown here).
Complete catalog (368 codes) in the Supplemental Material.}
\label{tab:pbb_best}
\scriptsize
\begin{tabular}{lccccccccc}
\toprule
Code & $(\ell,m)$ & $A(x,y)$ & $B(x,y)$ & $C(x,y)$ & $D(x,y)$ & $d_{\text{MILP}}$ & $\FOM$ & Wt & Codes \\
\midrule
$\code{360,12,{\leq}24}$ & $(30,6)$ & $xy^2{+}x^4 y^3{+}x^4 y^4$ & $1{+}xy^5{+}x^5 y^4$ & $x{+}x^4$ & $xy^2{+}x^4 y^2$ & $\leq 24$ & $\leq 19.2$ & 10 & 50 \\
$\code{360,12,{\leq}20}$ & $(30,6)$ & $x^{13}y^4{+}x^{22}y^2{+}x^{22}y^3$ & $1{+}x^5 y^2{+}x^{13}y$ & $x^{13}y^3{+}x^{22}y^3$ & $x^{13}y{+}x^{22}y$ & $\leq 20$ & $\leq 13.3$ & 8 & --- \\
$\code{144,12,12}$ & $(12,6)$ & $y{+}y^2{+}x^3$ & $y^3{+}x{+}x^2$ & $y{+}x^3 y$ & $y^3{+}x^3 y^3$ & \textbf{12} & 12.0 & 8 & 66 \\
$\code{180,6,{\leq}20}$ & $(15,6)$ & $x^8 y{+}x^{12}y^4{+}xy^4$ & $1{+}x^5 y{+}xy^5$ & $x^8 y{+}x^{12}y^4$ & $x^8 y^2{+}x^{12}y^5$ & $\leq 20$ & $\leq 13.3$ & 8 & 64 \\
$\code{72,12,6}$ & $(6,6)$ & $xy^2{+}x^4 y^3{+}x^4 y^4$ & $1{+}xy^5{+}x^5 y^4$ & $xy^3{+}x^4 y^3$ & $xy^5{+}x^4 y^5$ & \textbf{6} & 6.0 & 8 & 34 \\
$\code{108,8,10}$ & $(9,6)$ & $y{+}y^2{+}x^3$ & $y^3{+}x{+}x^2$ & $x^6 y^4$ & $x^6$ & \textbf{10} & 7.4 & 8 & 107 \\
$\code{36,4,6}$ & $(3,6)$ & $xy^3{+}xy^5{+}x^2$ & $xy^2{+}x^2 y^2{+}x^2 y^3$ & $y^4$ & $y$ & \textbf{6} & 4.0 & 8 & 2 \\
$\code{36,4,6}$ & $(6,3)$ & $y{+}x^3 y^2{+}x^4 y$ & $x^3 y{+}x^4{+}x^4 y$ & $y$ & $xy$ & \textbf{6} & 4.0 & 7 & 45 \\
\midrule
\multicolumn{9}{l}{Total} & 368 \\
\bottomrule
\end{tabular}
\end{table*}

All 368 codes use trinomial bases ($|A| = |B| = 3$); the non-CSS character comes entirely from the perturbation polynomials $C$ and $D$.
The optimal perturbation size is $|C| = |D| = 2$ (56\% of codes, highest average $\FOM = 6.4$); heavy perturbations ($|C| + |D| \geq 6$) cap distance at $d \leq 8$.
A clear regularity appears in the exponent structure: the $x$-exponents of $C$ are drawn from those of $A$ in all top-10 $\FOM$ codes and 55\% of the catalog, suggesting the perturbation works as a ``$y$-rotation'' of the base's $x$-skeleton.

This regularity raises the question of whether these codes are locally Clifford equivalent to CSS codes (e.g., related by per-qubit $S$ gates); computational verification (Sec.~\ref{sec:pbb}, Appendix~\ref{app:lc}) identifies 11 of the 368 codes as CSS-equivalent under tested single-qubit Clifford patterns---10 via non-uniform per-qubit Hadamard and 1 (the $\code{36,4,6}$, $\FOM = 4.0$) via uniform $S$ on both blocks---while the remaining 357 are CSS-inequivalent under all tested LC patterns (see Appendix~\ref{app:lc} for the residual coverage gaps).
The structural contrast with CSS codes is pronounced: 86\% of PBB codes use mixed-monomial bases (some term in $A$ or $B$ of the form $x^a y^b$ with $a,b > 0$), while only 14\% use $x/y$-swap (all $A$, $B$ terms univariate)---the LLM converged on a qualitatively different region of polynomial space when optimizing for non-CSS codes.

Twenty-five codes match or exceed the gross code's $\FOM = 12.0$: fourteen $\code{144,12,12}$ ($\FOM = 12.0$; 7 exact, 7 trusted upper bounds), four TRUSTED upper bounds at $n = 360$ or $n = 180$---$\code{360,12,{\leq}24}$ ($\FOM \leq 19.2$), $\code{360,10,{\leq}22}$ ($\FOM \leq 13.4$), $\code{360,12,{\leq}20}$ ($\FOM \leq 13.3$), $\code{180,6,{\leq}20}$ ($\FOM \leq 13.3$)---and seven PARTIAL upper bounds with $d/\sqrt{n} \in [1.57, 2.11]$: six at $\code{360,10,d}$ for $d \in \{30, 32, 40\}$ ($\FOM \leq 25$, $\leq 28.4$, $\leq 44.4$) plus $\code{180,6,{\leq}21}$ ($\FOM \leq 14.7$).
The $\code{360,60,6}$ (exact) encodes 60 logical qubits ($k/n = 1/6$), the highest encoding rate among all non-CSS codes with $d \geq 6$.
Among the 14 distinct $\code{144,12,12}$ codes, 13 use mixed-monomial bases, showing that the gross code's parameters can be reached through diverse non-CSS algebraic routes.

\paragraph{Summary.}
Table~\ref{tab:codes} presents selected CSS codes across all four campaigns, sorted by $\FOM$, with MILP distances and BP-OSD overestimation ratios. 
The BP-OSD overestimation pattern is clearly visible: codes with $d \geq 12$ show modest overestimates ($1$--$1.5\times$), while high-$k$ codes with $d \leq 4$ are overestimated by $3$--$10\times$---with the $\code{360,40,2}$ reaching $10\times$.

\begin{table*}[t]
\centering
\caption{Selected CSS codes sorted by $\FOM_{\text{MILP}}$ descending.
$d_{\text{MILP}}$: MILP distance (\textbf{bold} = proven exact); $d_{\text{BP}}$: 150k-trial BP-OSD bound (--- = not run).
``Ratio'' = $d_{\text{BP}}/d_{\text{MILP}}$ (\textbf{bold} $\geq 3\times$).
Pat.: UV = univariate, XY = $x/y$-swap, SD = $A = B$ (identical polynomials; not the standard $C \subseteq C^\perp$ self-duality), MX = mixed-monomial (Campaign~4).
Top block: $d \geq 6$; bottom: high-$k$ with $d \leq 4$.
\textsuperscript{\dag}7/32 logicals proven at weight~14.
$^{\S}$Direct sum of two gross codes (Sec.~\ref{sec:families}).
$^{\P}$Weight-8 stabilizers (4-term polynomials).
\textsuperscript{\ddag}$d_{\text{BP}} = 20$ here is from the second 150k-trial run; an extended $1.5\times10^6$-trial run gave $d_{\text{BP}} \leq 24$ ($12\times$ overestimate, Sec.~\ref{sec:bposd_findings} and Supplemental).}
\label{tab:codes}
\small
\begin{tabular}{lcccccccc}
\toprule
Code & $(\ell,m)$ & $A(x,y)$ & $B(x,y)$ & $d_{\text{MILP}}$ & $\FOM$ & $d_{\text{BP}}$ & Ratio & Pat.\ \\
\midrule
$\code{288,24,12}$$^{\S}$ & $(12,12)$ & $x^6{+}y{+}y^2$ & $y^3{+}x^2{+}x^4$ & \textbf{12} & 12.0$^{\S}$ & 18 & $1.5\times$ & XY \\
$\code{288,50,8}$$^{\P}$ & $(18,8)$ & $1{+}y^5{+}x{+}xy^5$ & $1{+}y{+}x^5{+}x^5 y$ & \textbf{8} & 11.1 & 34 & $\mathbf{4.2\times}$ & MX \\
$\code{360,20,{\leq}14}$$^{\P}$ & $(30,6)$ & $1{+}x{+}x^3 y^4{+}x^6 y^2$ & $1{+}y^2{+}x^3 y^2{+}x^6 y^4$ & $\leq 14$ & $\leq 10.9$ & 14 & $1\times$ & MX \\
$\code{360,16,{\leq}14}$ & $(15,12)$ & $y^2{+}y^4{+}x^3$ & $y^6{+}x^2{+}x^4$ & $\leq 14^\dagger$ & ${\leq}8.7$ & 14 & $1\times$ & XY \\
$\code{288,16,12}$ & $(12,12)$ & $x^3{+}y{+}y^2$ & $y^3{+}x{+}x^2$ & \textbf{12} & 8.0 & 12 & $1\times$ & XY \\
$\code{144,8,12}$ & $(12,6)$ & $1{+}xy^2{+}xy^3$ & $1{+}x^2 y^3{+}x^3 y^2$ & \textbf{12} & 8.0 & 12 & $1\times$ & MX \\
$\code{144,24,6}$ & $(12,6)$ & $x^6{+}y{+}y^2$ & $y^3{+}x^2{+}x^4$ & \textbf{6} & 6.0 & 12 & $2\times$ & XY \\
$\code{288,32,6}$ & $(12,12)$ & $x^3{+}y^2{+}y^{10}$ & $y^6{+}x{+}x^{11}$ & \textbf{6} & 4.0 & 20 & $\mathbf{3.3\times}$ & XY \\
\midrule
$\code{144,54,4}$$^{\P}$ & $(12,6)$ & $1{+}y^3{+}x^3{+}x^3 y^3$ & $1{+}y^3{+}x^3{+}x^9 y^3$ & \textbf{4} & 6.0 & 24 & $\mathbf{6\times}$ & MX \\
$\code{360,40,4}$ & $(30,6)$ & $1{+}y{+}y^2$ & $1{+}x^5{+}x^{25}$ & \textbf{4} & 1.8 & 24 & $\mathbf{6\times}$ & UV \\
$\code{144,32,2}$ & $(12,6)$ & $x^4{+}1{+}y^2$ & $y^2{+}1{+}x^4$ & \textbf{2} & 0.9 & 14 & $\mathbf{7\times}$ & SD \\
$\code{360,40,2}$\textsuperscript{\ddag} & $(15,12)$ & $1{+}y{+}y^2$ & $1{+}x^5{+}x^{10}$ & \textbf{2} & 0.4 & 20 & $\mathbf{10\times}$ & UV \\
\bottomrule
\end{tabular}
\end{table*}

\subsection{Rate--distance tradeoff}
\label{sec:tradeoff}

\begin{figure*}[t]
\centering
\includegraphics[width=\textwidth]{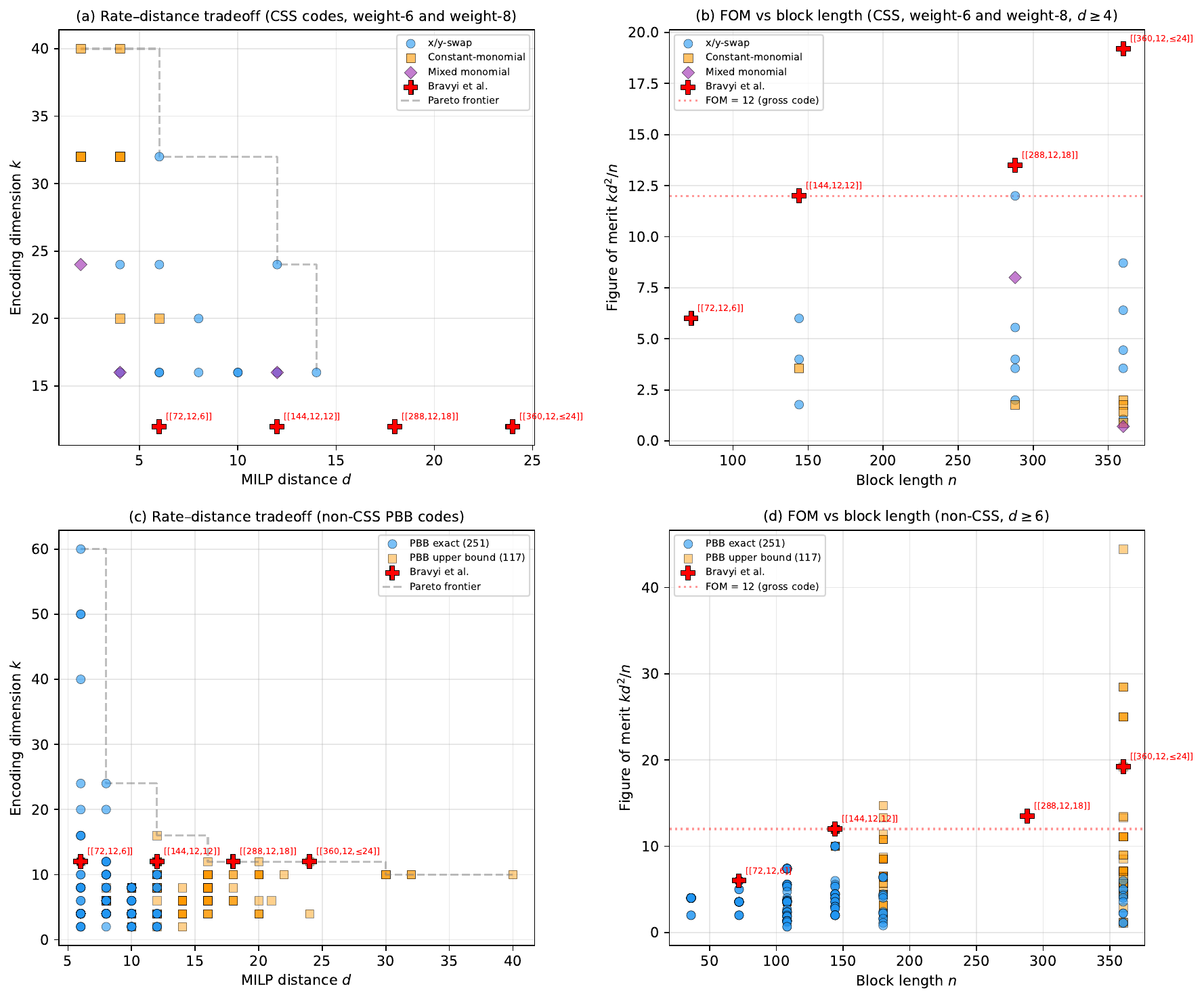}
\caption{Rate--distance landscape across all five campaigns.
\emph{Left panels:} $k$ vs.\ $d$ with Pareto frontier (dashed).
\emph{Right panels:} $\FOM = kd^2/n$ vs.\ block length.
\emph{Top:} 97 CSS codes (weight-6 and weight-8); \emph{bottom:} 368 non-CSS PBB codes.
Red markers: Bravyi et al.~\cite{bravyi2024high} baselines; dotted line: $\FOM = 12$ (gross code).
The $\code{288,24,12}$ ($d = 12$, $k = 24$) is a direct sum of two gross codes (Sec.~\ref{sec:families}); the highest-$k$ indecomposable code at $d = 12$ is $\code{288,16,12}$ ($k = 16$).
Among weight-6 codes, $k > 24$ implies $d \leq 4$; weight-8 codes reach $k = 50$ at $d = 8$; non-CSS codes occupy the same envelope.}
\label{fig:pareto}
\end{figure*}

Figure~\ref{fig:pareto} shows an empirical pattern across the searched catalog: a sharp tradeoff between encoding rate and distance that persists across all structural families, polynomial term counts, and the CSS/non-CSS boundary.
Whether this tradeoff is a structural constraint of the BB construction or an artifact of our (incomplete) search remains open---Sec.~\ref{sec:tradeoff} returns to this caveat.

The envelope is defined by two extremes.
Univariate codes achieve $k$ growing linearly with~$\ell$ ($k = 8\ell/3$) but $d \leq 4$, because the Tillich--Z\'emor min picks up a low-weight quotient $(x^\ell{-}1)/B$ or $(y^m{-}1)/A$ on at least one side (Sec.~\ref{sec:families}).
At the other extreme, the Bravyi et al.\ $\code{288,12,18}$ ($\FOM = 13.5$) and $\code{360,12,{\leq}24}$ ($\FOM \leq 19.2$) achieve high distance with $k = 12$.
Our codes fill the intermediate regime: the weight-6 $\code{288,16,12}$ achieves $k = 16$ at $d = 12$ (the indecomposable weight-6 maximum); the weight-8 $\code{288,50,8}$ encodes 50 qubits at $d = 8$; but no discovered indecomposable code exceeds the prior highest $\FOM$ at its own block length.

Campaign~4 provides evidence that the tradeoff is not an artifact of the trinomial restriction: higher-weight mixed-monomial codes (weight-8) access new $(k,d)$ combinations (e.g., $k = 50$ at $d = 8$) but do not escape the envelope.
The PBB ansatz produces non-CSS codes that match but do not exceed $\FOM = 12.0$ at $n = 144$, and the highest PBB $\FOM$ ($\leq 19.2$, trusted) remains an upper bound.
An important caveat: Campaign~5 excluded the $(12,12)$ and $(15,12)$ lattices where the highest-$\FOM$ CSS codes were found; extending PBB evolution to these lattices could reveal whether the non-CSS ansatz can escape the CSS envelope. 
Whether any construction over $\FF_2[x,y]/(x^\ell{-}1, y^m{-}1)$ can exceed the observed tradeoff---consistent with $d = O(\sqrt{n})$~\cite{postema2025existence}---remains open.

\subsection{Comparison with prior work}
\label{sec:comparison}

Table~\ref{tab:comparison} compares our highest-$k$ CSS codes with prior BB records at each block length, and Table~\ref{tab:useful} consolidates the most practically relevant codes ($d \geq 8$, $\FOM \geq 6$) across all five campaigns.

\begin{table}[t]
\centering
\caption{CSS comparison with prior art at matched block lengths.
$k$ ratio $= k_{\text{ours}}/k_{\text{Bravyi}}$; both $k$ values are exact ($\FF_2$ rank).
The $\FOM$ column reports the highest new trusted CSS $\FOM$ at each~$n$ after excluding Bravyi-equivalent rediscoveries and the decomposable $\code{288,24,12}$.
$^*$Campaign~4 (mixed-monomial, weight-8 stabilizers).}
\label{tab:comparison}
\scriptsize
\setlength{\tabcolsep}{3pt}
\begin{tabular}{@{}ccccc@{}}
\toprule
$n$ & Bravyi~\cite{bravyi2024high} & Ours CSS (max $k$) & $k$ ratio & Best new $\FOM$ \\
\midrule
144 & $\code{144,12,12}$ & $\code{144,54,{\leq}4}^*$ & $4.5\times$ & $8.0^*$ \\
288 & $\code{288,12,18}$ & $\code{288,50,8}^*$ & $4.2\times$ & $11.1^*$ \\
360 & $\code{360,12,{\leq}24}$ & $\code{360,40,4}$ & $3.3\times$ & ${\leq}10.9^*$ \\
\bottomrule
\end{tabular}
\end{table}

\begin{table*}[t]
\centering
\caption{Practically relevant codes ($d \geq 8$, $\FOM \geq 6$) with prior art.
``$d$~status'': E = MILP exact, I = MILP incumbent.
``Wt'' = stabilizer weight (controlled-NOT (CNOT) rounds per syndrome cycle).
Our PBB $\code{360,12,{\leq}24}$ and Bravyi et al.'s CSS $\code{360,12,{\leq}24}$ share $(n,k,d_{\rm upper})$ and FOM upper bound (${\leq}19.2$); they are structurally different codes (PBB: weight-8 mixed stabilizers; Bravyi: weight-6 CSS), and the $\FOM$ comparison is a tie of upper bounds, not a record.
Complete catalogs are available at the Supplemental Materials and ~\protect\cite{cruzbenito2026qcode}.}
\label{tab:useful}
\small
\begin{tabular}{llccccl}
\toprule
Source & Code & $d$ status & $\FOM$ & Wt & Type & Notes \\
\midrule
Bravyi~\cite{bravyi2024high} & $\code{360,12,{\leq}24}$ & I & $\leq 19.2$ & 6 & CSS & Highest published $\FOM$ \\ 
Symons~\cite{symons2025covering} & $\code{144,14,14}$ & E & 19.1 & 8 & CSS & 33\% more CNOTs \\
Bravyi~\cite{bravyi2024high} & $\code{288,12,18}$ & E & 13.5 & 6 & CSS & Highest exact wt-6 $\FOM$ \\ 
Bravyi~\cite{bravyi2024high} & $\code{144,12,12}$ & E & 12.0 & 6 & CSS & Gross code \\
Wang~\cite{wang2024coprime} & $\code{150,16,8}$ & E & 6.8 & 6 & CSS & Prior highest $k$ ($n \leq 360$) \\ 
Khesin~\cite{khesin2026mirror} & $\code{60,4,10}$ & E & 6.7 & 6 & non-CSS & Highest-$\FOM$ mirror code \\
\midrule
This work & $\code{288,50,8}$ & E & 11.1 & 8 & CSS & Cross-factored \\
This work & $\code{360,20,{\leq}14}$ & I & $\leq 10.9$ & 8 & CSS & Mixed-monomial \\
This work & $\code{360,16,{\leq}14}$ & E$^{*}$ & ${\leq}8.7$ & 6 & CSS & $^{*}$7/32 logicals proven \\
This work & $\code{288,16,12}$ & E & 8.0 & 6 & CSS & All shifts $\leq 3$ \\
This work & $\code{144,8,12}$ & E & 8.0 & 6 & CSS & Mixed-monomial \\
\midrule
This work & $\code{360,12,{\leq}24}$ & I & $\leq 19.2$ & 8 & PBB & Highest PBB $\FOM$ \\
This work & $\code{144,12,12}$ & E & \textbf{12.0} & 8 & PBB & Non-CSS gross code \\
This work & $\code{108,8,10}$ & E & 7.4 & 8 & PBB & \\
\bottomrule
\end{tabular}
\end{table*}

At every block length, the maximum encoding dimension increases substantially when higher check weights are permitted: $k = 54$ at $n = 144$ ($4.5\times$ Bravyi et al.'s $k = 12$, weight-8), $k = 50$ at $n = 288$ ($4.2\times$, weight-8), $k = 40$ at $n = 360$ ($3.3\times$, weight-6).
These comparisons are exact ($k$ is computed via $\FF_2$ rank) and independent of distance estimation.
However, no discovered indecomposable code exceeds the prior highest $\FOM$ at its own block length: the $\code{288,24,12}$ matches the gross code's $\FOM = 12.0$ but is a direct sum of two gross codes (Sec.~\ref{sec:families}); the highest-$\FOM$ previously unreported weight-6 CSS code, $\code{288,16,12}$ ($\FOM = 8.0$), falls below the $\code{288,12,18}$'s $\FOM = 13.5$.
The concurrent Bayesian-optimization result of Chengyu et al.~\cite{chengyu2026bayesian}---a $\code{144,36,?}$ without verified distance---highlights the importance of exact verification: our results show that high-$k$ codes routinely have $d \leq 4$.

For non-CSS codes, our PBB codes are complementary to the mirror codes of~\cite{khesin2026mirror}: the $\code{108,8,10}$ ($\FOM = 7.4$) matches the highest-$\FOM$ mirror code at comparable~$n$, while accessing larger block lengths not yet explored by that approach.

Among the most comprehensive published searches, Liang et al.~\cite{liang2025generalized} enumerate all weight-6 generalized toric codes of a canonical form $f = 1{+}x{+}x^a y^b$ for $n \leq 400$; at $n = 288$ they report only the $\code{288,12,18}$.
Our $x/y$-swap codes do not fall within their canonical form, though some may be equivalent under lattice automorphisms.
The $\code{288,16,12}$ at $(12,12)$ lies within their search space but was not reported, presumably because the $\code{288,12,18}$ dominates it in $\FOM$.
To our knowledge, 97 of our 99 CSS equivalence classes do not appear in prior catalogs; the two matches are Bravyi et al.'s gross code and $\code{360,12,{\leq}24}$, which Campaign~4 rediscovered through mixed-monomial representations. Extended novelty assessment---including comparison with weight-8 codes and overlap analysis with prior systematic searches---appears in the Supplemental Material.

\subsection{Code capacity simulations}
\label{sec:code_capacity}

We simulate two channels with BP-OSD decoding (OSD-CS order~7, 20~BP iterations): (i)~independent bit-flip noise, where each qubit gets $X$ at rate~$p$, decoded on $H_z$ with an iid prior of marginal rate~$p$; and (ii)~the depolarizing channel, where each qubit independently gets $X$, $Y$, or $Z$ with probability $p/3$, decoded on the symplectic check matrix $[H_z \mid H_x]$ with iid bit-flip priors at marginal rate $2p/3$ on each of the $2n$ columns.\footnote{\label{fn:depolarizing-decoder}For the depolarizing channel this iid prior is mismatched: under depolarizing noise $P(e_x{=}1, e_z{=}1) = p/3$ (induced by $Y$ errors) versus the iid prediction $(2p/3)^2$, so BP-OSD with iid bit priors does not model the $X$/$Z$ correlation. The depolarizing-channel results reported in this section are therefore values attainable under this specific (sub-optimal) decoder configuration---a lower bound on the depolarizing threshold under an optimal correlated decoder (e.g.\ quaternary BP over GF(4))---not the optimal depolarizing threshold itself. CSS-vs.-PBB comparisons under depolarizing noise are likewise decoder-dependent.}
Eight CSS and four non-CSS codes are simulated at 100{,}000 shots per rate.
The block logical error rate (LER) counts any trial with at least one logical error among~$k$ qubits; Wilson 95\% confidence intervals are $\leq 0.06$~percentage points for LER~$\leq 1\%$.
Note that CSS codes may have different X and Z distances; results under depolarizing noise (Table~\ref{tab:threshold_noncss}) can differ from X-only because both noise channels probe different distance properties.
For the CSS $\code{144,12,12}$ under depolarizing noise, the lower block LER compared to X-only noise reflects the iid X/Z prior placing X- and Z-bit error probabilities at marginal rate $2p/3$ each (footnote~\ref{fn:depolarizing-decoder}), giving an effective per-Pauli rate below the pseudo-threshold.

\begin{figure}[t]
\centering
\includegraphics[width=\columnwidth]{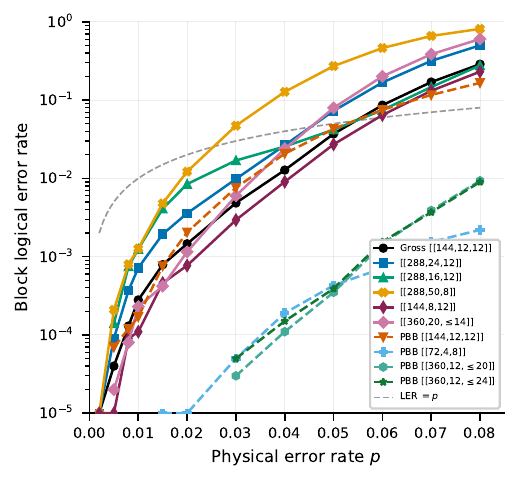}
\caption{Block logical error rate vs.\ physical bit-flip rate~$p$ for six CSS
and four non-CSS codes (100{,}000 shots per rate, BP-OSD with OSD-CS order~7).
Dashed line: LER~$= p$.  Weight-6 codes with $d \geq 12$ achieve LER~$< p$ through
$p \approx 5\%$; the PBB $\code{360,12,{\leq}24}$ and $\code{72,4,8}$ achieve
LER~$< 1\%$ at $p = 8\%$.  The weight-8 $\code{288,50,8}$ ($d = 8$, $k = 50$)
crosses LER~$= p$ between $p = 2\%$ and $3\%$, illustrating the rate--distance tradeoff.}
\label{fig:ler_curves}
\end{figure}

Among CSS codes, the gross code $\code{144,12,12}$ achieves block LER~$< p$ through $p \approx 5\%$, validating the framework.
The $\code{288,16,12}$ matches this pseudo-threshold despite encoding 4 additional qubits; the $\code{288,24,12}$ achieves LER~$< p$ through $p = 4\%$ with per-qubit rates $p_L$ identical to the gross code at every tested rate (Supplemental Table~XVI), consistent with its direct-sum decomposition into two independent gross codes.
(Block LER and per-qubit rate are related by $p_L = 1 - (1-\text{LER})^{1/k}$, defined in Appendix~\ref{app:threshold}; for the gross code at $p = 1\%$, LER~$= 0.03\%$ gives $p_L \approx 0.003\%$, rounded to ${<}0.01\%$ in the tables. Wilson 95\% confidence intervals are $\leq 0.06$~pp at LER~$\leq 1\%$; see Appendix~\ref{app:threshold} for details.)
The mixed-monomial $\code{144,8,12}$ achieves LER~$< p$ through $p = 5\%$, matching the gross code's pseudo-threshold with fewer logical qubits ($k = 8$) and per-qubit rates $p_L \leq 0.34\%$ through $p = 5\%$.
At $n = 360$, Campaign~4's $\code{360,20,{\leq}14}$ achieves LER~$< p$ through $p = 4\%$, encoding 4 more qubits than the $\code{360,16,{\leq}14}$ at comparable per-qubit error rates.
The weight-8 $\code{288,50,8}$ ($k = 50$, $d = 8$) achieves LER~$< p$ only through $p \approx 2\%$; its high encoding rate ($k/n = 17\%$) degrades the pseudo-threshold despite $\FOM = 11.1$.
The $\code{144,24,6}$ ($d = 6$) provides limited suppression (LER~$< p$ for $p \leq 1.5\%$), consistent with its shorter distance.
Overall, indecomposable codes with $d \geq 12$ achieve gross-code-level pseudo-thresholds while encoding up to $1.7\times$ as many logical qubits ($\code{360,20,{\leq}14}$, $k = 20$ vs.\ $k = 12$).

The non-CSS $\code{144,12,12}$ PBB code shares $\code{n,k,d}$ with the gross code, enabling a direct head-to-head---though its type-1 (mixed) stabilizers have weight~8 versus weight~6, so circuit-level costs differ.%
Under X-only noise, both achieve LER~$< p$ through $p = 5\%$, but the PBB code achieves lower LER at high error rates (LER~$= 16.6\%$ vs.\ $28.8\%$ at $p = 8\%$), consistent with its $\approx 1.6\times$ more independent Z-syndrome constraints: $\mathrm{rank}_{\FF_2}(H_Z\text{-block}) = 108$ for the PBB code (where every stabilizer row has nonzero Z-support, from $C, D$ in block~1 and $B^\top, A^\top$ in block~2) versus $\mathrm{rank}_{\FF_2}(H_Z) = 66$ for the gross code (computed directly from the stabilizer matrices for both codes).
At moderate rates ($p = 3$--$4\%$), the CSS code has slightly lower LER, suggesting the non-CSS advantage diminishes when errors are sparse relative to~$d$.
Under depolarizing noise (decoded with the iid X/Z BP-OSD prior; footnote~\ref{fn:depolarizing-decoder}), the advantage disappears: both achieve nearly identical LER at all rates (9.4\% vs.\ 9.3\% at $p = 8\%$), consistent with the CSS code's separate $H_X$, $H_Z$ providing equivalent detection when both X and Z errors are present under that prior (Table~\ref{tab:threshold_noncss}); a correlation-aware decoder could change this comparison.
The $\code{72,4,8}$ achieves LER~$< 0.3\%$ even at $p = 8\%$ under X-only noise, consistent with $d/\sqrt{n} \approx 0.94$ providing strong error suppression at $n = 72$.
The two simulated $n = 360$ PBB codes---the top-FOM $\code{360,12,{\leq}24}$ ($\FOM \leq 19.2$) and the $\code{360,12,{\leq}20}$ ($\FOM \leq 13.3$)---both achieve LER~$< p$ through all tested rates ($p \leq 8\%$) under both noise models, with LER~$= 0.90\%$/$4.52\%$ and $0.93\%$/$4.97\%$ (X-only/depolarizing) at $p = 8\%$ respectively; their $d \leq 24$ and $d \leq 20$ provide the strongest error suppression among all simulated codes.
Condensed simulation data appear in Appendix~\ref{app:threshold}; full data at all 12 error rates are in the Supplemental Material.

\subsection{BP-OSD overestimation}
\label{sec:bposd_findings}

MILP ground truth reveals systematic and severe BP-OSD overestimation for high-rate codes.
A representative example illustrates the severity: the $\code{360,40,2}$ yielded $d_{\text{BP}} \leq 24$ at 150k trials, $d \leq 6$ at 1.5M trials, and $d = 2$ by MILP---a $12\times$ overestimate that no feasible trial budget would resolve.
By contrast, the gross code ($k = 12$) returned $d = 12$ in all 30 batches with zero variance.

Per-decoder success rates stratified by encoding rate, full per-batch distance distributions, and extended 1{,}500{,}000-trial verification results are reported in the Supplemental Material.

\begin{figure}[t]
\centering
\includegraphics[width=\columnwidth]{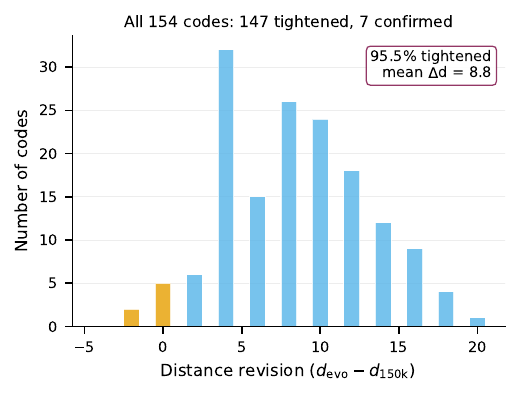}
\caption{Distance revisions across 154 Campaigns~1--3 polynomial representations.
147 (95.5\%) had bounds tightened under the 150k-trial multi-decoder protocol, with mean reduction~${\sim}8.4$ points.
Even tightened bounds remain gross overestimates: MILP reveals $d_{\text{true}} = 2$ for codes with $d_{\text{BP}} \leq 24$.}
\label{fig:tightening}
\end{figure}

Across all 154 Campaigns~1--3 representations, 147 (95.5\%) had bounds tightened under the multi-decoder protocol (Fig.~\ref{fig:tightening}), with a mean revision of ${\sim}8.4$ points.
OSD-CS$_{10}$ configurations consistently outperform OSD$_0$: the latter found the global minimum for only 49 of 154 codes (32\%), versus 100 (65\%) for OSD-CS$_{10}$/minimum-sum.
Yet even OSD-CS$_{10}$ produces large overestimates on high-rate codes: for the $\code{288,32,4}$ (MILP $d = 4$), OSD$_0$ returns a median of $d = 39$ (range 30--44) while OSD-CS$_{10}$/minimum-sum returns a median of $d = 26$ (range 18--30)---a $4.5$--$11\times$ overestimate.

The $A = B$ code $\code{144,32,2}$ (Theorem~\ref{thm:ab_d2}) illustrates the failure mode concretely: 29 of 30 batches at 150k trials failed to find any weight-2 logical, despite the theorem guaranteeing $k/2 = 16$ independent ones (72 total).
Only at 1.5M trials did a single batch find $d = 2$.
This demonstrates that BP-OSD cannot be trusted to find low-weight operators even when they are plentiful.

Campaign~5 extends these findings to non-CSS codes: deep MILP of 149 PBB catalog entries found 33 downward corrections (22\% rate), the largest being $d = 24 \to 16$ at $n = 360$.
An additional failure mode specific to non-CSS codes emerged: naive random syndrome sampling achieves 0\% decode success on many PBB codes because the achievable syndrome subspace can be exponentially small (Sec.~\ref{sec:noncss_pipeline}).

These findings are consistent with the concurrent benchmark of Webster et al.~\cite{webster2026distance}, who observe that BP-OSD fails to find the correct distance for codes as small as $\code{48,5,10}$ without enhanced sampling strategies.
Their QDistEvol algorithm and random permutation techniques may reduce the overestimation gap; integrating such methods into evolutionary search pipelines is a natural extension.
These findings have direct implications: distance claims for BB codes with $k/n > 0.1$ based on fewer than ${\sim}10^5$ trials with a single decoder should be treated with caution; MILP verification is essential where feasible.

\subsection{Ablation: LLM vs.\ genetic algorithm vs.\ random search}
\label{sec:ablation}

We evaluate the LLM's contribution using the deterministic metric $\Sigma_k = \sum_{\ell,m} \max\{k\}$ (sum of highest~$k$ per lattice), since MILP shows the distance component of the evolution's fitness was unreliable. 
Because $k$ is computed exactly via $\FF_2$ rank, this metric has zero stochasticity.
We compare eight arms on the same 8 Stage-2 lattices: the hand-crafted seed, uniform random trinomials at two budgets ($10^3$ and $10^4$ per lattice), a genetic algorithm (GA) with tournament selection on exponent tuples, a GA on the same ansatz representation the LLM uses (``GA-G''), and the highest-fitness evolved ansätze from Campaigns~1--3.

The key result is that GA-G achieves $\Sigma_k = 1{,}037 \pm 22$ (5 seeds)---exceeding the LLM's $704$ by $47\%$ with only $14{,}000$ evaluations versus $213{,}000$.
However, GA-G converges exclusively to low-distance codes ($d \leq 2$), reaching $k = n/2$ (the theoretical maximum) at every lattice across the union of seeds, all with $d \leq 2$.
The exponent-tuple GA ($\Sigma_k = 703 \pm 142$) matches Campaign~1 but likewise through low-distance codes at standard lattices.
At lattices $(16,9)$ and $(18,8)$---where the LLM fails---the exponent-tuple GA finds genuine $d \geq 3$ codes, but MILP verification yields $d = 4$--$6$ and $\FOM \leq 4.0$, a $3\times$ gap below the LLM's $\FOM = 12.0$.
Equal-budget random search ($\Sigma_k = 299 \pm 26$) shows that uniform sampling yields diminishing returns: a $10\times$ budget increase over the $10^3$/lattice baseline produces only a $1.7\times$ improvement.

This result controls for the effect of the ansatz representation: given the \emph{same} ansatz representation, a classical GA achieves higher $\Sigma_k$ but only through $d \leq 2$ codes.
This is consistent with a qualitative difference in mutation distribution---the LLM produces mutations that yield ans\"atze whose verified outputs include codes with $d \geq 6$ and $\FOM$ up to $8.0$ (indecomposable CSS), or $12.0$ counting the decomposable $\code{288,24,12}$, while the GA does not.
We caution against stronger interpretations: the comparison does not control for the budget mismatch (213k vs.\ 14k evaluations), the GA-G arm was not given distance feedback, and the LLM's own fitness signal during Campaigns~1--3 was inflated by BP-OSD overestimation (Sec.~\ref{sec:bposd_findings}).
Campaign~1's highest-fitness ansatz was found at generation depth~5 (5 mutations from the seed), with the largest fitness gain arising from the independent invention of a ``Pure-Axis'' (univariate) strategy after observing that the highest-$k$ code used separated variables---a type of structural insight that abstract syntax tree (AST)-level mutation operators cannot produce by construction.
Detailed per-lattice results appear in Appendix~\ref{app:ablation}.

\section{Discussion}
\label{sec:discussion}

\paragraph{Pipeline design lessons.} 
The most important lesson from five campaigns is that pipeline design matters more than any single code discovery.
Three design decisions proved critical:
(i)~the multi-stage cascade that screens on $k$ first (exact, cheap) and defers distance estimation to later stages, ensuring the search signal is never corrupted even when BP-OSD overestimates;
(ii)~the introduction of MILP-in-the-loop verification for Campaigns~4--5 after post-hoc analysis of Campaigns~1--3 revealed that BP-OSD distance estimates were inflated by up to $12\times$ (Sec.~\ref{sec:bposd_findings}), rendering the fitness signal unreliable for high-rate codes;
and (iii)~post-campaign verification (Tanner-graph decomposability, BLISS deduplication, Clifford equivalence for PBB) that caught the $\code{288,24,12}$ direct-sum decomposition and reduced 720 tuple-distinct PBB codes to 368 BLISS-distinct ones.
The framework produced ansätze whose verified outputs cluster into recognizable algebraic families---the univariate/HGP subspace, $x/y$-swap codes, factored-product codes---despite the misleading distance estimates in the fitness function during Campaigns~1--3, because $k$ is computed exactly via $\FF_2$ rank and was never corrupted.
Campaign~3's ``saturation'' was convergence to a local optimum of this \emph{effectively $k$-only} objective, not evidence that no better codes exist.

\paragraph{Hardware feasibility.}
All trinomial BB codes share degree-6 qubit connectivity. K\"onig's line coloring theorem guarantees 6 CNOT rounds per check type (X or Z separately), giving depth~12 in the worst case; Bravyi et al.~\cite{bravyi2024high} found a depth-7 CNOT circuit by interleaving X and Z check rounds.
Coupling ranges compare favorably under the assumption that qubits are laid out on an $\ell \times m$ grid with periodic boundary conditions, so that a shift $x^a$ corresponds to a coupler spanning $a$ sites along the $\ell$-direction: the $\code{288,16,12}$ has maximum $x$-shift~$x^3$ (3 out of $\ell = 12$, 25\% of that cyclic dimension) and maximum $y$-shift~$y^3$ (3 out of $m = 6$, 50\%), the shortest coupling range among $d = 12$ codes.
The $\code{288,24,12}$ (a direct sum of two gross codes) has maximum shift $x^6$ ($6/12 = 50\%$), inheriting the gross code's coupling fraction.
Campaign~4 codes with 4-term polynomials require degree-8 connectivity (8 CNOT rounds); whether the $\FOM$ advantage of weight-8 codes compensates for the increased circuit depth is hardware-dependent.
For non-CSS PBB codes, each block-1 stabilizer has both X and Z support, requiring interleaved CNOT and CZ gates; fault-tolerant syndrome extraction circuits remain an open problem~\cite{khesin2026mirror}.

\paragraph{Limitations.}
The search explored a restricted scope: CSS campaigns required $3 \mid \ell$ at Stage~1, missing lattices $(16,9)$ and $(18,8)$; Campaign~5 excluded the large-$m$ lattices $(12,12)$ and $(15,12)$ where the highest-$\FOM$ CSS codes were found; Campaign~3 saturated after ${\sim}225$ iterations, suggesting the ansatz representation was largely exhausted.
Different seeds, larger populations, or alternative representations could access unexplored regions.
The codes of~\cite{bravyi2024high} were optimized for pseudo-threshold behavior under circuit-level noise, a complementary criterion to the combinatorial $\FOM$ used here; circuit-level simulations are necessary to assess practical fault-tolerant performance.

\paragraph{Future directions.}
Escaping the rate--distance envelope likely requires moving beyond the bivariate bicycle family: higher-rank abelian groups (multivariate multicycle codes)~\cite{mian2026multivariate}, non-abelian lifts~\cite{breuckmann2021balanced}, or other quotient ring structures.
Circuit-level noise simulations for the $\code{288,16,12}$, the non-CSS $\code{144,12,12}$, and the $\code{360,12,{\leq}24}$ are a natural next step.
Combining LLM-guided search with Bayesian optimization~\cite{chengyu2026bayesian} or reinforcement learning, and scaling PBB evolution to $n > 360$ with improved MILP solvers are further opportunities.

\section{Conclusion}
\label{sec:conclusion}

We have presented an LLM-guided evolutionary pipeline for quantum code discovery, combining MAP-Elites program evolution with multi-stage verification including MILP distance computation, BLISS deduplication, and post-campaign structural analysis.
Applying this framework to bivariate bicycle codes, five campaigns discover 465 distinct codes---97 CSS (weight-6 and weight-8) and 368 non-CSS perturbed bivariate bicycle---at $n \leq 360$, with encoding dimensions up to $k = 54$ (prior weight-6 maximum: $k = 16$).
The codes most relevant for error correction are the CSS $\code{288,16,12}$ ($d = 12$ exact, all shifts~$\leq 3$), the non-CSS $\code{360,12,{\leq}24}$ ($\FOM \leq 19.2$, the highest trusted non-CSS FOM, LER~$< p$ at all tested rates), and the $\code{144,12,12}$ PBB (matching $\FOM = 12.0$ with a non-CSS stabilizer structure).
Early campaigns revealed verification pitfalls that drove pipeline improvements: the $\code{288,24,12}$, which Tanner graph analysis revealed to be a direct sum of two gross codes, and the $A{=}B$ distance trap (Appendix~\ref{app:ab_trap}), undetected by BP-OSD even at $1.5 \times 10^6$ trials. 

MILP distance computation reveals an empirical rate--distance tradeoff across the searched catalog: among weight-6 codes, high-$k$ codes have $d \leq 4$, while indecomposable $d = 12$ codes are limited to $k \leq 16$.
Higher-weight codes access new $(k,d)$ combinations---e.g., $k = 50$ at $d = 8$ with weight-8 stabilizers---but do not escape this envelope.
This tradeoff persists across polynomial term counts and check weights within the CSS family; the PBB non-CSS ansatz matches but does not exceed $\FOM = 12.0$ at $n = 144$, though whether similar non-CSS constructions over the same ring can escape this envelope remains open (asymptotically good codes exist~\cite{panteleev2022asymptotically,leverrier2022quantum,dinur2023good} but use fundamentally different algebraic structures).
All BB codes with $A = B$ have $d = 2$ (Appendix~\ref{app:ab_trap})---a distance trap undetected by BP-OSD even at $1.5 \times 10^6$ trials but caught immediately by the pipeline's MILP verification. 

More broadly, BP-OSD overestimates distance by up to $12\times$ for high-rate codes, and standard syndrome sampling fails entirely on non-CSS codes; MILP verification and achievable-syndrome sampling are essential for reliable distance claims.
Code-capacity simulations confirm that indecomposable codes with $d \geq 12$ achieve gross-code-level pseudo-thresholds while encoding up to $1.7\times$ as many logical qubits, and that the non-CSS $\code{144,12,12}$ outperforms its CSS counterpart under X-only noise at $p \geq 6\%$ (though the advantage disappears under depolarizing noise as decoded by our iid X/Z BP-OSD; see footnote~\ref{fn:depolarizing-decoder}).
The non-CSS $\code{360,12,{\leq}24}$ achieves LER~$< p$ through all tested rates under both $X$-only and depolarizing noise.
The total cost of ${\sim}$US\$400 indicates that LLM-guided program evolution is a feasible tool for quantum code discovery.

\begin{acknowledgments}
We thank the IBM Research and IBM Quantum teams for discussions and feedback.
\end{acknowledgments}

\paragraph{Data and code availability.}
All source code, evolved programs, LLM prompts and diffs, MILP distance computations, BLISS equivalence analysis, the 3-tier distance pipeline, and complete code catalogs (225 CSS representations, 368 non-CSS codes) are publicly available at \cite{cruzbenito2026qcode} \url{https://github.com/qiskit-community/qcode-discovery}.

\paragraph{Author contributions.}
J.C.-B.\ conceived the project, designed all evaluation and verification protocols (including the 3-tier adaptive distance pipeline and achievable-syndrome sampling), conducted all five evolution campaigns including deep MILP verification, performed all MILP and BLISS analysis, proved the $A = B$ distance-trap theorem, performed the LC non-equivalence verification, characterized the univariate family, derived the $k = 8\ell/3$ formula, and wrote the manuscript.
A.W.C.\ guided research direction for Campaigns~4--5, contributed the HGP-based proof of the univariate encoding dimension, and reviewed the manuscript.
D.K.\ and I.F.\ contributed to experimental design and reviewed the manuscript.

\paragraph{LLM usage.}
Six LLMs served as mutation operators: Gemini~3 Flash Preview~\cite{google2025gemini3} (Campaign~1); Claude~Opus~4.6 (Anthropic), GPT-5.2 (OpenAI), Gemini~3 Pro Preview (Google) (Campaigns~2--3); Claude~Opus~4.6, GPT-5.3-Codex (OpenAI), Gemini~3.1 Pro Preview (Google) (Campaigns~4--5).
They were accessed via LiteLLM cloud APIs at temperature~0.8 (Campaign~4: 1.0) with 16{,}384 maximum output tokens.
Claude~Opus~4.6/4.7 and GPT-5.2/5.5 additionally assisted with code development and manuscript drafting.
The authors are affiliated with IBM; model selection was based on availability and diversity rather than commercial relationships.

\appendix

\section{Code capacity simulation data}
\label{app:threshold}

Tables~\ref{tab:threshold_css} and~\ref{tab:threshold_noncss} give condensed code-capacity results for the CSS and PBB codes of Sec.~\ref{sec:code_capacity}.
The $X$-only model applies independent $X$ errors at rate~$p$; decoding uses BP-OSD (OSD-CS order~7, product-sum, 20~BP iterations).
Block LER counts any of 100{,}000 trials with $\geq 1$ logical error among~$k$ qubits, and $p_L = 1{-}(1{-}\text{LER})^{1/k}$.
Wilson 95\% confidence intervals are $\leq 0.06$ percentage points for LER~$\leq 1\%$; full 12-rate data are in the Supplemental Material.

\begin{table*}[t]
\centering
\caption{Code-capacity simulation for eight CSS codes (weight-6 BB and weight-8) under bit-flip noise.
Entries show block LER (\%) and per-qubit rate $p_L$ (\%) at $p = 1, 3, 5, 8\%$; \textbf{bold} LER entries exceed~$p$.
$^{\dagger}$Campaign~4 code (MILP incumbent $d \leq 14$; BP-OSD 300k confirms $d = 14$).
$^{\ddagger}$Campaign~4 code (weight-8 stabilizers).}
\label{tab:threshold_css}
\small
\begin{tabular}{l c | cc cc cc cc}
\toprule
& & \multicolumn{2}{c}{$p = 1\%$} & \multicolumn{2}{c}{$p = 3\%$} & \multicolumn{2}{c}{$p = 5\%$} & \multicolumn{2}{c}{$p = 8\%$} \\
Code & $k$ & LER & $p_L$ & LER & $p_L$ & LER & $p_L$ & LER & $p_L$ \\
\midrule
$\code{144,12,12}$~\cite{bravyi2024high} & 12 & 0.03 & ${<}0.01$ & 0.49 & 0.04 & 3.70 & 0.31 & \textbf{28.76} & 2.79 \\
$\code{288,24,12}$ & 24 & 0.07 & ${<}0.01$ & 0.98 & 0.04 & \textbf{7.28} & 0.31 & \textbf{50.06} & 2.85 \\
$\code{288,16,12}$ & 16 & 0.13 & 0.01 & 1.69 & 0.11 & 4.18 & 0.27 & \textbf{27.73} & 2.01 \\
$\code{288,50,8}^{\ddagger}$ & 50 & 0.13 & ${<}0.01$ & \textbf{4.71} & 0.10 & \textbf{27.20} & 0.63 & \textbf{81.28} & 3.30 \\
$\code{144,8,12}$ & 8 & 0.01 & ${<}0.01$ & 0.29 & 0.04 & 2.72 & 0.34 & \textbf{23.07} & 3.23 \\
$\code{360,16,{\leq}14}$ & 16 & 0.15 & 0.01 & 1.69 & 0.11 & \textbf{5.27} & 0.34 & \textbf{38.70} & 3.01 \\
$\code{360,20,{\leq}14}^\dagger$ & 20 & 0.02 & ${<}0.01$ & 0.60 & 0.03 & \textbf{7.95} & 0.41 & \textbf{60.29} & 4.51 \\
$\code{144,24,6}$ & 24 & 0.27 & 0.01 & \textbf{7.22} & 0.31 & \textbf{28.97} & 1.42 & \textbf{70.68} & 4.98 \\
\bottomrule
\end{tabular}
\end{table*}

\begin{table*}[t]
\centering
\caption{Code-capacity simulation for five non-CSS PBB codes and the CSS gross-code baseline under $X$-only and depolarizing noise.
Entries are block LER (\%) from 100{,}000 shots; \textbf{bold} entries exceed~$p$.
Under $X$-only noise, the PBB $\code{144,12,12}$ outperforms its CSS counterpart at $p \geq 6\%$; under depolarizing noise, decoded with the iid X/Z BP-OSD prior of footnote~\ref{fn:depolarizing-decoder}, both perform identically (a comparison that is decoder-dependent).
Both the $\code{360,12,{\leq}20}$ and the higher-FOM $\code{360,12,{\leq}24}$ achieve LER~$< p$ through all tested rates under both noise models---at $p = 8\%$, LER~$= 0.93\%$ vs $0.90\%$ ($X$-only) and $4.97\%$ vs $4.52\%$ (depolarizing).}
\label{tab:threshold_noncss}
\small
\begin{tabular}{l c | cccc | cccc}
\toprule
& & \multicolumn{4}{c|}{$X$-only noise} & \multicolumn{4}{c}{Depolarizing noise} \\
Code & $k$ & $1\%$ & $3\%$ & $5\%$ & $8\%$ & $1\%$ & $3\%$ & $5\%$ & $8\%$ \\
\midrule
$\code{144,12,12}$ CSS~\cite{bravyi2024high} & 12 & 0.03 & 0.49 & 3.70 & \textbf{28.76} & 0.03 & 0.37 & 1.27 & \textbf{9.43} \\
$\code{144,12,12}$ PBB & 12 & 0.02 & 0.76 & 4.29 & \textbf{16.61} & 0.02 & 0.34 & 1.40 & \textbf{9.29} \\
$\code{72,4,8}$ PBB & 4 & 0.00 & 0.01 & 0.04 & 0.22 & 0.01 & 0.25 & 1.80 & \textbf{10.47} \\
$\code{108,8,10}$ PBB & 8 & 0.00 & 0.01 & 0.05 & 0.51 & 0.01 & 0.18 & 1.05 & \textbf{8.38} \\
$\code{360,12,{\leq}20}$ PBB & 12 & 0.00 & 0.00 & 0.04 & 0.93 & 0.05 & 0.70 & 1.89 & 4.97 \\
$\code{360,12,{\leq}24}$ PBB & 12 & 0.00 & 0.01 & 0.04 & 0.90 & 0.08 & 0.81 & 1.82 & 4.52 \\
\bottomrule
\end{tabular}
\end{table*}

\section{Ablation details}
\label{app:ablation}

Tables~\ref{tab:ablation_k} and~\ref{tab:ablation_per_lattice} summarize the ablation study (Sec.~\ref{sec:ablation}).
All metrics use exact $\FF_2$ rank; $\Sigma_k = \sum_{\ell,m} \max\{k\}$.
For multi-seed arms, we report mean~$\pm$~std over 5 independent seeds.
GA-G denotes the GA operating on the same ansatz representation as the LLM, with AST-level mutations (integer literal changes, loop bound modification, strategy block duplication/removal/splicing, block reordering). 

\begin{table}[h]
\centering
\caption{$k$-only ablation summary.
$^\dagger$Mean $\pm$ std over 5 seeds.
Both GAs exceed Campaign~1 on $\Sigma_k$ but through low-distance codes ($d \leq 2$); the LLM's highest-$\FOM$ codes with $d \geq 6$ achieve $\FOM = 12.0$ vs.\ $\FOM \leq 4.0$ for either GA.} 
\label{tab:ablation_k}
\small
\begin{tabular}{@{}lrr@{}}
\toprule
Arm & Budget & $\Sigma_k$ \\
\midrule
Random ($10^3$/latt.) & 8{,}000 & 172 \\
Random ($10^4$/latt.)$^\dagger$ & 80{,}000 & $299 \pm 26$ \\
Seed & 21{,}962 & 276 \\
GA (exponent tuples)$^\dagger$ & 80{,}000 & $703 \pm 142$ \\
GA-G (generators)$^\dagger$ & 14{,}000 & $1{,}037 \pm 22$ \\
\midrule
Campaign~1 (Flash) & 213{,}216 & 704 \\
Campaign~2 (Ens.) & 15{,}690 & 464 \\
Campaign~3 (Ens.) & 49{,}699 & 704 \\
\bottomrule
\end{tabular}
\end{table}

\begin{table}[h]
\centering
\caption{Per-lattice maximum~$k$.
Rnd-10k, GA, GA-G: mean over 5 seeds; C1$=$C3 are identical.
The exponent-tuple GA finds low-distance $d \leq 2$ codes at standard lattices and genuine $d = 4$--$6$ codes (MILP-verified, $\FOM \leq 4.0$) at $(16,9)$, $(18,8)$ where LLM evolution fails.
GA-G reaches $k = n/2$ at every lattice across the union of seeds, all with $d \leq 2$.}
\label{tab:ablation_per_lattice}
\small
\begin{tabular}{@{}lrrrrrrrr@{}}
\toprule
Lattice & $n$ & Rnd & Rnd-10k & Seed & GA & GA-G & C1{=}C3 & C2 \\
\midrule
$(12,6)$ & 144 & 32 & 35 & 32 & 74 & 72 & 64 & 64 \\
$(6,12)$ & 144 & 24 & 45 & 32 & 83 & 72 & 64 & 32 \\
$(12,12)$ & 288 & 32 & 45 & 64 & 102 & 144 & 128 & 64 \\
$(24,6)$ & 288 & 16 & 38 & 48 & 115 & 144 & 128 & 64 \\
$(15,12)$ & 360 & 16 & 37 & 40 & 135 & 156 & 160 & 80 \\
$(30,6)$ & 360 & 24 & 48 & 60 & 93 & 180 & 160 & 160 \\
$(16,9)$ & 288 & 12 & 18 & 0 & 45 & 125 & 0 & 0 \\
$(18,8)$ & 288 & 16 & 34 & 0 & 56 & 144 & 0 & 0 \\
\midrule
$\Sigma_k$ & & 172 & 300 & 276 & 703 & 1037 & 704 & 464 \\
\bottomrule
\end{tabular}
\end{table}

\section{Representative LLM mutations}
\label{app:mutations}

Three mutations from Campaign~1's highest-fitness lineage (Gemini~3 Flash, generation depth~5) illustrate the nature of LLM-guided program synthesis. 

\begin{figure}[tbp]
\textbf{Mutation 1: Expanded search range} (Gen~0$\to$1, $+22\%$).
The LLM widened exponent bounds and added a coprimality heuristic reflecting the intuition that $\gcd(a,\ell) = 1$ avoids degenerate orbits:
\begin{lstlisting}[style=diff]
- max_x = min(ell, 5)
+ limit_x = min(ell, 13)
  ...
+ # Strategy 5: GCD-informed exponents
+ coprime_x = [i for i in range(1, ell)
+              if math.gcd(i, ell) == 1]
\end{lstlisting}
\caption{Coprimality heuristic.}
\label{fig:mut1}
\end{figure}

\begin{figure}[tbp]
\textbf{Mutation 2: Univariate strategy} (Gen~4$\to$5, $+40\%$).
A pivotal step: the LLM replaced broad enumeration with a ``Pure-Axis'' strategy after observing that the highest-$k$ code used separated variables---independently converging on univariate/HGP codes ~\cite{eberhardt2024pruning} from the fitness signal alone:
\begin{lstlisting}[style=diff]
+ # Strategy 1: Pure-Axis / Decoupled
+ y_pats = [(1,2),(2,4),(3,6),(1,3),
+           (2,7),(m//3, 2*m//3)]
+ x_pats = [(1,2),(2,4),(3,4),(3,6),
+           (4,8),(2,5),(ell//3, 2*ell//3)]
+ for y1, y2 in y_pats:
+   ...
\end{lstlisting}
\caption{Independent convergence on univariate codes.}
\label{fig:mut2}
\end{figure}

\begin{figure}[tbp]
\textbf{Mutation 3: Over-specialization} (Gen~1$\to$2, $-4\%$).
Over-fitting to one code's exponent ratios reduced coverage, illustrating the exploration--exploitation tension:
\begin{lstlisting}[style=diff]
+ def get_interesting_exponents(limit):
+   steps = [1, 2, 3, 4, 5, 7]
+   return [s%limit for s in steps
+           if s < limit]
\end{lstlisting}
\caption{Over-specialization reduces coverage.}
\label{fig:mut3}
\end{figure}

\section{The $A = B$ distance trap: proof}
\label{app:ab_trap}

\begin{theorem}[Distance trap]
\label{thm:ab_d2}
Every BB code with $A = B$ and $k > 0$ has $d = 2$ exactly.
\end{theorem}

\begin{proof}
For $A = B$, every $X$-stabilizer has the form $(A_r \mid A_r)$, so $\mathrm{rowspace}(H_X)$ is contained in the diagonal subspace $S := \{(\mathbf{w},\mathbf{w}) : \mathbf{w} \in \FF_2^{\ell m}\}$.
The weight-2 vector $\mathbf{v}_i := \mathbf{e}_i + \mathbf{e}_{i+\ell m} = (\mathbf{e}_i, \mathbf{e}_i)$ lies in $S$ and satisfies $H_Z \mathbf{v}_i = A^\top \mathbf{e}_i + A^\top \mathbf{e}_i = \mathbf{0}$, so $\mathbf{v}_i \in \ker(H_Z)$.
It is an $X$-stabilizer iff $\mathbf{e}_i \in \mathrm{rowspace}(A)$, but $k > 0$ forces $\mathrm{rank}(A) < \ell m$, so some $\mathbf{e}_i$ lies outside, yielding a nontrivial weight-2 $X$-logical and $d_X \leq 2$.
The same vector also satisfies $H_X \mathbf{v}_i = \mathbf{0}$ and is a $Z$-stabilizer iff $\mathbf{e}_i \in \mathrm{colspace}(A)$; again $\mathrm{rank}(A) < \ell m$ forces some $\mathbf{e}_i$ outside, giving $d_Z \leq 2$.
Polynomials with $\geq 2$ terms force every column of $H_X, H_Z$ to have weight $\geq 2$, so $d \geq 2$, hence $d = 2$.
\end{proof}

This condition ($A = B$, identical polynomials) is distinct from the antipodal self-duality $B = A^\top$ studied in~\cite{liang2025selfdual}, where self-dual BB codes with $B = A^\top$ can have $d$ as high as~16.
The theorem holds regardless of check weight (any polynomial weight $\geq 2$).

\paragraph{Extension to PBB codes.}
The same argument applies to PBB codes with $A = B$: the X-part of the first stabilizer block is $[A \mid A]$, so the weight-2 Z-type operator $(0|\mathbf{e}_i + \mathbf{e}_{i+\ell m})$ has symplectic inner product zero with every block-1 stabilizer (since the X-part contributes $(A_r)_i + (A_r)_i = 0 \bmod 2$) and with every block-2 stabilizer (which has zero X-part).
These $\ell m$ weight-2 operators lie in the normalizer $N(S)$.
In the CSS case, a direct counting argument shows that $k/2$ of them are nontrivial logicals (Theorem~\ref{thm:ab_d2}).
For PBB codes, the additional Z-content from perturbation polynomials $C$ and $D$ changes the stabilizer group, but the normalizer membership of these weight-2 operators is preserved; whenever any such operator is not itself a stabilizer element---as verified by MILP for all $A = B$ PBB codes in the catalog---$d \leq 2$.
Combined with $d \geq 2$ (from polynomial weight $\geq 2$), this gives $d = 2$.

\section{Encoding dimension of univariate codes}
\label{app:crt}

\begin{theorem}[Univariate encoding dimension]
\label{lem:crt_k}
Let $\ell$ and $m$ be positive integers divisible by~3 and let $c = \ell/3$. The BB code defined by $A(y) = 1{+}y{+}y^2 \in \FF_2[y]/(y^m{-}1)$ and $B(x) = A(x^c) \in \FF_2[x]/(x^\ell{-}1)$ is a hypergraph product code that encodes $k = 8\ell/3$ logical qubits.
\end{theorem}

\begin{proof}
Let $H_A\in {\mathbb F}_2^{m\times m}$ and $H_B\in {\mathbb F}_2^{\ell\times\ell}$ be the circulant matrices associated with $A(y)$ and $B(x)$, respectively. The check matrices of the bivariate bicycle code are
\begin{equation}
H_X = \left(\begin{array}{c|c}
I_\ell\otimes H_A & H_B\otimes I_m
\end{array}\right)
\end{equation}
and
\begin{equation}
H_Z = \left(\begin{array}{c|c}
H_B^T\otimes I_m & I_\ell\otimes H_A^T
\end{array}\right).
\end{equation}
The hypergraph product code $HGP(H_1,H_2)$ with $H_1\in {\mathbb F}_2^{r_1\times n_1}$ and $H_2\in {\mathbb F}_2^{r_2\times n_2}$ has check matrices
\begin{equation}
H_X^{HGP} = \left(\begin{array}{c|c}
H_1\otimes I_{n_2} & I_{r_1}\otimes H_2^T
\end{array}\right)
\end{equation}
and
\begin{equation}
H_Z^{HGP} = \left(\begin{array}{c|c}
I_{n_1}\otimes H_2 & H_1^T\otimes I_{r_2}
\end{array}\right).
\end{equation}
Choosing $H_1=H_B$ and $H_2=H_A^T$ with $r_1=n_1=\ell$ and $r_2=n_2=m$ produces a HGP code that is equivalent to the BB code with the left and right blocks of qubits swapped. Therefore $BB(A,B)$ is equivalent to $HGP(H_B,H_A^T)$. For binary parity check matrices $H_1$, $H_2$ of classical codes encoding $k_1$, $k_2$ bits, respectively, and whose transposes are classical codes encoding $k_1^T$, $k_2^T$ bits, the hypergraph product code $HGP(H_1,H_2)$ encodes $k_1k_2+k_1^Tk_2^T$ qubits \cite{tillich2014quantum}. For an $N$ by $N$ circulant matrix $H$ corresponding to the polynomial $f(z)\in {\mathbb F}_2[z]/(z^N-1)$, a standard cyclic-code identity~\cite{macwilliams1977theory} gives
\begin{equation}
\textrm{dim}\ \textrm{ker}\ H = \textrm{deg}\ \textrm{gcd}(f, z^N-1).
\end{equation}
The transpose of a square matrix $H$ has the same kernel dimension as $H$. Now, $(y-1)A(y)=y^3-1$ and $y^3-1$ divides $y^m-1$ whenever $m$ is divisible by 3, so $A(y)$ divides $y^m-1$. Therefore, $k_A=k_A^T=\textrm{deg}\ A=2$. Likewise, $(x^c-1)B(x)=x^\ell-1$, so $B(x)$ divides $x^\ell-1$ and we find $k_B=k_B^T=\textrm{deg}\ B=2c=2\ell/3$. Putting this together, we conclude that
\begin{align}
k & = k_1k_2+k_1^Tk_2^T = k_Bk_A^T + k_B^Tk_A \\
& = 2k_Bk_A^T = 2\cdot \frac{2\ell}{3}\cdot 2 = 8\ell/3.  %
\end{align}
\end{proof}

\section{Local Clifford equivalence of PBB codes}
\label{app:lc}

A natural question is whether the 368 non-CSS PBB codes (full catalog in the Supplemental Material) are merely CSS codes in disguise---that is, equivalent to a CSS code via per-qubit single-qubit Clifford conjugation (local Clifford, or LC, equivalence).

Throughout this appendix, \emph{block~1} refers to the first $\ell m$ qubits (indices $1,\dots,\ell m$, the columns of $A$ in $H$) and \emph{block~2} to the second $\ell m$ qubits (indices $\ell m{+}1,\dots,2\ell m$, the columns of $B$); a per-qubit Clifford assignment is a tuple of $2\ell m$ single-qubit Cliffords, one per qubit, and ``uniform per block'' means the assignment is constant on each of these two index sets.

\paragraph{Group-CSS rank condition.}
For each PBB code, we test whether any per-qubit Clifford assignment renders the stabilizer group CSS---in the sense of~\cite{calderbank1996good,steane1996multiple}.
We do this by first applying the corresponding product of single-qubit Cliffords to the stabilizer generators, and then testing whether the resulting stabilizer group $\mathcal{S}$ is CSS.
We are aided by the fact that $\mathcal{S}$ is CSS iff $\operatorname{rank}[X \mid Z] = \operatorname{rank} X + \operatorname{rank} Z$ over~$\mathrm{GF}(2)$; this is Lemma~7.4 of~\cite{cross2025small} (attributed there to T.~Yoder), to which we refer for the proof.
Cross and Vandeth additionally remark in their~\S7.2 that no fast test for LC-CSS-equivalence is currently known, which is why our procedure below must enumerate Clifford patterns explicitly and admits the residual gaps itemized in ``Coverage and gaps.''

\paragraph{Reduction to 6 Clifford representatives.}
The single-qubit Clifford group has 24 elements but only 6 cosets modulo the single-qubit Pauli group. Pauli conjugation flips signs of stabilizer generators without changing their X- and Z-supports, and whether the stabilizer \emph{group} is CSS depends only on the symplectic subspace it spans, not on those signs. It therefore suffices to check one representative per coset; we use $\{I, S, H, HS, SH, HSH\}$, the same reduction as in~\cite[\S7.2]{cross2025small}.

\paragraph{Hadamard 2-coloring.}
The check at Sec.~\ref{sec:pbb} that a stabilizer code is Hadamard-equivalent to CSS at the level of its supplied generators is decided by a parity 2-coloring whose correctness we sketch here (a similar 2-coloring was used by~\cite{khesin2026mirror} in the mirror-code setting; the derivation below is given directly for general stabilizer codes).
Write each Pauli factor on qubit $j$ as a pair $(x_j, z_j) \in \FF_2^2$ (so $I = (0,0)$, $X = (1,0)$, $Z = (0,1)$, $Y = (1,1)$); conjugation by $H$ on qubit $j$ swaps coordinates, $(x_j, z_j) \mapsto (z_j, x_j)$, and sends $Y \to -Y$ (the symplectic representation of $Y$ is invariant).
For $J \subseteq \{1,\dots,n\}$ let $s_j = \mathbb{1}[j \in J]$ and let $c_r \in \{0,1\}$ encode the desired post-conjugation type of generator $g_r$ ($c_r = 1$ pure-$X$, $c_r = 0$ pure-$Z$).
\emph{Y obstruction:} if any $g_r$ has Y at qubit $j$ both coordinates are set; conjugation by any $H_J$ leaves $g_r$ mixed.
\emph{No Y, parity 2-coloring:} on each $j \in \mathrm{supp}(g_r)$ the local Pauli is X or Z; let $t_{rj} \in \{0,1\}$ denote which.
The post-conjugation X-bit at $(r,j)$ is $t_{rj} \oplus s_j$, and demanding it equal $c_r$ on every $j \in \mathrm{supp}(g_r)$ gives $s_j = c_r \oplus t_{rj}$.
For two generators acting on the same qubit $j$ the right-hand sides must agree, yielding $c_{r_1} \oplus c_{r_2} = t_{r_1, j} \oplus t_{r_2, j}$---equal Pauli types force equal colors, opposite types force different colors.
Solvability is a 2-coloring problem decidable in linear time by union-find with parity.
The decision is sound and complete \emph{at the level of the supplied generators}: a True result exhibits an explicit $H_J$ pattern that conjugates each generator into pure-X or pure-Z form (a fortiori making the stabilizer group CSS); a False result rules out any $H_J$ that achieves this for the supplied generators.
The strictly broader question of whether the stabilizer \emph{group} can be made CSS under a different (row-reduced) generating set after some $H_J$ is decoupled from the parity 2-coloring; sub-classes of that question are addressed by the rank-condition machinery above and by the non-uniform $\{I,S\}$/$\{H,HS\}$ enumeration that follows, with the residual gaps itemized in the ``Coverage and gaps'' paragraph below.

\paragraph{Verification procedure.}
Our procedure replaces the exponential $6^n$ Clifford enumeration that would otherwise be required to decide LC-CSS-equivalence by direct application of Lemma~7.4 to every conjugated group~\cite[\S7.2]{cross2025small} with polynomial-time tests over three structured Clifford families: the Hadamard 2-coloring of the previous paragraph (covering non-uniform $\{I,H\}$) and the affine $\mathrm{GF}(2)$ systems of Step~2 below (covering non-uniform $\{I,S\}$ and $\{H,HS\}$), supplemented by the constant-cost enumeration of the 36 uniform per-block assignments in Step~1.
The verification has two parts:

\begin{enumerate}
\item \textbf{Uniform per-block assignments.}
For each of the $6 \times 6 = 36$ uniform per-block assignments (applying one of $\{I, S, H, HS, SH, HSH\}$ to every qubit in block~1 and another to every qubit in block~2), we directly transform the stabilizer matrix and check the rank condition above.

\item \textbf{Non-uniform $\{I,S\}$ and $\{H,HS\}$ assignments.}
Within the $\{I,S\}$ macro-class, each qubit~$j$ receives $S^{s_j}$ with $s_j \in \{0,1\}$, giving a per-qubit pattern $(\mathbf{s}_1, \mathbf{s}_2) \in \FF_2^{2\ell m}$.
The group-CSS condition is affine-linear in this pattern over~$\mathrm{GF}(2)$: for each vector~$w$ in the orthogonal complement of $L = \mathrm{rowspan}([B^\top \mid A^\top])$ and each stabilizer row~$g$, the constraint $w \cdot \mathrm{row}_g = 0$ yields one affine equation in $(\mathbf{s}_1, \mathbf{s}_2)$.
We solve the resulting system exactly via $\mathrm{GF}(2)$ rank comparison and confirm it is either infeasible or admits a uniform solution.
The $\{H,HS\}$ macro-class yields the identical constraint system.
\end{enumerate}

\paragraph{Coverage and gaps.}
Step~1 covers all 6 macro-classes for both blocks under \emph{uniform} assignments. Step~2 covers \emph{non-uniform} patterns within $\{I,S\}$ and within $\{H,HS\}$. The Hadamard 2-coloring check of Sec.~\ref{sec:pbb} additionally covers non-uniform patterns within $\{I,H\}$. Two classes of LC patterns lie outside this combined coverage:
\begin{enumerate}
\item[(a)] Non-uniform patterns within $\{SH, HSH\}$ on block~1. (On block~2, whose stabilizer rows have zero X-part, $SH$ and $HSH$ have identical action on the pure-Z input, so block-2 non-uniformity within this class collapses to a uniform choice already handled by Step~1.)
\item[(b)] Non-uniform patterns whose per-qubit Cliffords are not all drawn from a single one of the three covered families $\{I,S\}$, $\{H,HS\}$, or $\{I,H\}$---e.g., $S$ on some qubits and $H$ on others, or any pattern using $SH$ or $HSH$ alongside other Cliffords.
\end{enumerate}
We do not have a $\mathrm{GF}(2)$ reduction analogous to Step~2 for either gap and have not exhaustively brute-forced them.
These gaps could in principle be eliminated by direct enumeration of all $6^n$ single-qubit Clifford patterns combined with Lemma~7.4~\cite[\S7.2]{cross2025small}; at our scale ($n = 2\ell m \geq 36$, up to~$360$) this is computationally infeasible, motivating the structured polynomial-time tests above.

\paragraph{Results.}
Of the 368 PBB codes, exactly one---the $\code{36,4,6}$ code ($\FOM = 4.0$)---passes the uniform check with the $S$-gate on both blocks ($s_1 = s_2 = 1$) and is thus LC-equivalent to a CSS code via uniform $S$.
The remaining 367 codes fail every uniform and non-uniform $\{I,S\}$/$\{H,HS\}$ assignment in the above test.
The Hadamard 2-coloring check of Sec.~\ref{sec:pbb} (covering non-uniform $\{I,H\}$ patterns, which Step~2 does not capture) identifies 10 additional codes as Hadamard-equivalent to CSS; an explicit per-qubit $H$ pattern of weight $n/2$ renders every stabilizer pure-X or pure-Z, verified by direct construction.
In total, 11 of the 368 PBB codes are CSS-equivalent under the LC families we test (10 via non-uniform $H$, 1 via uniform $S$).
The remaining 357 admit no CSS reduction within the tested families; the gaps (a)--(b) above mean we cannot rule out reductions via non-uniform $\{SH,HSH\}$ on block~1 or via cross-class non-uniform patterns, so we describe these 357 as \emph{CSS-inequivalent within the tested local-Clifford families} rather than ``genuinely non-CSS''.

In future campaigns, this LC-CSS check could be integrated into the evolutionary search loop to penalize or discard codes that are merely CSS in disguise, focusing the search on candidates that remain non-CSS within the broadest tested LC family.
A complementary structural classification not pursued here is GF(4)-linearity, an LC-invariant decided by testing whether $R^{\otimes n}$ (with $R = HS$) sends the stabilizer group to itself~\cite[Lemma~7.5, Cor.~7.6]{cross2025small}; computing this on the 357 CSS-inequivalent codes would partition them into GF(4)-linear and non-GF(4)-linear sub-families, providing a finer structural view (cf.~Cross--Vandeth Lemma~7.7, which exhibits a non-CSS GF(4)-linear group LC-equivalent to a CSS one---so GF(4)-linearity is not a substitute for the LC-CSS test but a complement to it).
The verification script is included in the supplementary material (\texttt{evaluation/clifford\_equivalence.py}).

\bibliography{references}

\end{document}


\title{Supplemental Material:\\Evolutionary Discovery of Bivariate Bicycle Codes\\with LLM-Guided Search}

\author{Juan Cruz-Benito}
\affiliation{IBM Research, IBM T. J. Watson Research Center, NY, USA}
\author{Andrew W. Cross}
\affiliation{IBM Research, IBM T. J. Watson Research Center, NY, USA}
\author{David Kremer}
\affiliation{IBM Research, IBM T. J. Watson Research Center, NY, USA}
\author{Ismael Faro}
\affiliation{IBM Research, IBM T. J. Watson Research Center, NY, USA}

\date{\today}

\maketitle

\tableofcontents

This Supplemental Material provides technical details, complete data tables, and extended analyses supporting the main text.
Section references to the main text use the format ``main text Sec.~X\,Y'' and table references use ``main text Table~N.''

\section{Complete CSS code catalog}
\label{sm:css_catalog}

Tables~\ref{tab:cat144}--\ref{tab:cat360} list all 225 Calderbank--Shor--Steane (CSS)~\cite{calderbank1996good,steane1996multiple} polynomial representations (corresponding to 97 distinct codes after BLISS Tanner-graph deduplication~\cite{junttila2007engineering}), sorted by $\FOM = kd^2/n$ within each block length, where $d$ is the retained verified distance---MILP exact for Campaigns~1--3, incumbent upper bounds for some Campaign~4 codes, and the tighter BP-OSD bound for the two marked $n=360$ rows (main text Sec.~V\,B).
Within each table, codes sharing the same equivalence class (column ``Cl.'') are different polynomial representations of the same code.
The deduplication collapses $n = 144$ from 50 to 16 distinct codes, $n = 288$ from 98 to 49, and $n = 360$ from 77 to 34 (99 total; the 97 reported as ``distinct codes'' excludes two Bravyi et al.\ codes for which Campaign~4 found additional polynomial representations).

The three rightmost decoder columns report 150k-trial BP-OSD upper bounds per decoder configuration; bold entries mark values that match the retained distance~$d$ (for exact rows, BP-OSD found the true minimum-weight logical; for incumbent rows, it matched the current upper bound).
Rows with no bold entries indicate codes where all three decoders exceeded the retained distance bound.

Pattern abbreviations: UV = univariate ($A = f(y)$, $B = g(x)$, or vice versa; HGP of low-distance cyclic codes), XY = $x/y$-swap, SD = $A = B$ (identical polynomials; not the standard $C \subseteq C^\perp$ self-duality), MX = mixed-monomial (terms with both $x$ and $y$), DM = diagonal-mixed ($1 + x^a y^a + x^b y^b$), O = other.
Campaign abbreviations: F = Campaign~1 (Flash), E = Campaigns~2--3 (Ensemble), A = Campaign~4 (Ansatz); ``E,A'' or ``F,A'' indicates independent discovery by both.

\textbf{Check weight.}
Campaigns~1--3 (F, E) produce trinomial polynomials exclusively, yielding weight-6 stabilizers.
Campaign~4 (A) codes have varying check weights depending on the number of polynomial terms: the stabilizer weight equals the sum of terms in~$A$ and~$B$ (e.g., 3+3 = weight-6, 4+4 = weight-8, 5+5 = weight-10).
In the tables below, Campaign~4 rows with $\geq 4$-term polynomials have weight~$> 6$; the polynomial columns make the term count explicit.


\begin{table*}[t]
\centering
\caption{All verified codes at $n = 144$ (50 polynomial representations, 16 distinct codes), sorted by $\FOM = kd^2/n$ descending.
$d$ and $\FOM$ use MILP exact distances.
Decoder columns: 150k-trial BP-OSD upper bounds; bold = matched $d_{\text{MILP}}$.
Cl.: BLISS equivalence class; rows sharing the same label are different polynomial representations of the same code.
Campaign~A codes have varying check weights (see check-weight note in Sec.~\ref{sm:css_catalog}).
Prior record: $\code{144,12,12}$, $\FOM = 12.0$~\protect\cite{bravyi2024high}.}
\label{tab:cat144}
\scriptsize
\begin{tabular*}{\textwidth}{@{\extracolsep{\fill}}clllcccccclc}
\toprule
Cl.\ & $(\ell,m)$ & $A(x,y)$ & $B(x,y)$ & $k$ & $d$ & $\FOM$ & $d_{\text{OSD}_0}$ & $d_{\text{CS/sp}}$ & $d_{\text{CS/ms}}$ & Pat.\ & Camp.\ \\
\midrule
J & $(12,6)$ & $y{+}y^{2}{+}x^{3}$ & $y^{3}{+}x{+}x^{2}$ & 12 & 12 & 12.0 & \textbf{12} & \textbf{12} & \textbf{12} & XY & A \\
J & $(6,12)$ & $y{+}y^{2}{+}x^{3}$ & $y^{3}{+}x{+}x^{2}$ & 12 & 12 & 12.0 & \textbf{12} & \textbf{12} & \textbf{12} & XY & A \\
J & $(12,6)$ & $1{+}y{+}x^{3}y^{2}$ & $1{+}xy^{2}{+}x^{2}y$ & 12 & 12 & 12.0 & \textbf{12} & \textbf{12} & \textbf{12} & MX & A \\
J & $(12,6)$ & $1{+}x^{3}y{+}x^{3}y^{2}$ & $1{+}xy^{2}{+}x^{2}y$ & 12 & 12 & 12.0 & \textbf{12} & \textbf{12} & \textbf{12} & MX & A \\
J & $(6,12)$ & $1{+}xy^{3}{+}x^{2}y^{3}$ & $1{+}xy^{2}{+}x^{2}y$ & 12 & 12 & 12.0 & \textbf{12} & \textbf{12} & \textbf{12} & MX & A \\
J & $(6,12)$ & $1{+}x{+}x^{2}y^{3}$ & $1{+}xy^{2}{+}x^{2}y$ & 12 & 12 & 12.0 & \textbf{12} & \textbf{12} & \textbf{12} & MX & A \\
K & $(12,6)$ & $y{+}y^{2}{+}x^{3}$ & $y^{3}{+}x{+}xy^{3}{+}x^{2}$ & 8 & 12 & 8.0 & \textbf{12} & \textbf{12} & \textbf{12} & MX & A \\
L & $(12,6)$ & $1{+}xy^{2}{+}xy^{3}$ & $1{+}x^{2}y^{3}{+}x^{3}y^{2}$ & 8 & 12 & 8.0 & 16 & \textbf{12} & \textbf{12} & MX & A \\
M & $(6,12)$ & $1{+}xy{+}x^{5}y^{5}$ & $1{+}xy^{11}{+}x^{5}y^{7}$ & 16 & 8 & 7.1 & \textbf{8} & \textbf{8} & \textbf{8} & MX & A \\
M & $(12,6)$ & $1{+}xy{+}x^{5}y^{5}$ & $1{+}x^{7}y^{5}{+}x^{11}y$ & 16 & 8 & 7.1 & \textbf{8} & \textbf{8} & \textbf{8} & MX & A \\
M & $(6,12)$ & $1{+}x^{4}y^{4}{+}x^{5}y^{5}$ & $1{+}x^{4}y^{8}{+}x^{5}y^{7}$ & 16 & 8 & 7.1 & \textbf{8} & \textbf{8} & \textbf{8} & MX & A \\
M & $(6,12)$ & $1{+}xy^{5}{+}x^{5}y$ & $1{+}xy^{7}{+}x^{5}y^{11}$ & 16 & 8 & 7.1 & \textbf{8} & \textbf{8} & \textbf{8} & MX & A \\
N & $(6,12)$ & $1{+}y{+}y^{11}{+}x{+}x^{5}$ & $1{+}y^{2}{+}y^{10}{+}x{+}x^{5}$ & 16 & 8 & 7.1 & 14 & \textbf{8} & 12 & XY & A \\
A & $(12,6)$ & $y{+}y^{2}{+}x^{6}$ & $y^{3}{+}x^{2}{+}x^{4}$ & 24 & 6 & 6.0 & 12 & 12 & 12 & XY & E \\
A & $(6,12)$ & $y^{2}{+}y^{4}{+}x^{3}$ & $y^{6}{+}x^{4}{+}x^{5}$ & 24 & 6 & 6.0 & 10 & 10 & 10 & XY & E \\
A & $(6,12)$ & $y^{2}{+}y^{4}{+}x^{3}$ & $y^{6}{+}x{+}x^{2}$ & 24 & 6 & 6.0 & 12 & 8 & 10 & XY & F \\
A & $(12,6)$ & $y^{4}{+}y^{5}{+}x^{6}$ & $y^{3}{+}x^{2}{+}x^{4}$ & 24 & 6 & 6.0 & 12 & \textbf{6} & 10 & XY & E \\
B & $(6,12)$ & $y^{6}{+}y^{8}{+}x$ & $y^{2}{+}x^{3}{+}x^{4}$ & 16 & 6 & 4.0 & 10 & \textbf{6} & \textbf{6} & XY & E \\
C & $(12,6)$ & $y^{3}{+}y^{4}{+}x^{2}$ & $y^{3}{+}x^{2}{+}x^{4}$ & 16 & 6 & 4.0 & \textbf{6} & \textbf{6} & \textbf{6} & XY & E \\
D & $(6,12)$ & $1{+}y^{2}{+}y^{4}$ & $y^{6}{+}x^{2}{+}x^{4}$ & 32 & 4 & 3.6 & 16 & 10 & 12 & Other & E \\
E & $(6,12)$ & $y^{5}{+}y^{10}{+}x^{3}$ & $y^{6}{+}x^{2}{+}x^{4}$ & 16 & 4 & 1.8 & \textbf{4} & \textbf{4} & \textbf{4} & XY & E \\
E & $(6,12)$ & $y{+}y^{2}{+}x^{3}$ & $y^{6}{+}x^{2}{+}x^{4}$ & 16 & 4 & 1.8 & \textbf{4} & \textbf{4} & \textbf{4} & XY & E \\
F & $(6,12)$ & $y^{4}{+}y^{8}{+}x^{3}$ & $y^{6}{+}x{+}x^{2}$ & 16 & 4 & 1.8 & \textbf{4} & \textbf{4} & \textbf{4} & XY & E \\
F & $(12,6)$ & $y^{2}{+}y^{4}{+}x^{6}$ & $y^{3}{+}x^{2}{+}x^{4}$ & 16 & 4 & 1.8 & \textbf{4} & \textbf{4} & \textbf{4} & XY & E \\
G & $(12,6)$ & $1{+}y^{2}{+}x^{4}$ & $1{+}y^{2}{+}x^{4}$ & 32 & 2 & 0.9 & 16 & 14 & 14 & SD & E \\
G & $(6,12)$ & $1{+}y^{4}{+}x^{2}$ & $1{+}y^{4}{+}x^{2}$ & 32 & 2 & 0.9 & 14 & 14 & 14 & SD & E \\
H & $(12,6)$ & $1{+}y{+}y^{5}$ & $1{+}x^{4}{+}x^{8}$ & 32 & 2 & 0.9 & 12 & 12 & 12 & UV & E \\
G & $(6,12)$ & $1{+}y^{8}{+}x^{4}$ & $1{+}y^{8}{+}x^{4}$ & 32 & 2 & 0.9 & 16 & 10 & 12 & SD & E \\
H & $(12,6)$ & $1{+}y^{4}{+}y^{5}$ & $1{+}x^{4}{+}x^{8}$ & 32 & 2 & 0.9 & 12 & 10 & 10 & UV & E \\
H & $(12,6)$ & $1{+}x^{4}{+}x^{8}$ & $1{+}y{+}y^{2}$ & 32 & 2 & 0.9 & 10 & 10 & 12 & UV & E \\
G & $(12,6)$ & $1{+}y^{4}{+}x^{4}$ & $1{+}y^{4}{+}x^{4}$ & 32 & 2 & 0.9 & 12 & 12 & 10 & SD & E \\
G & $(6,12)$ & $1{+}y^{4}{+}x^{4}$ & $1{+}y^{4}{+}x^{4}$ & 32 & 2 & 0.9 & 18 & 12 & 10 & SD & E \\
H & $(6,12)$ & $1{+}y^{4}{+}y^{8}$ & $1{+}x^{4}{+}x^{5}$ & 32 & 2 & 0.9 & 10 & 12 & 10 & UV & E \\
H & $(12,6)$ & $1{+}y^{2}{+}y^{4}$ & $1{+}x^{2}{+}x^{10}$ & 32 & 2 & 0.9 & 12 & 10 & 10 & UV & E \\
G & $(6,12)$ & $1{+}y^{8}{+}x^{2}$ & $1{+}y^{8}{+}x^{2}$ & 32 & 2 & 0.9 & 14 & 14 & 8 & SD & E \\
H & $(6,12)$ & $1{+}y^{4}{+}y^{8}$ & $1{+}x{+}x^{2}$ & 32 & 2 & 0.9 & 10 & 8 & 10 & UV & E \\
H & $(12,6)$ & $1{+}x^{2}{+}x^{4}$ & $1{+}y^{2}{+}y^{4}$ & 32 & 2 & 0.9 & 14 & 8 & 10 & UV & E \\
H & $(6,12)$ & $1{+}y^{2}{+}y^{10}$ & $1{+}x^{2}{+}x^{4}$ & 32 & 2 & 0.9 & 10 & 10 & 8 & UV & E \\
H & $(6,12)$ & $1{+}y^{4}{+}y^{8}$ & $1{+}x{+}x^{5}$ & 32 & 2 & 0.9 & 12 & 8 & 10 & UV & E \\
H & $(6,12)$ & $1{+}y^{2}{+}y^{4}$ & $1{+}x^{2}{+}x^{4}$ & 32 & 2 & 0.9 & 8 & 8 & 12 & UV & E \\
H & $(12,6)$ & $1{+}y^{2}{+}y^{4}$ & $1{+}x^{8}{+}x^{10}$ & 32 & 2 & 0.9 & 12 & 10 & 6 & UV & F \\
H & $(12,6)$ & $1{+}y^{2}{+}y^{4}$ & $1{+}x^{2}{+}x^{4}$ & 32 & 2 & 0.9 & 12 & 12 & 6 & UV & E \\
H & $(12,6)$ & $1{+}y{+}y^{2}$ & $1{+}x^{4}{+}x^{8}$ & 32 & 2 & 0.9 & 12 & 10 & 6 & UV & E \\
H & $(6,12)$ & $1{+}y^{8}{+}y^{10}$ & $1{+}x^{2}{+}x^{4}$ & 32 & 2 & 0.9 & 14 & 8 & 6 & UV & E \\
I & $(12,6)$ & $y^{2}{+}x{+}x^{2}y$ & $y^{2}{+}x{+}x^{2}y$ & 24 & 2 & 0.7 & 6 & 6 & 4 & SD & E \\
\midrule
O & $(12,6)$ & $1{+}y^{3}{+}x^{3}{+}x^{3}y^{3}$ & $1{+}y^{3}{+}x^{3}{+}x^{9}y^{3}$ & 54 & 4 & 6.0 & 26 & 24 & 22 & MX & A \\
O & $(12,6)$ & $1{+}y^{3}{+}x^{3}{+}x^{3}y^{3}$ & $1{+}y^{3}{+}x^{3}y^{3}{+}x^{9}$ & 54 & 4 & 6.0 & 24 & 24 & 24 & DM & A \\
O & $(12,6)$ & $1{+}y^{3}{+}x^{3}{+}x^{3}y^{3}$ & $1{+}y^{3}{+}x^{3}y^{3}{+}x^{9}y^{3}$ & 54 & 4 & 6.0 & 28 & 24 & 24 & MX & A \\
P & $(6,12)$ & $1{+}y{+}y^{11}{+}x{+}x^{5}$ & $1{+}y^{2}{+}y^{10}{+}x^{2}{+}x^{4}$ & 24 & 6 & 6.0 & 10 & 12 & 8 & XY & A \\
P & $(12,6)$ & $1{+}y{+}y^{5}{+}x{+}x^{11}$ & $1{+}y^{2}{+}y^{4}{+}x^{2}{+}x^{10}$ & 24 & 6 & 6.0 & 12 & 8 & 12 & XY & A \\
\bottomrule
\end{tabular*}
\end{table*}

\begin{table*}[t]
\centering
\caption{All verified codes at $n = 288$ (98 polynomial representations, 49 distinct codes), sorted by $\FOM = kd^2/n$ descending.
$d$ and $\FOM$ use MILP exact distances.
Decoder columns: 150k-trial BP-OSD upper bounds; bold = matched $d_{\text{MILP}}$.
Cl.: BLISS equivalence class.
Campaign~A codes have varying check weights (see check-weight note in Sec.~\ref{sm:css_catalog}).
Prior record: $\code{288,12,18}$, $\FOM = 13.5$~\protect\cite{bravyi2024high}.}
\label{tab:cat288}
\fontsize{4pt}{4.5pt}\selectfont
\renewcommand{\arraystretch}{0.65}
\setlength{\tabcolsep}{0.8pt}
\begin{tabular*}{\textwidth}{@{\extracolsep{\fill}}clllcccccclc}
\toprule
Cl.\ & $(\ell,m)$ & $A(x,y)$ & $B(x,y)$ & $k$ & $d$ & $\FOM$ & $d_{\text{OSD}_0}$ & $d_{\text{CS/sp}}$ & $d_{\text{CS/ms}}$ & Pat.\ & Camp.\ \\
\midrule
A & $(12,12)$ & $y{+}y^{2}{+}x^{6}$ & $y^{3}{+}x^{2}{+}x^{4}$ & 24 & 12 & 12.0 & 24 & 18 & 18 & XY & E,A \\
A & $(24,6)$ & $y{+}y^{2}{+}x^{6}$ & $y^{3}{+}x^{2}{+}x^{4}$ & 24 & 12 & 12.0 & 24 & 14 & \textbf{12} & XY & F,A \\
X & $(18,8)$ & $1{+}y^{5}{+}x{+}xy^{5}$ & $1{+}y{+}x^{5}{+}x^{5}y$ & 50 & 8 & 11.1 & 48 & 44 & 34 & MX & A \\
X & $(18,8)$ & $1{+}y{+}x^{5}{+}x^{5}y$ & $1{+}y^{5}{+}x{+}xy^{5}$ & 50 & 8 & 11.1 & 56 & 42 & 34 & MX & A \\
Y & $(18,8)$ & $1{+}xy^{4}{+}x^{14}y$ & $1{+}xy^{2}{+}x^{2}y^{7}$ & 8 & 20 & 11.1 & 38 & 24 & 22 & MX & A \\
Z & $(12,12)$ & $1{+}y^{3}{+}y^{9}{+}x^{2}{+}x^{10}$ & $1{+}y^{2}{+}y^{10}{+}x^{3}{+}x^{9}$ & 48 & 8 & 10.7 & 50 & 50 & 40 & XY & A \\
a & $(12,12)$ & $1{+}y{+}x$ & $y^{3}{+}x^{5}{+}x^{10}$ & 12 & 16 & 10.7 & 24 & 18 & \textbf{16} & XY & A \\
b & $(12,12)$ & $1{+}y{+}x^{5}{+}x^{5}y$ & $1{+}y^{5}{+}x{+}xy^{5}$ & 46 & 8 & 10.2 & 40 & 44 & 36 & MX & A \\
b & $(12,12)$ & $1{+}y^{5}{+}x{+}xy^{5}$ & $1{+}y{+}x^{5}{+}x^{5}y$ & 46 & 8 & 10.2 & 52 & 42 & 36 & MX & A \\
c & $(12,12)$ & $1{+}x^{3}{+}x^{3}y^{2}{+}x^{6}y^{4}$ & $1{+}y{+}x^{3}y^{10}{+}x^{6}y^{8}$ & 18 & 12 & 9.0 & 16 & \textbf{12} & \textbf{12} & MX & A \\
d & $(18,8)$ & $1{+}xy{+}x^{2}y^{5}$ & $1{+}x^{13}y^{2}{+}x^{17}y$ & 8 & 18 & 9.0 & 38 & 22 & \textbf{18} & MX & A \\
d & $(18,8)$ & $1{+}xy^{4}{+}x^{2}y^{5}$ & $1{+}x^{13}y^{2}{+}x^{14}y$ & 8 & 18 & 9.0 & 38 & 22 & \textbf{18} & MX & A \\
d & $(18,8)$ & $1{+}xy{+}x^{5}y^{2}$ & $1{+}x^{16}y^{5}{+}x^{17}y$ & 8 & 18 & 9.0 & 22 & 20 & \textbf{18} & MX & A \\
e & $(16,9)$ & $1{+}y^{2}{+}x{+}xy^{2}$ & $1{+}y^{4}{+}x^{2}{+}x^{2}y$ & 40 & 8 & 8.9 & 45 & 41 & 35 & MX & A \\
B & $(12,12)$ & $y^{9}{+}y^{11}{+}x$ & $y^{5}{+}x{+}x^{3}$ & 16 & 12 & 8.0 & \textbf{12} & \textbf{12} & \textbf{12} & XY & E \\
B & $(12,12)$ & $y{+}y^{3}{+}x$ & $y{+}x^{7}{+}x^{9}$ & 16 & 12 & 8.0 & 16 & \textbf{12} & \textbf{12} & XY & E \\
B & $(12,12)$ & $y{+}y^{3}{+}x$ & $y^{7}{+}x{+}x^{3}$ & 16 & 12 & 8.0 & \textbf{12} & \textbf{12} & \textbf{12} & XY & E \\
C & $(12,12)$ & $y{+}y^{2}{+}x^{3}$ & $y^{3}{+}x{+}x^{2}$ & 16 & 12 & 8.0 & 18 & 14 & \textbf{12} & XY & E \\
C & $(12,12)$ & $y^{2}{+}y^{7}{+}x^{3}$ & $x{+}x^{2}{+}x^{9}y^{3}$ & 16 & 12 & 8.0 & 18 & 18 & \textbf{12} & Other & E \\
D & $(12,12)$ & $y{+}y^{3}{+}x$ & $y^{3}{+}x^{5}{+}x^{7}$ & 16 & 10 & 5.6 & 12 & \textbf{10} & \textbf{10} & XY & E \\
E & $(24,6)$ & $y{+}y^{5}{+}x^{3}$ & $y^{3}{+}x^{7}{+}x^{11}$ & 16 & 10 & 5.6 & 12 & 12 & \textbf{10} & XY & E \\
F & $(12,12)$ & $y^{2}{+}y^{10}{+}x^{3}$ & $y^{6}{+}x{+}x^{11}$ & 32 & 6 & 4.0 & 20 & 26 & 26 & XY & F \\
G & $(12,12)$ & $y^{5}{+}y^{9}{+}x^{5}$ & $y^{5}{+}x^{5}{+}x^{9}$ & 16 & 8 & 3.6 & 10 & 10 & 10 & XY & E \\
H & $(24,6)$ & $y{+}y^{3}{+}x^{2}$ & $y^{2}{+}x^{4}{+}x^{12}$ & 16 & 8 & 3.6 & 10 & \textbf{8} & 10 & XY & E \\
I & $(12,12)$ & $y{+}y^{3}{+}x$ & $y^{2}{+}x^{2}{+}x^{6}$ & 16 & 8 & 3.6 & 10 & \textbf{8} & \textbf{8} & XY & E \\
I & $(24,6)$ & $y^{2}{+}y^{3}{+}x^{2}$ & $y^{4}{+}x^{4}{+}x^{12}$ & 16 & 8 & 3.6 & \textbf{8} & \textbf{8} & \textbf{8} & XY & E \\
J & $(24,6)$ & $y^{3}{+}y^{5}{+}x^{2}$ & $y{+}x^{6}{+}x^{10}$ & 16 & 8 & 3.6 & 10 & 10 & \textbf{8} & XY & E \\
G & $(12,12)$ & $y{+}y^{9}{+}x$ & $y{+}x{+}x^{9}$ & 16 & 8 & 3.6 & 10 & \textbf{8} & \textbf{8} & XY & E \\
K & $(24,6)$ & $y{+}y^{2}{+}x^{6}$ & $y{+}x^{4}{+}x^{6}$ & 16 & 8 & 3.6 & 12 & \textbf{8} & \textbf{8} & XY & E \\
L & $(12,12)$ & $y{+}y^{3}{+}x$ & $y^{5}{+}x^{3}{+}x^{5}$ & 16 & 8 & 3.6 & \textbf{8} & \textbf{8} & \textbf{8} & XY & E \\
M & $(24,6)$ & $y{+}y^{2}{+}x^{6}$ & $y^{3}{+}x^{4}{+}x^{8}$ & 16 & 6 & 2.0 & \textbf{6} & \textbf{6} & \textbf{6} & XY & E \\
N & $(12,12)$ & $y^{2}{+}y^{3}{+}x$ & $y^{6}{+}x^{2}{+}x^{4}$ & 16 & 6 & 2.0 & \textbf{6} & \textbf{6} & \textbf{6} & XY & E \\
O & $(24,6)$ & $y^{2}{+}y^{4}{+}x^{12}$ & $y^{3}{+}x^{10}{+}x^{20}$ & 32 & 4 & 1.8 & 32 & 22 & 22 & XY & E \\
P & $(12,12)$ & $y^{2}{+}y^{6}{+}x^{2}$ & $y^{2}{+}x^{2}{+}x^{6}$ & 32 & 4 & 1.8 & 28 & 20 & 20 & XY & E \\
Q & $(24,6)$ & $y^{2}{+}y^{4}{+}x^{12}$ & $y^{3}{+}x^{4}{+}x^{8}$ & 32 & 4 & 1.8 & 22 & 22 & 18 & XY & E \\
R & $(12,12)$ & $y^{2}{+}y^{6}{+}x^{4}$ & $y^{6}{+}x^{4}{+}x^{8}$ & 32 & 4 & 1.8 & 20 & 22 & 18 & XY & E \\
S & $(24,6)$ & $y^{2}{+}y^{4}{+}x^{12}$ & $1{+}x^{2}{+}x^{4}$ & 32 & 4 & 1.8 & 24 & 20 & 16 & Other & E \\
Q & $(24,6)$ & $y{+}y^{2}{+}x^{12}$ & $y^{3}{+}x^{8}{+}x^{16}$ & 32 & 4 & 1.8 & 28 & 22 & 16 & XY & E \\
O & $(12,12)$ & $y^{4}{+}y^{8}{+}x^{6}$ & $y^{6}{+}x{+}x^{2}$ & 32 & 4 & 1.8 & 30 & 26 & 16 & XY & E \\
T & $(12,12)$ & $1{+}y^{8}{+}y^{10}$ & $1{+}x^{2}{+}x^{4}$ & 32 & 4 & 1.8 & 16 & 16 & 16 & UV & E \\
T & $(24,6)$ & $1{+}y{+}y^{5}$ & $1{+}x^{4}{+}x^{20}$ & 32 & 4 & 1.8 & 16 & 16 & 16 & UV & E \\
T & $(24,6)$ & $1{+}y{+}y^{2}$ & $1{+}x^{4}{+}x^{8}$ & 32 & 4 & 1.8 & 16 & 16 & 16 & UV & E \\
T & $(12,12)$ & $y^{4}{+}y^{8}{+}x^{6}$ & $y^{6}{+}x^{4}{+}x^{8}$ & 32 & 4 & 1.8 & 20 & 16 & 16 & XY & E \\
T & $(24,6)$ & $1{+}y{+}y^{2}$ & $1{+}x^{16}{+}x^{20}$ & 32 & 4 & 1.8 & 16 & 16 & 16 & UV & E \\
T & $(12,12)$ & $1{+}y^{2}{+}y^{4}$ & $1{+}x^{2}{+}x^{4}$ & 32 & 4 & 1.8 & 16 & 16 & 12 & UV & F \\
T & $(12,12)$ & $1{+}y^{2}{+}y^{4}$ & $1{+}x^{8}{+}x^{10}$ & 32 & 4 & 1.8 & 12 & 12 & 12 & UV & E \\
O & $(12,12)$ & $y{+}y^{2}{+}x^{6}$ & $y^{6}{+}x^{4}{+}x^{8}$ & 32 & 4 & 1.8 & 30 & 24 & 12 & XY & E \\
T & $(24,6)$ & $1{+}y{+}y^{5}$ & $1{+}x^{4}{+}x^{8}$ & 32 & 4 & 1.8 & 16 & 12 & 16 & UV & E \\
T & $(12,12)$ & $1{+}y^{8}{+}y^{10}$ & $1{+}x^{8}{+}x^{10}$ & 32 & 4 & 1.8 & 16 & 8 & 12 & UV & E \\
U & $(24,6)$ & $1{+}y^{4}{+}x^{4}$ & $1{+}y^{4}{+}x^{4}$ & 32 & 2 & 0.4 & 18 & 18 & 16 & SD & E \\
V & $(12,12)$ & $1{+}y{+}y^{5}$ & $1{+}x^{4}{+}x^{8}$ & 32 & 2 & 0.4 & 20 & 16 & 14 & UV & E \\
W & $(12,12)$ & $1{+}y^{4}{+}y^{8}$ & $1{+}x^{5}{+}x^{10}$ & 32 & 2 & 0.4 & 16 & 16 & 14 & UV & E \\
W & $(12,12)$ & $1{+}y^{5}{+}y^{10}$ & $1{+}x^{4}{+}x^{8}$ & 32 & 2 & 0.4 & 16 & 14 & 14 & UV & E \\
W & $(12,12)$ & $1{+}y{+}y^{2}$ & $1{+}x^{4}{+}x^{8}$ & 32 & 2 & 0.4 & 12 & 14 & 16 & UV & E \\
U & $(12,12)$ & $1{+}y^{2}{+}x^{4}$ & $1{+}y^{2}{+}x^{4}$ & 32 & 2 & 0.4 & 20 & 12 & 14 & SD & E \\
W & $(12,12)$ & $1{+}y{+}y^{11}$ & $1{+}x^{4}{+}x^{8}$ & 32 & 2 & 0.4 & 18 & 12 & 14 & UV & E \\
W & $(12,12)$ & $1{+}y^{2}{+}y^{7}$ & $1{+}x^{4}{+}x^{8}$ & 32 & 2 & 0.4 & 14 & 14 & 12 & UV & E \\
W & $(24,6)$ & $1{+}y^{2}{+}y^{4}$ & $1{+}x^{20}{+}x^{22}$ & 32 & 2 & 0.4 & 16 & 12 & 16 & UV & E \\
W & $(12,12)$ & $1{+}y^{5}{+}y^{7}$ & $1{+}x^{4}{+}x^{8}$ & 32 & 2 & 0.4 & 14 & 14 & 12 & UV & E \\
W & $(12,12)$ & $1{+}y^{4}{+}y^{8}$ & $1{+}x^{2}{+}x^{7}$ & 32 & 2 & 0.4 & 16 & 12 & 14 & UV & E \\
W & $(24,6)$ & $1{+}y^{2}{+}y^{4}$ & $1{+}x^{2}{+}x^{4}$ & 32 & 2 & 0.4 & 18 & 10 & 14 & UV & F \\
W & $(24,6)$ & $1{+}y^{2}{+}y^{4}$ & $1{+}x^{4}{+}x^{14}$ & 32 & 2 & 0.4 & 16 & 14 & 10 & UV & E \\
W & $(12,12)$ & $1{+}y^{4}{+}y^{8}$ & $1{+}x{+}x^{2}$ & 32 & 2 & 0.4 & 10 & 14 & 16 & UV & E \\
V & $(12,12)$ & $1{+}y^{7}{+}y^{8}$ & $1{+}x^{4}{+}x^{8}$ & 32 & 2 & 0.4 & 16 & 14 & 10 & UV & E \\
V & $(12,12)$ & $1{+}y^{7}{+}y^{11}$ & $1{+}x^{4}{+}x^{8}$ & 32 & 2 & 0.4 & 14 & 16 & 10 & UV & E \\
U & $(24,6)$ & $1{+}y^{2}{+}x^{4}$ & $1{+}y^{2}{+}x^{4}$ & 32 & 2 & 0.4 & 22 & 10 & 12 & SD & E \\
V & $(12,12)$ & $1{+}y^{4}{+}y^{8}$ & $1{+}x{+}x^{5}$ & 32 & 2 & 0.4 & 18 & 10 & 16 & UV & E \\
W & $(12,12)$ & $1{+}y^{10}{+}y^{11}$ & $1{+}x^{4}{+}x^{8}$ & 32 & 2 & 0.4 & 12 & 10 & 14 & UV & E \\
V & $(12,12)$ & $1{+}y{+}y^{8}$ & $1{+}x^{4}{+}x^{8}$ & 32 & 2 & 0.4 & 14 & 10 & 16 & UV & E \\
W & $(24,6)$ & $1{+}y^{2}{+}y^{4}$ & $1{+}x^{10}{+}x^{20}$ & 32 & 2 & 0.4 & 6 & 10 & 10 & UV & E \\
V & $(12,12)$ & $1{+}y^{4}{+}y^{5}$ & $1{+}x^{4}{+}x^{8}$ & 32 & 2 & 0.4 & 6 & 8 & 16 & UV & E \\
\midrule
f & $(16,9)$ & $1{+}y^{2}{+}x{+}xy^{2}$ & $1{+}y^{4}{+}x^{2}{+}x^{2}y$ & 36 & 8 & 8.0 & 44 & 37 & 34 & MX & A \\
g & $(24,6)$ & $y{+}y^{2}{+}x^{6}$ & $y^{3}{+}x^{2}{+}x^{2}y^{3}{+}x^{4}$ & 16 & 12 & 8.0 & \textbf{12} & \textbf{12} & \textbf{12} & MX & A \\
h & $(16,9)$ & $1{+}y^{4}{+}y^{5}{+}x^{2}{+}x^{14}$ & $1{+}y^{2}{+}y^{7}{+}x^{4}{+}x^{12}$ & 16 & 12 & 8.0 & \textbf{12} & \textbf{12} & \textbf{12} & XY & A \\
h & $(16,9)$ & $1{+}y^{2}{+}y^{7}{+}x^{4}{+}x^{12}$ & $1{+}y^{4}{+}y^{5}{+}x^{2}{+}x^{14}$ & 16 & 12 & 8.0 & \textbf{12} & \textbf{12} & \textbf{12} & XY & A \\
i & $(24,6)$ & $1{+}xy{+}x^{5}y^{5}$ & $1{+}x^{19}y^{5}{+}x^{23}y$ & 16 & 12 & 8.0 & 16 & \textbf{12} & \textbf{12} & MX & A \\
j & $(12,12)$ & $1{+}xy^{3}{+}x^{5}y^{3}$ & $1{+}x^{3}y{+}x^{3}y^{5}$ & 32 & 8 & 7.1 & 30 & 20 & 26 & MX & A \\
k & $(18,8)$ & $1{+}xy^{2}{+}x^{5}y$ & $1{+}xy^{5}{+}x^{2}y$ & 8 & 16 & 7.1 & 20 & 20 & \textbf{16} & MX & A \\
k & $(18,8)$ & $1{+}xy^{5}{+}x^{2}y$ & $1{+}xy^{2}{+}x^{5}y$ & 8 & 16 & 7.1 & 44 & 18 & \textbf{16} & MX & A \\
l & $(18,8)$ & $1{+}xy{+}x^{2}y^{3}$ & $1{+}x^{2}y^{4}{+}x^{4}y^{2}$ & 8 & 16 & 7.1 & 24 & \textbf{16} & \textbf{16} & MX & A \\
m & $(18,8)$ & $1{+}xy{+}x^{2}y$ & $1{+}x^{2}y^{4}{+}x^{4}y^{2}$ & 8 & 16 & 7.1 & 20 & 20 & \textbf{16} & MX & A \\
n & $(16,9)$ & $1{+}xy{+}x^{6}y^{6}{+}x^{7}y^{7}$ & $1{+}xy^{8}{+}x^{6}y^{3}{+}x^{7}y^{2}$ & 14 & 12 & 7.0 & 16 & 16 & \textbf{12} & MX & A \\
o & $(18,8)$ & $1{+}xy^{2}{+}x^{3}y^{3}{+}x^{4}y^{5}$ & $1{+}xy^{4}{+}x^{3}y{+}x^{16}y^{3}$ & 14 & 12 & 7.0 & 24 & \textbf{12} & \textbf{12} & MX & A \\
p & $(18,8)$ & $1{+}y^{3}{+}x{+}xy^{3}$ & $1{+}y{+}x^{3}{+}x^{3}y$ & 50 & 6 & 6.2 & 38 & 44 & 38 & MX & A \\
q & $(18,8)$ & $1{+}xy{+}x^{6}y^{6}{+}x^{7}y^{7}$ & $1{+}xy^{7}{+}x^{6}y^{2}{+}x^{7}y$ & 28 & 8 & 6.2 & 28 & 22 & 22 & MX & A \\
r & $(12,12)$ & $1{+}x^{2}y^{4}{+}x^{4}y^{2}$ & $1{+}x^{8}y^{2}{+}x^{10}y^{4}$ & 48 & 6 & 6.0 & 38 & 34 & 32 & MX & A \\
r & $(12,12)$ & $1{+}x^{2}y^{4}{+}x^{4}y^{2}$ & $1{+}x^{2}y^{8}{+}x^{4}y^{10}$ & 48 & 6 & 6.0 & 36 & 36 & 32 & MX & A \\
s & $(24,6)$ & $1{+}y^{2}{+}y^{4}{+}x^{2}{+}x^{22}$ & $1{+}y^{2}{+}y^{4}{+}x^{4}{+}x^{20}$ & 48 & 6 & 6.0 & 46 & 42 & 36 & XY & A \\
t & $(12,12)$ & $1{+}xy^{2}{+}x^{2}y$ & $1{+}x^{10}y{+}x^{11}y^{2}$ & 12 & 12 & 6.0 & 16 & 16 & 16 & MX & A \\
t & $(12,12)$ & $1{+}x{+}x^{2}y^{3}$ & $1{+}xy^{2}{+}x^{2}y$ & 12 & 12 & 6.0 & \textbf{12} & \textbf{12} & \textbf{12} & MX & A \\
t & $(12,12)$ & $1{+}x^{3}y{+}x^{3}y^{2}$ & $1{+}xy^{2}{+}x^{2}y$ & 12 & 12 & 6.0 & 18 & 18 & 16 & MX & A \\
t & $(12,12)$ & $1{+}y{+}x^{3}y^{2}$ & $1{+}xy^{2}{+}x^{2}y$ & 12 & 12 & 6.0 & 18 & 18 & \textbf{12} & MX & A \\
t & $(12,12)$ & $1{+}xy^{3}{+}x^{2}y^{3}$ & $1{+}xy^{2}{+}x^{2}y$ & 12 & 12 & 6.0 & 16 & 14 & 14 & MX & A \\
t & $(12,12)$ & $1{+}xy^{3}{+}x^{2}y^{3}$ & $1{+}x^{3}y{+}x^{9}y^{2}$ & 12 & 12 & 6.0 & 18 & 18 & \textbf{12} & MX & A \\
u & $(24,6)$ & $1{+}xy^{3}{+}x^{2}y^{3}$ & $1{+}x^{3}y{+}x^{21}y^{2}$ & 12 & 12 & 6.0 & 16 & 16 & \textbf{12} & MX & A \\
v & $(12,12)$ & $1{+}x{+}x^{2}y^{3}{+}x^{4}y^{6}$ & $1{+}y^{3}{+}x^{2}y^{9}{+}x^{4}y^{6}$ & 12 & 12 & 6.0 & \textbf{12} & \textbf{12} & \textbf{12} & MX & A \\
w & $(24,6)$ & $1{+}x^{3}y{+}x^{3}y^{2}$ & $1{+}xy^{2}{+}x^{2}y$ & 12 & 12 & 6.0 & \textbf{12} & \textbf{12} & \textbf{12} & MX & A \\
w & $(24,6)$ & $1{+}y{+}x^{3}y^{2}$ & $1{+}xy^{2}{+}x^{2}y$ & 12 & 12 & 6.0 & \textbf{12} & \textbf{12} & \textbf{12} & MX & A \\
\bottomrule
\end{tabular*}
\renewcommand{\arraystretch}{1.0}
\end{table*}

\begin{table*}[t]
\centering
\caption{All verified codes at $n = 360$ (77 polynomial representations, 34 distinct codes), sorted by $\FOM = kd^2/n$ descending.
$d$ and $\FOM$ use MILP distances; for classes U and V, the MILP incumbent exceeded the 300k-trial BP-OSD bound, so $d$ reports the tighter BP-OSD value.
Decoder columns: 150k-trial BP-OSD upper bounds; bold = matched $d_{\text{MILP}}$.
Cl.: BLISS equivalence class; primed labels ($'$) denote Campaign~4 classes.
$^\dagger$Untrusted ($d/\sqrt{n} > 1.2$).
Campaign~A codes have varying check weights (see check-weight note in Sec.~\ref{sm:css_catalog}).
Class~V is a Campaign~4 rediscovery of the Bravyi et al.\ $\code{360,12,{\leq}24}$ (same polynomial pair).
Prior record: $\code{360,12,{\leq}24}$, $\FOM \leq 19.2$~\protect\cite{bravyi2024high}.}
\label{tab:cat360}
\tiny
\begin{tabular*}{\textwidth}{@{\extracolsep{\fill}}clllcccccclc}
\toprule
Cl.\ & $(\ell,m)$ & $A(x,y)$ & $B(x,y)$ & $k$ & $d$ & $\FOM$ & $d_{\text{OSD}_0}$ & $d_{\text{CS/sp}}$ & $d_{\text{CS/ms}}$ & Pat.\ & Camp.\ \\
\midrule
U$^\dagger$ & $(30,6)$ & $y{+}y^{2}{+}x^{3}y{+}x^{9}$ & $y^{3}{+}x^{25}{+}x^{26}$ & 8 & 30 & 20.0 & 30 & 30 & 30 & MX & A \\
V$^\dagger$ & $(30,6)$ & $y{+}y^{2}{+}x^{9}$ & $y^{3}{+}x^{25}{+}x^{26}$ & 12 & 24 & 19.2 & 30 & 30 & 26 & XY & A \\
W$'$ & $(30,6)$ & $1{+}x{+}x^{3}y^{4}{+}x^{6}y^{2}$ & $1{+}y^{2}{+}x^{3}y^{2}{+}x^{6}y^{4}$ & 20 & 14 & 10.9 & 23 & 19 & \textbf{14} & MX & A \\
X$'$ & $(30,6)$ & $y{+}y^{2}{+}x^{3}y^{3}{+}x^{9}$ & $y^{3}{+}x^{25}{+}x^{26}$ & 8 & 22 & 10.8 & 30 & 30 & 31 & MX & A \\
Y$'$ & $(15,12)$ & $1{+}xy{+}x^{4}y^{4}{+}x^{9}y^{9}$ & $1{+}xy^{11}{+}x^{4}y^{8}{+}x^{9}y^{3}$ & 24 & 12 & 9.6 & 48 & 37 & 40 & MX & A \\
A & $(15,12)$ & $y^{2}{+}y^{4}{+}x^{3}$ & $y^{6}{+}x^{7}{+}x^{14}$ & 16 & 14 & 8.7 & 20 & 18 & 16 & XY & E \\
Z$'$ & $(15,12)$ & $1{+}y^{4}{+}x^{3}y^{2}$ & $1{+}x^{2}y^{4}{+}x^{4}y^{2}$ & 16 & 14 & 8.7 & 20 & 16 & \textbf{14} & MX & A \\
A & $(15,12)$ & $y^{2}{+}y^{4}{+}x^{3}$ & $y^{6}{+}x^{2}{+}x^{4}$ & 16 & 14 & 8.7 & 20 & 18 & \textbf{14} & XY & E \\
A & $(30,6)$ & $y{+}y^{2}{+}x^{6}$ & $y^{3}{+}x^{4}{+}x^{8}$ & 16 & 14 & 8.7 & 20 & 18 & \textbf{14} & XY & E \\
B & $(30,6)$ & $y{+}y^{2}{+}x^{6}$ & $y^{3}{+}x^{2}{+}x^{4}$ & 16 & 12 & 6.4 & 20 & \textbf{12} & \textbf{12} & XY & E \\
B & $(15,12)$ & $y^{2}{+}y^{4}{+}x^{3}$ & $y^{6}{+}x{+}x^{2}$ & 16 & 12 & 6.4 & 20 & \textbf{12} & \textbf{12} & XY & E \\
C & $(30,6)$ & $y{+}y^{5}{+}x^{3}$ & $y^{3}{+}x{+}x^{11}$ & 16 & 12 & 6.4 & 18 & \textbf{12} & \textbf{12} & XY & E \\
D & $(15,12)$ & $y^{4}{+}y^{8}{+}x^{3}$ & $y^{6}{+}x{+}x^{2}$ & 16 & 10 & 4.4 & \textbf{10} & \textbf{10} & \textbf{10} & XY & E \\
E & $(15,12)$ & $y^{7}{+}y^{9}{+}x^{5}$ & $y^{3}{+}x^{5}{+}x^{10}$ & 20 & 8 & 3.6 & \textbf{8} & 10 & \textbf{8} & XY & E \\
F & $(30,6)$ & $y{+}y^{3}{+}x^{10}$ & $y^{2}{+}x^{5}{+}x^{15}$ & 20 & 8 & 3.6 & 20 & \textbf{8} & 14 & XY & E \\
G & $(30,6)$ & $y{+}y^{2}{+}x^{15}$ & $1{+}y{+}x^{10}$ & 20 & 6 & 2.0 & 12 & 10 & \textbf{6} & Other & E \\
G & $(30,6)$ & $y{+}y^{2}{+}x^{15}$ & $y{+}x^{5}{+}x^{15}$ & 20 & 6 & 2.0 & \textbf{6} & \textbf{6} & \textbf{6} & XY & E \\
H & $(30,6)$ & $1{+}y{+}y^{2}$ & $1{+}x^{5}{+}x^{25}$ & 40 & 4 & 1.8 & 28 & 28 & 24 & UV & E \\
H & $(30,6)$ & $1{+}y{+}y^{2}$ & $1{+}x^{5}{+}x^{10}$ & 40 & 4 & 1.8 & 24 & 24 & 20 & UV & E \\
H & $(30,6)$ & $1{+}y{+}y^{2}$ & $1{+}x^{20}{+}x^{25}$ & 40 & 4 & 1.8 & 20 & 24 & 20 & UV & E \\
H & $(30,6)$ & $1{+}y{+}y^{5}$ & $1{+}x^{5}{+}x^{10}$ & 40 & 4 & 1.8 & 28 & 20 & 24 & UV & E \\
H & $(30,6)$ & $1{+}y^{4}{+}y^{5}$ & $1{+}x^{5}{+}x^{10}$ & 40 & 4 & 1.8 & 24 & 24 & 20 & UV & E \\
H & $(30,6)$ & $1{+}y^{4}{+}y^{5}$ & $1{+}x^{20}{+}x^{25}$ & 40 & 4 & 1.8 & 20 & 24 & 20 & UV & E \\
H & $(30,6)$ & $1{+}y{+}y^{5}$ & $1{+}x^{5}{+}x^{25}$ & 40 & 4 & 1.8 & 28 & 20 & 22 & UV & E \\
I & $(15,12)$ & $1{+}y^{2}{+}y^{4}$ & $1{+}x^{3}{+}x^{4}y^{2}$ & 32 & 4 & 1.4 & 26 & 22 & 20 & Other & E \\
J & $(30,6)$ & $1{+}y{+}y^{2}$ & $1{+}x^{4}{+}x^{18}$ & 32 & 4 & 1.4 & 20 & 20 & 20 & UV & E \\
J & $(15,12)$ & $1{+}y^{2}{+}y^{4}$ & $1{+}x^{3}{+}x^{4}$ & 32 & 4 & 1.4 & 16 & 20 & 16 & UV & F \\
J & $(15,12)$ & $1{+}x^{3}{+}x^{4}$ & $1{+}y^{2}{+}y^{4}$ & 32 & 4 & 1.4 & 16 & 20 & 20 & UV & E \\
J & $(15,12)$ & $1{+}y^{2}{+}y^{4}$ & $1{+}x{+}x^{4}$ & 32 & 4 & 1.4 & 24 & 20 & 16 & UV & E \\
J & $(30,6)$ & $1{+}y{+}y^{2}$ & $1{+}x^{2}{+}x^{8}$ & 32 & 4 & 1.4 & 22 & 16 & 16 & UV & E \\
J & $(15,12)$ & $1{+}y^{2}{+}y^{4}$ & $1{+}x^{11}{+}x^{14}$ & 32 & 4 & 1.4 & 18 & 16 & 16 & UV & E \\
J & $(30,6)$ & $1{+}y{+}y^{2}$ & $1{+}x^{4}{+}x^{16}$ & 32 & 4 & 1.4 & 20 & 20 & 16 & UV & E \\
J & $(30,6)$ & $1{+}y{+}y^{2}$ & $1{+}x^{6}{+}x^{8}$ & 32 & 4 & 1.4 & 24 & 16 & 20 & UV & E \\
J & $(15,12)$ & $1{+}y^{2}{+}y^{10}$ & $1{+}x^{3}{+}x^{4}$ & 32 & 4 & 1.4 & 20 & 20 & 16 & UV & E \\
J & $(15,12)$ & $1{+}y^{2}{+}y^{4}$ & $1{+}x^{2}{+}x^{9}$ & 32 & 4 & 1.4 & 20 & 16 & 16 & UV & E \\
J & $(15,12)$ & $1{+}x{+}x^{4}$ & $1{+}y^{2}{+}y^{4}$ & 32 & 4 & 1.4 & 12 & 16 & 16 & UV & E \\
J & $(15,12)$ & $1{+}x^{6}{+}x^{8}$ & $1{+}y^{2}{+}y^{4}$ & 32 & 4 & 1.4 & 20 & 20 & 12 & UV & E \\
J & $(15,12)$ & $1{+}y^{2}{+}y^{4}$ & $1{+}x^{7}{+}x^{9}$ & 32 & 4 & 1.4 & 12 & 20 & 16 & UV & E \\
K & $(15,12)$ & $1{+}y^{2}{+}y^{4}$ & $y^{2}{+}x^{6}{+}x^{7}$ & 24 & 4 & 1.1 & 12 & 8 & 8 & Other & E \\
K & $(30,6)$ & $y{+}y^{5}{+}x^{15}$ & $y{+}x{+}x^{3}$ & 24 & 4 & 1.1 & 12 & \textbf{4} & 8 & XY & E \\
L & $(15,12)$ & $y^{3}{+}y^{4}{+}x^{5}$ & $y^{6}{+}x^{5}{+}x^{10}$ & 20 & 4 & 0.9 & 8 & 8 & 8 & XY & E \\
M & $(30,6)$ & $1{+}y{+}x^{5}$ & $1{+}y^{4}{+}x^{5}$ & 20 & 4 & 0.9 & \textbf{4} & \textbf{4} & \textbf{4} & Other & E \\
N & $(15,12)$ & $1{+}y^{2}{+}y^{4}$ & $1{+}x^{3}{+}x^{4}y^{11}$ & 16 & 4 & 0.7 & \textbf{4} & \textbf{4} & \textbf{4} & Other & E \\
O & $(15,12)$ & $1{+}y^{10}{+}y^{11}$ & $1{+}x^{5}{+}x^{10}$ & 40 & 2 & 0.4 & 28 & 26 & 24 & UV & E \\
O & $(15,12)$ & $1{+}y^{2}{+}y^{7}$ & $1{+}x^{5}{+}x^{10}$ & 40 & 2 & 0.4 & 24 & 22 & 22 & UV & E \\
O & $(15,12)$ & $1{+}y{+}y^{2}$ & $1{+}x^{5}{+}x^{10}$ & 40 & 2 & 0.4 & 22 & 20 & 26 & UV & F \\
P & $(15,12)$ & $1{+}y^{4}{+}y^{5}$ & $1{+}x^{5}{+}x^{10}$ & 40 & 2 & 0.4 & 30 & 22 & 20 & UV & E \\
O & $(15,12)$ & $1{+}y{+}y^{11}$ & $1{+}x^{5}{+}x^{10}$ & 40 & 2 & 0.4 & 24 & 18 & 28 & UV & E \\
P & $(15,12)$ & $1{+}y^{4}{+}y^{11}$ & $1{+}x^{5}{+}x^{10}$ & 40 & 2 & 0.4 & 20 & 18 & 20 & UV & E \\
P & $(15,12)$ & $1{+}y{+}y^{5}$ & $1{+}x^{5}{+}x^{10}$ & 40 & 2 & 0.4 & 22 & 18 & 24 & UV & E \\
Q & $(30,6)$ & $1{+}y^{2}{+}y^{4}$ & $1{+}x^{16}{+}x^{19}$ & 32 & 2 & 0.4 & 20 & 22 & 22 & UV & E \\
O & $(15,12)$ & $1{+}y^{5}{+}y^{10}$ & $1{+}x^{5}{+}x^{10}$ & 40 & 2 & 0.4 & 22 & 18 & 16 & UV & E \\
P & $(15,12)$ & $1{+}y^{7}{+}y^{11}$ & $1{+}x^{5}{+}x^{10}$ & 40 & 2 & 0.4 & 26 & 16 & 24 & UV & E \\
P & $(15,12)$ & $1{+}y^{5}{+}y^{7}$ & $1{+}x^{5}{+}x^{10}$ & 40 & 2 & 0.4 & 16 & 26 & 22 & UV & E \\
P & $(15,12)$ & $1{+}y{+}y^{8}$ & $1{+}x^{5}{+}x^{10}$ & 40 & 2 & 0.4 & 22 & 16 & 18 & UV & E \\
R & $(15,12)$ & $1{+}y^{4}{+}y^{8}$ & $1{+}x^{2}{+}x^{4}$ & 32 & 2 & 0.4 & 18 & 16 & 16 & UV & E \\
S & $(30,6)$ & $1{+}y^{2}{+}y^{4}$ & $1{+}x^{2}{+}x^{10}$ & 32 & 2 & 0.4 & 18 & 14 & 18 & UV & E \\
R & $(30,6)$ & $1{+}y^{2}{+}y^{4}$ & $1{+}x^{14}{+}x^{22}$ & 32 & 2 & 0.4 & 20 & 12 & 18 & UV & E \\
Q & $(30,6)$ & $1{+}y^{2}{+}y^{4}$ & $1{+}x^{3}{+}x^{4}$ & 32 & 2 & 0.4 & 10 & 22 & 14 & UV & E \\
R & $(15,12)$ & $1{+}y^{4}{+}y^{8}$ & $1{+}x^{4}{+}x^{8}$ & 32 & 2 & 0.4 & 18 & 10 & 16 & UV & E \\
R & $(15,12)$ & $1{+}y^{4}{+}y^{8}$ & $1{+}x{+}x^{2}$ & 32 & 2 & 0.4 & 10 & 20 & 16 & UV & E \\
R & $(30,6)$ & $1{+}y^{2}{+}y^{4}$ & $1{+}x^{2}{+}x^{4}$ & 32 & 2 & 0.4 & 12 & 10 & 18 & UV & E \\
S & $(15,12)$ & $1{+}y^{4}{+}y^{8}$ & $1{+}x^{4}{+}x^{5}$ & 32 & 2 & 0.4 & 10 & 12 & 16 & UV & E \\
R & $(15,12)$ & $1{+}y^{4}{+}y^{8}$ & $1{+}x^{13}{+}x^{14}$ & 32 & 2 & 0.4 & 10 & 10 & 16 & UV & E \\
S & $(30,6)$ & $1{+}y^{2}{+}y^{4}$ & $1{+}x^{8}{+}x^{10}$ & 32 & 2 & 0.4 & 18 & 8 & 12 & UV & F \\
T & $(30,6)$ & $1{+}y^{2}{+}y^{4}$ & $1{+}x^{11}{+}x^{12}$ & 32 & 2 & 0.4 & 20 & 8 & 20 & UV & E \\
R & $(30,6)$ & $1{+}y^{2}{+}y^{4}$ & $1{+}x^{4}{+}x^{8}$ & 32 & 2 & 0.4 & 16 & 14 & 8 & UV & E \\
Q & $(30,6)$ & $1{+}y^{2}{+}y^{4}$ & $1{+}x{+}x^{4}$ & 32 & 2 & 0.4 & 16 & 6 & 24 & UV & E \\
R & $(30,6)$ & $1{+}y^{2}{+}y^{4}$ & $1{+}x^{22}{+}x^{26}$ & 32 & 2 & 0.4 & 6 & 10 & 18 & UV & E \\
\midrule
a$'$ & $(15,12)$ & $1{+}xy{+}x^{3}y^{3}{+}x^{7}y^{7}$ & $1{+}xy^{11}{+}x^{3}y^{9}{+}x^{7}y^{5}$ & 26 & 10 & 7.2 & 40 & 44 & 40 & MX & A \\
b$'$ & $(15,12)$ & $1{+}xy{+}x^{4}y^{4}{+}x^{6}y^{6}$ & $1{+}xy^{11}{+}x^{4}y^{8}{+}x^{6}y^{6}$ & 18 & 12 & 7.2 & 36 & 14 & 14 & MX & A \\
c$'$ & $(15,12)$ & $1{+}xy{+}x^{2}y$ & $1{+}x^{2}y^{3}{+}x^{3}y^{2}$ & 8 & 18 & 7.2 & 66 & 20 & 20 & MX & A \\
d$'$ & $(15,12)$ & $1{+}xy{+}x^{3}y^{2}$ & $1{+}xy{+}x^{13}y^{3}$ & 8 & 18 & 7.2 & 48 & 34 & 20 & MX & A \\
e$'$ & $(15,12)$ & $1{+}y^{2}{+}y^{10}{+}x{+}x^{14}$ & $1{+}y^{2}{+}y^{10}{+}x^{2}{+}x^{13}$ & 24 & 10 & 6.7 & 30 & 34 & 22 & XY & A \\
f$'$ & $(30,6)$ & $y{+}y^{2}{+}x^{3}$ & $y{+}y^{2}{+}x^{21}$ & 24 & 10 & 6.7 & 20 & 20 & 18 & XY & A \\
g$'$ & $(15,12)$ & $1{+}x^{4}y^{2}{+}x^{13}y^{4}$ & $1{+}x^{2}y^{4}{+}x^{4}y^{10}$ & 16 & 12 & 6.4 & 16 & 14 & \textbf{12} & MX & A \\
h$'$ & $(15,12)$ & $1{+}xy^{4}{+}x^{2}y^{2}$ & $1{+}x^{2}y^{4}{+}x^{4}y^{2}$ & 16 & 12 & 6.4 & 20 & 16 & \textbf{12} & MX & A \\
\bottomrule
\end{tabular*}
\end{table*}

\section{Complete non-CSS PBB code catalog}
\label{sm:pbb_catalog}

Tables~\ref{tab:pbb_cat36}--\ref{tab:pbb_cat360} list all 368 BLISS-unique non-CSS PBB codes from Campaign~5, sorted by $\FOM = kd^2/n$ descending within each block length.
Each code is defined by four polynomials $(A, B, C, D) \in R = \FF_2[x,y]/(x^\ell{-}1, y^m{-}1)$; when $C = D = 0$ the code reduces to a CSS BB code.
The distance $d$ is the MILP value; $\leq$ indicates an incumbent upper bound (otherwise distances are exact, with all logicals proven optimal).
The catalog spans seven lattices: $(3,6)$ and $(6,3)$ at $n = 36$, $(6,6)$ at $n = 72$, $(9,6)$ at $n = 108$, $(12,6)$ at $n = 144$, $(15,6)$ at $n = 180$, and $(30,6)$ at $n = 360$.

\begin{table*}[t]
\centering
\caption{All verified non-CSS PBB codes at $n = 36$ (47 codes, 47 distinct), sorted by $\FOM = kd^2/n$ descending.
$d$: MILP distance; $\leq$ indicates an incumbent upper bound (otherwise distances are exact, with all logicals proven optimal).}
\label{tab:pbb_cat36}
\tiny
\renewcommand{\arraystretch}{0.82}
\setlength{\tabcolsep}{1.5pt}
\begin{tabular*}{\textwidth}{@{\extracolsep{\fill}}cllllllcc}
\toprule
Cl.\ & $(\ell,m)$ & $A(x,y)$ & $B(x,y)$ & $C(x,y)$ & $D(x,y)$ & $k$ & $d$ & $\FOM$ \\
\midrule
4 & $(6,3)$ & $y{+}x^{3}y^{2}{+}x^{4}y$ & $x^{3}y{+}x^{4}{+}x^{4}y$ & $y{+}y^{2}{+}xy^{2}{+}x^{3}y{+}x^{5}{+}x^{5}y^{2}$ & $x{+}xy{+}xy^{2}{+}x^{2}y^{2}{+}x^{4}{+}x^{4}y{+}x^{4}y^{2}$ & 4 & 6 & 4.0 \\
29 & $(6,3)$ & $y{+}x^{3}y^{2}{+}x^{4}y$ & $x^{3}y{+}x^{4}{+}x^{4}y$ & $y^{2}{+}x^{2}y$ & $x^{3}y^{2}{+}x^{4}y^{2}{+}x^{5}y^{2}$ & 4 & 6 & 4.0 \\
18 & $(6,3)$ & $y{+}x^{3}y^{2}{+}x^{4}y$ & $x^{3}y{+}x^{4}{+}x^{4}y$ & $x^{5}$ & $xy^{2}{+}x^{2}y^{2}{+}x^{4}y^{2}$ & 4 & 6 & 4.0 \\
42 & $(6,3)$ & $y{+}x^{3}y^{2}{+}x^{4}y$ & $x^{3}y{+}x^{4}{+}x^{4}y$ & $y^{2}{+}x^{2}{+}x^{5}$ & $x^{2}y^{2}{+}x^{3}y^{2}$ & 4 & 6 & 4.0 \\
39 & $(6,3)$ & $y{+}x^{3}y^{2}{+}x^{4}y$ & $x^{3}y{+}x^{4}{+}x^{4}y$ & $x^{2}y{+}x^{4}y{+}x^{5}$ & $1{+}x{+}xy^{2}$ & 4 & 6 & 4.0 \\
19 & $(6,3)$ & $y{+}x^{3}y^{2}{+}x^{4}y$ & $x^{3}y{+}x^{4}{+}x^{4}y$ & $1{+}xy{+}x^{2}y^{2}{+}x^{4}y{+}x^{5}$ & $x{+}x^{4}y{+}x^{5}y$ & 4 & 6 & 4.0 \\
17 & $(3,6)$ & $y^{2}{+}xy^{4}{+}x^{2}y^{5}$ & $1{+}xy^{3}{+}x^{2}y$ & $1{+}y^{3}{+}xy^{3}$ & $x^{2}y^{5}$ & 4 & 6 & 4.0 \\
\bottomrule
\end{tabular*}
\end{table*}

\begin{table*}[t]
\centering
\caption{All verified non-CSS PBB codes at $n = 72$ (34 codes, 34 distinct), sorted by $\FOM = kd^2/n$ descending.
$d$: MILP distance; $\leq$ indicates an incumbent upper bound (otherwise distances are exact, with all logicals proven optimal).}
\label{tab:pbb_cat72}
\tiny
\renewcommand{\arraystretch}{0.82}
\setlength{\tabcolsep}{1.5pt}
\begin{tabular*}{\textwidth}{@{\extracolsep{\fill}}cllllllcc}
\toprule
Cl.\ & $(\ell,m)$ & $A(x,y)$ & $B(x,y)$ & $C(x,y)$ & $D(x,y)$ & $k$ & $d$ & $\FOM$ \\
\midrule
3 & $(6,6)$ & $x^{2}y{+}x^{3}y{+}x^{4}y^{4}$ & $1{+}x^{4}y^{5}{+}x^{5}y$ & $x^{2}y^{2}{+}x^{2}y^{5}$ & $y^{4}{+}x^{4}y^{4}{+}x^{5}y$ & 12 & 6 & 6.0 \\
2 & $(6,6)$ & $xy^{2}{+}x^{4}y^{3}{+}x^{4}y^{4}$ & $1{+}xy^{5}{+}x^{5}y^{4}$ & $x{+}xy{+}x^{4}y^{2}$ & $x^{3}{+}x^{3}y^{3}$ & 12 & 6 & 6.0 \\
5 & $(6,6)$ & $x^{2}y{+}x^{3}y^{4}{+}x^{4}y^{4}$ & $1{+}x^{4}y^{5}{+}x^{5}y$ & $x^{2}y^{4}{+}x^{5}y^{4}$ & $1{+}x^{3}$ & 12 & 6 & 6.0 \\
1 & $(6,6)$ & $xy^{2}{+}x^{4}y^{3}{+}x^{4}y^{4}$ & $1{+}xy^{5}{+}x^{5}y^{4}$ & $xy^{3}{+}x^{4}y^{3}$ & $xy^{5}{+}x^{4}y^{5}$ & 12 & 6 & 6.0 \\
8 & $(6,6)$ & $xy^{2}{+}x^{4}y^{3}{+}x^{4}y^{4}$ & $1{+}xy^{5}{+}x^{5}y^{4}$ & $xy^{3}{+}x^{4}y^{3}$ & $y^{3}{+}xy^{5}{+}x^{2}y$ & 12 & 6 & 6.0 \\
9 & $(6,6)$ & $x^{2}y{+}x^{3}y^{4}{+}x^{4}y^{4}$ & $1{+}x^{4}y^{5}{+}x^{5}y$ & $y^{3}{+}y^{4}{+}x^{2}y{+}x^{2}y^{2}{+}x^{4}y$ & $y{+}x{+}x^{2}y^{4}{+}x^{4}y^{5}$ & 12 & 6 & 6.0 \\
7 & $(6,6)$ & $xy^{2}{+}x^{4}y^{3}{+}x^{4}y^{4}$ & $1{+}xy^{5}{+}x^{5}y^{4}$ & $x^{3}y^{4}{+}x^{5}y^{3}$ & $x^{2}y^{5}{+}x^{4}y^{4}$ & 12 & 6 & 6.0 \\
\bottomrule
\end{tabular*}
\end{table*}

\begin{table*}[t]
\centering
\caption{All verified non-CSS PBB codes at $n = 108$ (107 codes, 107 distinct), sorted by $\FOM = kd^2/n$ descending. Part 1 of 2 (codes ranked 1--54 by FOM; continued in Table~\ref{tab:pbb_cat108_b}).
$d$: MILP distance; $\leq$ indicates an incumbent upper bound (otherwise distances are exact, with all logicals proven optimal).}
\label{tab:pbb_cat108}
\tiny
\renewcommand{\arraystretch}{0.82}
\setlength{\tabcolsep}{1.5pt}
\begin{tabular*}{\textwidth}{@{\extracolsep{\fill}}cllllllcc}
\toprule
Cl.\ & $(\ell,m)$ & $A(x,y)$ & $B(x,y)$ & $C(x,y)$ & $D(x,y)$ & $k$ & $d$ & $\FOM$ \\
\midrule
7 & $(9,6)$ & $xy^{2}{+}x^{4}y^{3}{+}x^{4}y^{4}$ & $1{+}xy^{5}{+}x^{5}y^{4}$ & $xy^{5}{+}x^{4}y$ & $xy{+}x^{4}y^{3}$ & 8 & 10 & 7.4 \\
15 & $(9,6)$ & $y{+}y^{2}{+}x^{3}$ & $y^{3}{+}x^{2}{+}x^{4}$ & $x^{6}y{+}x^{6}y^{4}$ & $x^{3}{+}x^{3}y^{3}$ & 8 & 10 & 7.4 \\
9 & $(9,6)$ & $xy^{2}{+}x^{4}y^{3}{+}x^{4}y^{4}$ & $1{+}xy^{5}{+}x^{5}y^{4}$ & $xy{+}x^{4}y^{4}$ & $xy^{2}{+}x^{4}y^{2}$ & 8 & 10 & 7.4 \\
13 & $(9,6)$ & $xy^{2}{+}x^{4}y^{3}{+}x^{4}y^{4}$ & $1{+}xy^{5}{+}x^{5}y^{4}$ & $xy{+}x^{4}y^{5}$ & $xy^{3}{+}x^{4}y$ & 8 & 10 & 7.4 \\
18 & $(9,6)$ & $x^{3}y^{2}{+}x^{6}y^{3}{+}x^{6}y^{4}$ & $x^{2}{+}x^{3}y^{5}{+}x^{7}y^{4}$ & $x^{7}y^{2}{+}x^{7}y^{3}$ & $x^{2}y^{5}{+}x^{3}y^{4}$ & 8 & 10 & 7.4 \\
16 & $(9,6)$ & $xy^{2}{+}x^{7}y^{3}{+}x^{7}y^{4}$ & $1{+}xy^{5}{+}x^{5}y^{4}$ & $xy{+}x^{4}y^{5}$ & $x{+}x^{4}y^{4}$ & 8 & 10 & 7.4 \\
12 & $(9,6)$ & $xy^{2}{+}x^{4}y^{3}{+}x^{4}y^{4}$ & $1{+}xy^{5}{+}x^{5}y^{4}$ & $x{+}x^{4}$ & $xy^{2}{+}x^{4}y^{2}$ & 8 & 10 & 7.4 \\
\bottomrule
\end{tabular*}
\end{table*}

\begin{table*}[t]
\centering
\caption{All verified non-CSS PBB codes at $n = 108$ (continued from Table~\ref{tab:pbb_cat108}; codes ranked 55--107 by FOM).
$d$: MILP distance; $\leq$ indicates an incumbent upper bound (otherwise distances are exact, with all logicals proven optimal).}
\label{tab:pbb_cat108_b}
\tiny
\renewcommand{\arraystretch}{0.82}
\setlength{\tabcolsep}{1.5pt}
\begin{tabular*}{\textwidth}{@{\extracolsep{\fill}}cllllllcc}
\toprule
Cl.\ & $(\ell,m)$ & $A(x,y)$ & $B(x,y)$ & $C(x,y)$ & $D(x,y)$ & $k$ & $d$ & $\FOM$ \\
\midrule
68 & $(9,6)$ & $xy{+}x^{2}y^{4}{+}x^{6}y$ & $1{+}x^{2}y^{5}{+}x^{7}y$ & $xy^{4}{+}x^{6}y{+}x^{6}y^{4}$ & $xy^{5}{+}x^{6}y^{2}$ & 6 & 8 & 3.6 \\
66 & $(9,6)$ & $x^{4}y{+}x^{6}y{+}x^{8}y^{4}$ & $1{+}xy{+}x^{8}y^{5}$ & $x^{6}y{+}x^{8}y^{4}$ & $x^{4}y^{2}{+}x^{6}y^{5}$ & 6 & 8 & 3.6 \\
45 & $(9,6)$ & $xy{+}x^{2}y^{4}{+}x^{6}y$ & $1{+}x^{2}y^{5}{+}x^{7}y$ & $x^{2}y{+}x^{6}y$ & $xy^{5}{+}x^{6}y^{2}$ & 6 & 8 & 3.6 \\
62 & $(9,6)$ & $x^{5}y^{4}{+}x^{6}y^{4}{+}x^{7}y$ & $1{+}x^{4}y{+}x^{5}y^{5}$ & $x^{6}y^{4}{+}x^{7}y$ & $x^{6}y^{5}{+}x^{7}y^{2}$ & 6 & 8 & 3.6 \\
64 & $(9,6)$ & $x^{4}y{+}x^{6}y^{4}{+}x^{8}y^{4}$ & $1{+}xy{+}x^{8}y^{5}$ & $x^{4}y^{4}{+}x^{6}y^{4}$ & $x^{4}y^{5}{+}x^{6}y^{2}$ & 6 & $\leq$8 & $\leq$3.6 \\
63 & $(9,6)$ & $x^{5}y^{4}{+}x^{6}y{+}x^{7}y$ & $1{+}x^{4}y{+}x^{5}y^{5}$ & $x^{4}y^{3}{+}x^{7}y^{2}$ & $xy^{3}{+}x^{5}y^{4}$ & 6 & 8 & 3.6 \\
58 & $(9,6)$ & $x^{3}y^{4}{+}x^{7}y^{4}{+}x^{8}y$ & $1{+}x^{2}y{+}x^{7}y^{5}$ & $x^{3}y^{2}{+}x^{8}y^{3}$ & $x^{3}y^{3}{+}x^{8}y^{4}$ & 6 & $\leq$8 & $\leq$3.6 \\
\bottomrule
\end{tabular*}
\end{table*}

\begin{table*}[t]
\centering
\caption{All verified non-CSS PBB codes at $n = 144$ (66 codes, 66 distinct), sorted by $\FOM = kd^2/n$ descending.
$d$: MILP distance; $\leq$ indicates an incumbent upper bound (otherwise distances are exact, with all logicals proven optimal).}
\label{tab:pbb_cat144}
\tiny
\renewcommand{\arraystretch}{0.82}
\setlength{\tabcolsep}{1.5pt}
\begin{tabular*}{\textwidth}{@{\extracolsep{\fill}}cllllllcc}
\toprule
Cl.\ & $(\ell,m)$ & $A(x,y)$ & $B(x,y)$ & $C(x,y)$ & $D(x,y)$ & $k$ & $d$ & $\FOM$ \\
\midrule
12 & $(12,6)$ & $x^{8}y^{4}{+}x^{9}y^{4}{+}x^{10}y$ & $1{+}x^{7}y{+}x^{8}y^{5}$ & $x^{9}y{+}x^{10}y$ & $x^{9}y^{5}{+}x^{10}y^{2}$ & 12 & $\leq$12 & $\leq$12.0 \\
10 & $(12,6)$ & $x^{3}y{+}x^{8}y^{4}{+}x^{10}y$ & $1{+}xy{+}x^{8}y^{5}$ & $x^{3}y{+}x^{8}y^{4}$ & $x^{3}y^{5}{+}x^{10}y^{2}$ & 12 & $\leq$12 & $\leq$12.0 \\
9 & $(12,6)$ & $xy^{2}{+}x^{10}y^{3}{+}x^{10}y^{4}$ & $1{+}xy^{5}{+}x^{11}y^{4}$ & $xy^{3}{+}x^{10}y^{3}$ & $xy^{5}{+}x^{10}y^{5}$ & 12 & 12 & 12.0 \\
6 & $(12,6)$ & $x^{2}y{+}x^{4}y^{4}{+}x^{9}y$ & $1{+}x^{4}y^{5}{+}x^{11}y$ & $x^{4}y^{2}{+}x^{9}y^{3}$ & $x^{2}y^{4}{+}x^{9}y^{3}$ & 12 & 12 & 12.0 \\
14 & $(12,6)$ & $x^{8}y^{4}{+}x^{9}y{+}x^{10}y$ & $1{+}x^{7}y{+}x^{8}y^{5}$ & $x^{10}y$ & $x^{9}y^{5}{+}x^{10}y^{2}$ & 12 & $\leq$12 & $\leq$12.0 \\
8 & $(12,6)$ & $x^{2}y{+}x^{3}y{+}x^{4}y^{4}$ & $1{+}x^{4}y^{5}{+}x^{5}y$ & $x^{3}{+}x^{4}y^{5}$ & $x^{2}y{+}x^{3}$ & 12 & $\leq$12 & $\leq$12.0 \\
11 & $(12,6)$ & $x^{3}y{+}x^{8}y^{4}{+}x^{10}y$ & $1{+}xy{+}x^{8}y^{5}$ & $x^{3}y^{2}{+}x^{8}y^{3}$ & $x^{3}y^{4}{+}x^{10}y^{3}$ & 12 & $\leq$12 & $\leq$12.0 \\
\bottomrule
\end{tabular*}
\end{table*}

\begin{table*}[t]
\centering
\caption{All verified non-CSS PBB codes at $n = 180$ (64 codes, 64 distinct), sorted by $\FOM = kd^2/n$ descending.
$d$: MILP distance; $\leq$ indicates an incumbent upper bound (otherwise distances are exact, with all logicals proven optimal).}
\label{tab:pbb_cat180}
\tiny
\renewcommand{\arraystretch}{0.82}
\setlength{\tabcolsep}{1.5pt}
\begin{tabular*}{\textwidth}{@{\extracolsep{\fill}}cllllllcc}
\toprule
Cl.\ & $(\ell,m)$ & $A(x,y)$ & $B(x,y)$ & $C(x,y)$ & $D(x,y)$ & $k$ & $d$ & $\FOM$ \\
\midrule
1 & $(15,6)$ & $x^{4}y{+}x^{6}y^{4}{+}x^{8}y^{4}$ & $1{+}x^{8}y^{5}{+}x^{10}y$ & $x^{8}y$ & $x^{4}y^{5}{+}x^{6}y^{2}$ & 6 & $\leq$21 & $\leq$14.7 \\
2 & $(15,6)$ & $xy^{4}{+}x^{8}y{+}x^{12}y^{4}$ & $1{+}xy^{5}{+}x^{5}y$ & $x^{8}y{+}x^{12}y^{4}$ & $x^{8}y^{2}{+}x^{12}y^{5}$ & 6 & $\leq$20 & $\leq$13.3 \\
3 & $(15,6)$ & $x^{2}y^{2}{+}x^{8}y^{3}{+}x^{8}y^{4}$ & $1{+}x^{2}y^{5}{+}x^{10}y^{4}$ & $x^{2}y{+}x^{8}y^{5}$ & $x^{2}y^{3}{+}x^{8}y$ & 8 & $\leq$16 & $\leq$11.4 \\
4 & $(15,6)$ & $xy{+}x^{2}y^{4}{+}x^{9}y$ & $1{+}x^{2}y^{5}{+}x^{10}y$ & $x^{2}y{+}x^{9}y^{4}$ & $xy^{5}{+}x^{9}y^{2}$ & 6 & $\leq$18 & $\leq$10.8 \\
6 & $(15,6)$ & $x^{7}y^{4}{+}x^{9}y{+}x^{11}y$ & $1{+}x^{5}y{+}x^{7}y^{5}$ & $x^{7}y^{4}{+}x^{9}y^{4}$ & $x^{9}y^{5}{+}x^{11}y^{2}$ & 6 & $\leq$18 & $\leq$10.8 \\
5 & $(15,6)$ & $x^{11}y^{4}{+}x^{12}y^{4}{+}x^{13}y$ & $1{+}x^{10}y{+}x^{11}y^{5}$ & $x^{12}y^{2}{+}x^{13}y^{3}$ & $x^{12}y^{3}{+}x^{13}y^{4}$ & 6 & $\leq$18 & $\leq$10.8 \\
7 & $(15,6)$ & $x^{11}y^{3}{+}x^{11}y^{4}{+}x^{14}y^{2}$ & $1{+}x^{10}y^{4}{+}x^{14}y^{5}$ & $x^{11}y{+}x^{14}y^{5}$ & $x^{11}y^{3}{+}x^{14}y$ & 8 & $\leq$14 & $\leq$8.7 \\
\bottomrule
\end{tabular*}
\end{table*}

\begin{table*}[t]
\centering
\caption{All verified non-CSS PBB codes at $n = 360$ (50 codes, 50 distinct), sorted by $\FOM = kd^2/n$ descending.
$d$: MILP distance; $\leq$ indicates an incumbent upper bound (otherwise distances are exact, with all logicals proven optimal).}
\label{tab:pbb_cat360}
\tiny
\renewcommand{\arraystretch}{0.82}
\setlength{\tabcolsep}{1.5pt}
\begin{tabular*}{\textwidth}{@{\extracolsep{\fill}}cllllllcc}
\toprule
Cl.\ & $(\ell,m)$ & $A(x,y)$ & $B(x,y)$ & $C(x,y)$ & $D(x,y)$ & $k$ & $d$ & $\FOM$ \\
\midrule
1 & $(30,6)$ & $x^{8}y{+}x^{16}y^{4}{+}x^{27}y$ & $1{+}x^{5}y{+}x^{16}y^{5}$ & $x^{16}y^{2}{+}x^{27}y^{3}$ & $x^{8}y^{4}{+}x^{27}y^{3}$ & 10 & $\leq$40 & $\leq$44.4 \\
3 & $(30,6)$ & $x^{8}y^{3}{+}x^{8}y^{4}{+}x^{17}y^{2}$ & $1{+}x^{17}y^{5}{+}x^{25}y^{4}$ & $x^{8}y{+}x^{17}y^{4}$ & $x^{8}y^{5}{+}x^{17}y^{5}$ & 10 & $\leq$32 & $\leq$28.4 \\
2 & $(30,6)$ & $x^{7}y^{2}{+}x^{28}y^{3}{+}x^{28}y^{4}$ & $1{+}x^{5}y^{4}{+}x^{7}y^{5}$ & $x^{7}y^{4}{+}x^{28}y^{2}$ & $x^{7}{+}x^{28}y^{4}$ & 10 & $\leq$32 & $\leq$28.4 \\
5 & $(30,6)$ & $x^{26}y^{3}{+}x^{26}y^{4}{+}x^{29}y^{2}$ & $1{+}x^{25}y^{4}{+}x^{29}y^{5}$ & $x^{26}y^{4}{+}x^{29}y$ & $x^{26}y^{2}{+}x^{29}y^{2}$ & 10 & $\leq$30 & $\leq$25.0 \\
6 & $(30,6)$ & $xy^{2}{+}x^{4}y^{3}{+}x^{4}y^{4}$ & $1{+}xy^{5}{+}x^{5}y^{4}$ & $x^{4}$ & $xy{+}x^{4}y^{3}$ & 10 & $\leq$30 & $\leq$25.0 \\
4 & $(30,6)$ & $x^{7}y^{2}{+}x^{28}y^{3}{+}x^{28}y^{4}$ & $1{+}x^{5}y^{4}{+}x^{7}y^{5}$ & $x^{7}y^{5}{+}x^{28}y$ & $x^{7}y{+}x^{28}y^{3}$ & 10 & $\leq$30 & $\leq$25.0 \\
7 & $(30,6)$ & $xy^{2}{+}x^{4}y^{3}{+}x^{4}y^{4}$ & $1{+}xy^{5}{+}x^{5}y^{4}$ & $x{+}x^{4}$ & $xy^{2}{+}x^{4}y^{2}$ & 12 & $\leq$24 & $\leq$19.2 \\
\bottomrule
\end{tabular*}
\end{table*}

\section{Per-decoder comparison statistics}
\label{sm:decoders}

Table~\ref{tab:decoder_rates} breaks down the fraction of Campaigns~1--3 trinomial codes for which each decoder configuration found the global minimum distance, stratified by encoding rate $k/n$.
OSD-CS$_{10}$ configurations consistently outperform OSD$_0$, with the advantage most pronounced at high encoding rates ($k/n > 15\%$).
Among all 154 Campaigns~1--3 polynomial representations, OSD$_0$/product-sum found the global minimum for only 49 (31.8\%), compared to 90 (58.4\%) for OSD-CS$_{10}$/product-sum and 100 (64.9\%) for OSD-CS$_{10}$/minimum-sum.
The mean gap between the OSD$_0$ minimum and the global minimum was 3.8 distance points (median 4, maximum 18).

These findings directly support the recommendation in the main text (Sec.~VI\,E) that distance verification should always employ multiple decoder configurations.

\begin{table*}[h]
\centering
\caption{Fraction of codes for which each decoder found the global minimum $d$, by encoding rate $k/n$.  ``Sole disc.'' counts codes where \emph{only} that decoder found $d$ across 50{,}000 trials.}
\label{tab:decoder_rates}
\begin{tabular}{lcccc}
\toprule
$k/n$ & Codes & OSD$_0$/sp & CS$_{10}$/sp & CS$_{10}$/ms \\
\midrule
$<5\%$ & 8 & 2 (25\%) & 5 (62\%) & 8 (100\%) \\
$5$--$10\%$ & 56 & 22 (39\%) & 34 (61\%) & 36 (64\%) \\
$10$--$15\%$ & 64 & 18 (28\%) & 35 (55\%) & 39 (61\%) \\
$>15\%$ & 26 & 7 (27\%) & 16 (62\%) & 17 (65\%) \\
\midrule
All & 154 & 49 (32\%) & 90 (58\%) & 100 (65\%) \\
Sole disc. & --- & 14 & 37 & 45 \\
\bottomrule
\end{tabular}
\end{table*}

\section{BP-OSD per-batch analysis}
\label{sm:per_batch}

\subsection{Per-batch distance distributions}

Figure~\ref{fig:sm_per_batch} shows the full per-batch distance distributions for five codes under the 150{,}000-trial multi-decoder protocol, broken down by decoder configuration (10 batches of 5{,}000 trials each).
These codes span the range from $d = 2$ to $d = 12$ (MILP exact), with the gross code as a stable control.
Red dashed lines mark the global minimum (verified bound).

The gross code ($k = 12$) returns $d = 12$ in all 30 batches with zero variance, confirming protocol soundness.
In contrast, all four higher-$k$ codes exhibit large inter-batch variance---e.g., the $\code{288,32,4}$ code (MILP $d = 4$) ranges from $d_{\text{BP}} = 18$ to $44$ across batches, a $4.5$--$11\times$ overestimate.
Decoder choice matters: OSD-CS$_{10}$/minimum-sum consistently produces lower medians and tighter bounds than OSD$_0$ for all four higher-$k$ codes.

\begin{figure*}[t]
\centering
\includegraphics[width=\textwidth]{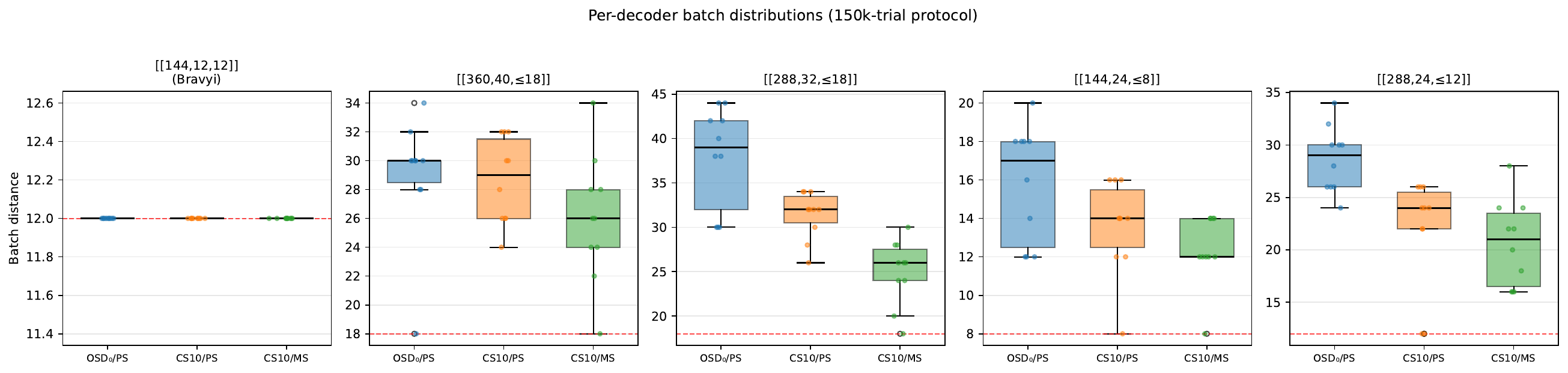}
\caption{Per-batch distance distributions under the 150{,}000-trial multi-decoder protocol.
Red dashed lines show the global minimum across all 30~batches.
The gross code ($k = 12$) is perfectly stable; all four higher-$k$ codes show inter-batch variance.
The three evolution-highlighted codes ($d_{\text{MILP}} = 2, 4, 6$) show dramatic overestimation; even the $\code{288,24,12}$ ($d_{\text{MILP}} = 12$; a direct sum of two gross codes) has OSD$_0$ medians $\approx 2.4\times$ the true distance, with only 2 of 30 batches finding the exact value.}
\label{fig:sm_per_batch}
\end{figure*}

\subsection{Per-batch statistics}

Table~\ref{tab:per_batch_stats} reports the full per-batch statistics for five representative codes under a second independent 150{,}000-trial run.
MILP exact distances are: $d = 2$ for $\code{360,40}$ at $(15,12)$, $d = 4$ for $\code{288,32}$ at $(24,6)$, $d = 6$ for $\code{144,24}$ at $(12,6)$, $d = 12$ for $\code{288,24}$ at $(12,12)$, and $d = 12$ for the gross code.

\begin{table*}[t]
\centering
\caption{Per-batch distance statistics under the 150{,}000-trial protocol (second independent run).
Each row reports 10~independent batches of 5{,}000 trials per decoder configuration.
Q1/Med/Q3 = 25th/50th/75th percentiles.
The gross code shows perfect stability; high-$k$ codes show dramatic overestimation vs.\ MILP exact distances.}
\label{tab:per_batch_stats}
\small
\begin{tabular}{llcccccc}
\toprule
Code & Decoder & Min & Q1 & Med & Q3 & Max & Mean $\pm$ Std \\
\midrule
\multirow{3}{*}{$\code{144,12,12}$~\cite{bravyi2024high}}
  & OSD$_0$/PS    & 12 & 12 & 12 & 12 & 12 & $12.0 \pm 0.0$ \\
  & CS$_{10}$/PS  & 12 & 12 & 12 & 12 & 12 & $12.0 \pm 0.0$ \\
  & CS$_{10}$/MS  & 12 & 12 & 12 & 12 & 12 & $12.0 \pm 0.0$ \\
\midrule
\multirow{3}{*}{$\code{360,40,2}$}
  & OSD$_0$/PS    & 18 & 28 & 30 & 30 & 34 & $29.0 \pm 4.0$ \\
  & CS$_{10}$/PS  & 24 & 26 & 29 & 31 & 32 & $28.6 \pm 2.8$ \\
  & CS$_{10}$/MS  & 18 & 24 & 26 & 28 & 34 & $26.0 \pm 4.2$ \\
\midrule
\multirow{3}{*}{$\code{288,32,4}$}
  & OSD$_0$/PS    & 30 & 32 & 39 & 42 & 44 & $37.8 \pm 5.5$ \\
  & CS$_{10}$/PS  & 26 & 30 & 32 & 33 & 34 & $31.4 \pm 2.5$ \\
  & CS$_{10}$/MS  & 18 & 24 & 26 & 27 & 30 & $25.0 \pm 3.5$ \\
\midrule
\multirow{3}{*}{$\code{144,24,6}$}
  & OSD$_0$/PS    & 12 & 12 & 17 & 18 & 20 & $15.8 \pm 2.9$ \\
  & CS$_{10}$/PS  &  8 & 12 & 14 & 15 & 16 & $13.6 \pm 2.3$ \\
  & CS$_{10}$/MS  &  8 & 12 & 12 & 14 & 14 & $12.4 \pm 1.7$ \\
\midrule
\multirow{3}{*}{$\code{288,24,12}$}
  & OSD$_0$/PS    & 24 & 26 & 29 & 30 & 34 & $28.6 \pm 3.0$ \\
  & CS$_{10}$/PS  & 12 & 22 & 24 & 25 & 26 & $21.8 \pm 5.1$ \\
  & CS$_{10}$/MS  & 16 & 16 & 21 & 23 & 28 & $20.6 \pm 3.9$ \\
\bottomrule
\end{tabular}
\end{table*}

\subsection{Extended verification at 1{,}500{,}000 trials}

To stress-test the 150k-trial bounds, we ran an extended protocol at $10\times$ depth: 3~decoders $\times$ 10~batches $\times$ 50{,}000~trials $= 1{,}500{,}000$ total trials per code, on three headline codes plus the gross code as a control (Table~\ref{tab:extended}).

\begin{table*}[h]
\centering
\caption{Extended verification at 1{,}500{,}000 trials per code, with MILP exact distances for comparison.
$d_{150\text{k}}$: 150k-trial protocol bound.
$d_{1.5\text{M}}$: extended bound.
$d_{\text{MILP}}$: exact distance via MILP.
Even at $10\times$ depth, BP-OSD overestimates persist.
The $\code{144,32,2}$ code has $A = B$; its $d = 2$ is proven exactly (main text Appendix~D, Theorem~1).}
\label{tab:extended}
\footnotesize
\begin{tabular}{lccccccc}
\toprule
Code & $(\ell,m)$ & $k$ & $d_{150\text{k}}$ & $d_{1.5\text{M}}$ & $d_{\text{MILP}}$ & Batches & $\FOM_{\text{MILP}}$ \\
\midrule
$\code{144,12,12}$ & $(12,6)$ & 12 & 12 & 12 & 12 & 30/30 & 12.0 \\
$\code{360,40,2}$ & $(15,12)$ & 40 & 24 & 6 & 2 & 1/30 & 0.4 \\
$\code{288,32,4}$ & $(24,6)$ & 32 & 22 & 12 & 4 & 1/30 & 1.8 \\
$\code{144,32,2}$ & $(12,6)$ & 32 & 14 & 2 & 2 & 1/30 & 0.9 \\
\bottomrule
\end{tabular}
\end{table*}

The gross code returned $d = 12$ in every batch, confirming that low-rate codes have stable BP-OSD bounds.
All three high-$k$ codes had bounds tightened: $d \leq 6$ for $\code{360,40}$ (MILP: $d = 2$), $d \leq 12$ for $\code{288,32}$ (MILP: $d = 4$), and $d = 2$ for $\code{144,32}$ (MILP: $d = 2$, proven exactly by main text Appendix~D, Theorem~1).
The $10\times$ deeper protocol correctly identified the \emph{direction} of the bias but still overestimated true distances by $3\times$, demonstrating that even million-trial BP-OSD protocols cannot substitute for exact distance computation.

The $d = 2$ result for the $A = B$ code is now \emph{proven exact}: all BB codes with $A = B$ and $k > 0$ have $d = 2$, because the identical column blocks of $H_Z = [A^\top | A^\top]$ produce $\ell m$ weight-2 vectors in $\ker(H_Z)$, exactly $k/2$ of which are independent nontrivial logicals.
The fact that 29 of 30 batches \emph{failed} to find this weight-2 logical---despite 72~independent weight-2 logicals existing---is a cautionary demonstration of BP-OSD's unreliability.

\section{Distance tightening across Campaigns~1--3}
\label{sm:tightening}

Table~\ref{tab:decoder_comparison_full} shows representative distance tightening examples from the Campaign~1 codes when moving from 60{,}000 to 150{,}000 trials under the multi-decoder protocol.
This data supplements Fig.~4 in the main text.

\begin{table*}[t]
\centering
\caption{Impact of verification protocol rigor on distance bounds (Campaign~1 codes).
$\Delta d$ is the reduction in the distance upper bound.
MILP exact distances confirm that even the 150k bounds are gross overestimates.}
\label{tab:decoder_comparison_full}
\begin{tabular}{lcccr}
\toprule
Code & $d_{\leq}$ (60k) & $d_{\leq}$ (150k) & $\Delta d$ & $\Delta\FOM_{\text{BP}}$ \\
\midrule
$\code{360,40}$ & 20 & 20 & 0 & 0.0 \\
$\code{360,32}$ & 20 & 16 & $-4$ & $-12.8$ \\
$\code{288,24}$ & 20 & 12 & $-8$ & $-21.3$ \\
$\code{288,32}$ & 16 & 12 & $-4$ & $-12.4$ \\
$\code{144,24}$ & 12 & 8 & $-4$ & $-13.3$ \\
$\code{144,32}$ & 10 & 6 & $-4$ & $-14.2$ \\
$\code{360,32}$ & 16 & 8 & $-8$ & $-17.1$ \\
$\code{288,32}$ & 24 & 20 & $-4$ & $-19.6$ \\
$\code{288,32}$ & 14 & 10 & $-4$ & $-10.7$ \\
\bottomrule
\end{tabular}
\end{table*}

\section{Univariate code families}
\label{sm:univariate}

Table~\ref{tab:uv_families} lists the top univariate (HGP of low-distance cyclic codes~\cite{eberhardt2024pruning}) polynomial families identified across all campaigns, showing how the same polynomial pair yields different parameters at different lattice dimensions.
Of the 154 Campaigns~1--3 trinomial representations, 87 (57\%) are univariate; after BLISS deduplication, these correspond to 12 distinct codes.
Multiple univariate families achieve $k = 40$ (the highest among trinomial codes; Campaign~4's mixed-monomial codes reach $k = 54$ via a factored-product generalization).
MILP exact distance computation reveals that \emph{every} univariate code has $d \in \{2, 4\}$.
By the Tillich--Z\'emor formula $d = \min(d_1, d_2, d_1^\top, d_2^\top)$~\cite{tillich2014quantum}, the BB distance is bounded by the smallest distance among the four component cyclic codes $\ker H_A$, $\ker H_B$, $\ker H_A^\top$, $\ker H_B^\top$ (for palindromic check polynomials as here, $H^\top = H$, so the four-way min collapses to $\min(d_A, d_B)$); low-weight quotients $(y^m{-}1)/A$ or $(x^\ell{-}1)/B$ force this min to be small (typically 2 or 4), and the odd weight of trinomial check polynomials forces all four component distances to be even, ruling out $d = 3$.
See main text Sec.~VI\,A for the derivation.
The FOM collapse is dramatic: the top-ranked code drops from $\FOM_{\text{BP}} = 64.0$ to $\FOM_{\text{MILP}} = 0.4$.

\begin{table*}[t]
\centering
\caption{Top univariate families ($A = f(y)$, $B = g(x)$, or vice versa; equivalent to HGP of low-distance cyclic codes~\cite{eberhardt2024pruning}).
Rows grouped by polynomial pair; sorted by $\FOM_{\text{BP}}$ descending (the metric available during evolution).
All univariate codes have $d_{\text{MILP}} \leq 4$, regardless of BP-OSD estimate (up to $12\times$ overestimation).}
\label{tab:uv_families}
\begin{tabular}{llccccccccc}
\toprule
$A(y)$ & $B(x)$ & $(\ell,m)$ & $n$ & $k$ & $d_{\text{MILP}}$ & $\FOM_{\text{MILP}}$ & $d_{\text{BP}}$ & $\FOM_{\text{BP}}$ & Ratio \\
\midrule
$1{+}y^{10}{+}y^{11}$ & $1{+}x^{5}{+}x^{10}$ & $(15,12)$ & 360 & 40 & 2 & 0.4 & 24 & 64.0 & 12$\times$ \\
\midrule
$1{+}y{+}y^{2}$ & $1{+}x^{5}{+}x^{25}$ & $(30,6)$ & 360 & 40 & 4 & 1.8 & 24 & 64.0 & 6$\times$ \\
\midrule
$1{+}y^{2}{+}y^{7}$ & $1{+}x^{5}{+}x^{10}$ & $(15,12)$ & 360 & 40 & 2 & 0.4 & 22 & 53.8 & 11$\times$ \\
\midrule
$1{+}y{+}y^{2}$ & $1{+}x^{5}{+}x^{10}$ & $(15,12)$ & 360 & 40 & 2 & 0.4 & 20 & 44.4 & 10$\times$ \\
 & & $(30,6)$ & 360 & 40 & 4 & 1.8 & 20 & 44.4 & 5$\times$ \\
\midrule
$1{+}y{+}y^{2}$ & $1{+}x^{20}{+}x^{25}$ & $(30,6)$ & 360 & 40 & 4 & 1.8 & 20 & 44.4 & 5$\times$ \\
\midrule
$1{+}y{+}y^{5}$ & $1{+}x^{5}{+}x^{10}$ & $(15,12)$ & 360 & 40 & 2 & 0.4 & 18 & 36.0 & 9$\times$ \\
 & & $(30,6)$ & 360 & 40 & 4 & 1.8 & 20 & 44.4 & 5$\times$ \\
\midrule
$1{+}y^{4}{+}y^{5}$ & $1{+}x^{5}{+}x^{10}$ & $(15,12)$ & 360 & 40 & 2 & 0.4 & 20 & 44.4 & 10$\times$ \\
 & & $(30,6)$ & 360 & 40 & 4 & 1.8 & 20 & 44.4 & 5$\times$ \\
\midrule
$1{+}y^{4}{+}y^{5}$ & $1{+}x^{20}{+}x^{25}$ & $(30,6)$ & 360 & 40 & 4 & 1.8 & 20 & 44.4 & 5$\times$ \\
\midrule
$1{+}y{+}y^{5}$ & $1{+}x^{5}{+}x^{25}$ & $(30,6)$ & 360 & 40 & 4 & 1.8 & 20 & 44.4 & 5$\times$ \\
\midrule
$1{+}y{+}y^{11}$ & $1{+}x^{5}{+}x^{10}$ & $(15,12)$ & 360 & 40 & 2 & 0.4 & 18 & 36.0 & 9$\times$ \\
\midrule
$1{+}y^{4}{+}y^{11}$ & $1{+}x^{5}{+}x^{10}$ & $(15,12)$ & 360 & 40 & 2 & 0.4 & 18 & 36.0 & 9$\times$ \\
\midrule
$1{+}y{+}y^{2}$ & $1{+}x^{4}{+}x^{18}$ & $(30,6)$ & 360 & 32 & 4 & 1.4 & 20 & 35.6 & 5$\times$ \\
\midrule
$1{+}y^{2}{+}y^{4}$ & $1{+}x^{16}{+}x^{19}$ & $(30,6)$ & 360 & 32 & 2 & 0.4 & 20 & 35.6 & 10$\times$ \\
\midrule
$1{+}y{+}y^{5}$ & $1{+}x^{4}{+}x^{8}$ & $(12,6)$ & 144 & 32 & 2 & 0.9 & 12 & 32.0 & 6$\times$ \\
 & & $(12,12)$ & 288 & 32 & 2 & 0.4 & 14 & 21.8 & 7$\times$ \\
 & & $(24,6)$ & 288 & 32 & 4 & 1.8 & 12 & 16.0 & 3$\times$ \\
\midrule
$1{+}y{+}y^{2}$ & $1{+}x^{4}{+}x^{8}$ & $(12,6)$ & 144 & 32 & 2 & 0.9 & 6 & 8.0 & 3$\times$ \\
 & & $(12,12)$ & 288 & 32 & 2 & 0.4 & 12 & 16.0 & 6$\times$ \\
 & & $(24,6)$ & 288 & 32 & 4 & 1.8 & 16 & 28.4 & 4$\times$ \\
\midrule
$1{+}y^{8}{+}y^{10}$ & $1{+}x^{2}{+}x^{4}$ & $(6,12)$ & 144 & 32 & 2 & 0.9 & 6 & 8.0 & 3$\times$ \\
 & & $(12,12)$ & 288 & 32 & 4 & 1.8 & 16 & 28.4 & 4$\times$ \\
\midrule
$1{+}y{+}y^{5}$ & $1{+}x^{4}{+}x^{20}$ & $(24,6)$ & 288 & 32 & 4 & 1.8 & 16 & 28.4 & 4$\times$ \\
\midrule
$1{+}y{+}y^{2}$ & $1{+}x^{16}{+}x^{20}$ & $(24,6)$ & 288 & 32 & 4 & 1.8 & 16 & 28.4 & 4$\times$ \\
\midrule
$1{+}y^{5}{+}y^{10}$ & $1{+}x^{5}{+}x^{10}$ & $(15,12)$ & 360 & 40 & 2 & 0.4 & 16 & 28.4 & 8$\times$ \\
\midrule
$1{+}y^{7}{+}y^{11}$ & $1{+}x^{5}{+}x^{10}$ & $(15,12)$ & 360 & 40 & 2 & 0.4 & 16 & 28.4 & 8$\times$ \\
\bottomrule
\end{tabular}
\end{table*}

\section{Full code capacity simulation data}
\label{sm:threshold}

This section presents the complete numerical results of the code capacity simulation summarized in the main text (Sec.~VI\,D).
Tables~\ref{tab:threshold_full} and~\ref{tab:per_qubit_full} cover six CSS codes (five weight-6 BB and one weight-8) at all 12 tested error rates; Table~\ref{tab:threshold_full_extra} adds two additional CSS codes ($\code{288,50,8}$ and $\code{144,8,12}$).
Tables~\ref{tab:threshold_noncss_full} and~\ref{tab:per_qubit_noncss_full} cover five non-CSS PBB codes alongside the CSS gross code baseline under $X$-only noise.
Table~\ref{tab:threshold_depolarizing_full} presents the same non-CSS codes under depolarizing noise; the table caption notes the decoder configuration and that, like the main text Sec.~VI\,D depolarizing data, this is a decoder-mismatched threshold (iid X/Z BP-OSD prior on the depolarizing channel) rather than the optimal depolarizing threshold.

\paragraph{Simulation protocol.}
Two noise channels are simulated with separate decoder configurations.
For the X-only tables (Tables~\ref{tab:threshold_full}--\ref{tab:per_qubit_noncss_full}): independent $X$ errors at rate~$p$ on each of~$n$ qubits, decoded with BP-OSD on $H_z$ with an iid prior of marginal rate~$p$ on $n$ columns.
For Table~\ref{tab:threshold_depolarizing_full}: the depolarizing channel (each qubit independently gets $X$, $Y$, or $Z$ with probability $p/3$), decoded with BP-OSD on the symplectic check matrix $[H_z \mid H_x]$ with iid bit-flip priors at marginal rate $2p/3$ on each of the $2n$ columns.
This iid prior is mismatched with respect to the depolarizing channel (it does not model the $X$/$Z$ correlation $P(e_x{=}1, e_z{=}1) = p/3$ induced by $Y$ errors); see main text Sec.~VI\,D and the table caption.
All decoders use OSD-CS order~7, product-sum, 20~BP iterations~\cite{panteleev2021degenerate,roffe2020decoding}.
Block LER: fraction of 100{,}000 shots with $\geq 1$ logical error among $k$~qubits.
Per-logical-qubit rate: $p_L = 1 - (1 - \text{LER})^{1/k}$.
Wilson 95\% CIs $\leq 0.06$~pp for LER~$\leq 1\%$; entries reported as $0.00$ have upper bounds $< 0.01\%$.

\begin{table*}[t]
\centering
\caption{Code capacity simulation: block logical error rate (LER, \%) vs.\ physical error rate~$p$ for six CSS codes (five weight-6 BB and one weight-8$^\dagger$; 100{,}000 shots per point).
Bold entries indicate block LER~$> p$.
$^\dagger$Campaign~4 code (weight-8 stabilizers; MILP incumbent $d \leq 14$; BP-OSD 300k confirms $d = 14$).
See Table~\ref{tab:threshold_full_extra} for two additional CSS codes.}
\label{tab:threshold_full}
\small
\begin{tabular}{r|cc|cc|cc|cc|cc|cc}
\toprule
& \multicolumn{2}{c|}{$\code{144,12,12}$~\cite{bravyi2024high}} & \multicolumn{2}{c|}{$\code{288,24,12}$} & \multicolumn{2}{c|}{$\code{288,16,12}$} & \multicolumn{2}{c|}{$\code{360,16,{\leq}14}$} & \multicolumn{2}{c|}{$\code{360,20,{\leq}14}$$^\dagger$} & \multicolumn{2}{c}{$\code{144,24,6}$} \\
$p$ (\%) & LER & Unc. & LER & Unc. & LER & Unc. & LER & Unc. & LER & Unc. & LER & Unc. \\
\midrule
0.2 & 0.00 & 2.37 & 0.00 & 4.69 & 0.00 & 3.15 & 0.00 & 3.15 & 0.00 & 3.92 & 0.00 & 4.69 \\
0.5 & 0.00 & 5.84 & 0.01 & 11.33 & 0.01 & 7.71 & 0.01 & 7.71 & 0.00 & 9.54 & 0.02 & 11.33 \\
0.8 & 0.01 & 9.19 & 0.04 & 17.53 & 0.08 & 12.06 & 0.07 & 12.06 & 0.01 & 14.84 & 0.10 & 17.53 \\
1.0 & 0.03 & 11.36 & 0.07 & 21.43 & 0.13 & 14.85 & 0.15 & 14.85 & 0.02 & 18.21 & 0.27 & 21.43 \\
1.5 & 0.08 & 16.59 & 0.19 & 30.42 & 0.41 & 21.48 & 0.43 & 21.48 & 0.04 & 26.09 & 0.85 & 30.42 \\
2.0 & 0.15 & 21.53 & 0.36 & 38.42 & 0.85 & 27.62 & 0.84 & 27.62 & 0.12 & 33.24 & \textbf{2.11} & 38.42 \\
3.0 & 0.49 & 30.62 & 0.98 & 51.86 & 1.69 & 38.57 & 1.69 & 38.57 & 0.60 & 45.62 & \textbf{7.22} & 51.86 \\
4.0 & 1.28 & 38.73 & 2.66 & 62.46 & 2.55 & 47.96 & 2.98 & 47.96 & 2.38 & 55.80 & \textbf{16.77} & 62.46 \\
5.0 & 3.70 & 45.96 & \textbf{7.28} & 70.80 & 4.18 & 55.99 & \textbf{5.27} & 55.99 & \textbf{7.95} & 64.15 & \textbf{28.97} & 70.80 \\
6.0 & \textbf{8.61} & 52.41 & \textbf{16.77} & 77.35 & \textbf{7.45} & 62.84 & \textbf{10.99} & 62.84 & \textbf{20.10} & 70.99 & \textbf{43.56} & 77.35 \\
7.0 & \textbf{16.99} & 58.14 & \textbf{31.69} & 82.48 & \textbf{14.66} & 68.69 & \textbf{21.95} & 68.69 & \textbf{38.50} & 76.58 & \textbf{57.89} & 82.48 \\
8.0 & \textbf{28.76} & 63.23 & \textbf{50.06} & 86.48 & \textbf{27.73} & 73.66 & \textbf{38.70} & 73.66 & \textbf{60.29} & 81.13 & \textbf{70.68} & 86.48 \\
\bottomrule
\end{tabular}
\end{table*}

\begin{table*}[t]
\centering
\caption{Per-logical-qubit error rate $p_L = 1{-}(1{-}\text{LER})^{1/k}$ (\%) for the six CSS codes in Table~\ref{tab:threshold_full}.
All six codes achieve $p_L < p$ at every tested rate, confirming genuine per-qubit error suppression.
$^\dagger$Campaign~4 code (weight-8 stabilizers).}
\label{tab:per_qubit_full}
\footnotesize
\begin{tabular}{r|c|c|c|c|c|c}
\toprule
$p$ (\%) & $\code{144,12,12}$ ($k{=}12$) & $\code{288,24,12}$ ($k{=}24$) & $\code{288,16,12}$ ($k{=}16$) & $\code{360,16,{\leq}14}$ ($k{=}16$) & $\code{360,20,{\leq}14}$$^\dagger$ ($k{=}20$) & $\code{144,24,6}$ ($k{=}24$) \\
\midrule
0.2 & ${<}0.01$ & ${<}0.01$ & ${<}0.01$ & ${<}0.01$ & ${<}0.01$ & ${<}0.01$ \\
0.5 & ${<}0.01$ & ${<}0.01$ & ${<}0.01$ & ${<}0.01$ & ${<}0.01$ & ${<}0.01$ \\
0.8 & ${<}0.01$ & ${<}0.01$ & ${<}0.01$ & ${<}0.01$ & ${<}0.01$ & ${<}0.01$ \\
1.0 & ${<}0.01$ & ${<}0.01$ & 0.01 & 0.01 & ${<}0.01$ & 0.01 \\
1.5 & 0.01 & 0.01 & 0.03 & 0.03 & ${<}0.01$ & 0.04 \\
2.0 & 0.01 & 0.01 & 0.05 & 0.05 & 0.01 & 0.09 \\
3.0 & 0.04 & 0.04 & 0.11 & 0.11 & 0.03 & 0.31 \\
4.0 & 0.11 & 0.11 & 0.16 & 0.19 & 0.12 & 0.76 \\
5.0 & 0.31 & 0.31 & 0.27 & 0.34 & 0.41 & 1.42 \\
6.0 & 0.75 & 0.76 & 0.48 & 0.73 & 1.12 & 2.35 \\
7.0 & 1.54 & 1.58 & 0.99 & 1.54 & 2.40 & 3.54 \\
8.0 & 2.79 & 2.85 & 2.01 & 3.01 & 4.51 & 4.98 \\
\bottomrule
\end{tabular}
\end{table*}

\begin{table*}[t]
\centering
\caption{Code capacity simulation: block LER (\%) and per-logical-qubit rate $p_L$ (\%) for two additional CSS codes (100{,}000 shots per point).
Bold entries indicate block LER~$> p$.
$^\ddagger$Campaign~4 code (weight-8 stabilizers, cross-factored).}
\label{tab:threshold_full_extra}
\small
\begin{tabular}{r|ccc|ccc}
\toprule
& \multicolumn{3}{c|}{$\code{288,50,8}^{\ddagger}$} & \multicolumn{3}{c}{$\code{144,8,12}$} \\
$p$ (\%) & LER & $p_L$ & Unc. & LER & $p_L$ & Unc. \\
\midrule
0.2 & 0.00 & ${<}0.01$ & 9.53 & 0.00 & ${<}0.01$ & 1.59 \\
0.5 & 0.02 & ${<}0.01$ & 22.17 & 0.00 & ${<}0.01$ & 3.93 \\
0.8 & 0.08 & ${<}0.01$ & 33.08 & 0.01 & ${<}0.01$ & 6.22 \\
1.0 & 0.13 & ${<}0.01$ & 39.50 & 0.01 & ${<}0.01$ & 7.73 \\
1.5 & 0.48 & 0.01 & 53.03 & 0.05 & 0.01 & 11.39 \\
2.0 & 1.22 & 0.02 & 63.58 & 0.08 & 0.01 & 14.92 \\
3.0 & \textbf{4.71} & 0.10 & 78.19 & 0.29 & 0.04 & 21.63 \\
4.0 & \textbf{12.74} & 0.27 & 87.01 & 0.91 & 0.11 & 27.86 \\
5.0 & \textbf{27.20} & 0.63 & 92.31 & 2.72 & 0.34 & 33.66 \\
6.0 & \textbf{46.27} & 1.23 & 95.47 & \textbf{6.44} & 0.83 & 39.04 \\
7.0 & \textbf{65.92} & 2.13 & 97.34 & \textbf{13.13} & 1.74 & 44.04 \\
8.0 & \textbf{81.28} & 3.30 & 98.45 & \textbf{23.07} & 3.23 & 48.68 \\
\bottomrule
\end{tabular}
\end{table*}

\begin{table*}[t]
\centering
\caption{Code capacity simulation ($X$-only noise): block LER (\%) for five non-CSS PBB codes and the CSS gross code baseline (100{,}000 shots).
Bold entries indicate block LER~$> p$.}
\label{tab:threshold_noncss_full}
\small
\setlength{\tabcolsep}{3pt}
\begin{tabular}{r|cc|cc|cc|cc|cc|cc}
\toprule
& \multicolumn{2}{c|}{$\code{144,12,12}$~\cite{bravyi2024high} (CSS)} & \multicolumn{2}{c|}{$\code{144,12,12}$ PBB} & \multicolumn{2}{c|}{$\code{72,4,8}$ PBB} & \multicolumn{2}{c|}{$\code{108,8,10}$ PBB} & \multicolumn{2}{c|}{$\code{360,12,{\leq}20}$ PBB} & \multicolumn{2}{c}{$\code{360,12,{\leq}24}$ PBB} \\
$p$ (\%) & LER & Unc. & LER & Unc. & LER & Unc. & LER & Unc. & LER & Unc. & LER & Unc. \\
\midrule
0.2 & 0.00 & 2.37 & 0.00 & 2.37 & 0.00 & 0.80 & 0.00 & 1.59 & 0.00 & 2.37 & 0.00 & 2.37 \\
0.5 & 0.00 & 5.84 & 0.01 & 5.84 & 0.00 & 1.99 & 0.00 & 3.93 & 0.00 & 5.84 & 0.00 & 5.84 \\
0.8 & 0.01 & 9.19 & 0.01 & 9.19 & 0.00 & 3.16 & 0.00 & 6.22 & 0.00 & 9.19 & 0.00 & 9.19 \\
1.0 & 0.03 & 11.36 & 0.02 & 11.36 & 0.00 & 3.94 & 0.00 & 7.73 & 0.00 & 11.36 & 0.00 & 11.36 \\
1.5 & 0.08 & 16.59 & 0.08 & 16.59 & 0.00 & 5.87 & 0.00 & 11.39 & 0.00 & 16.59 & 0.00 & 16.59 \\
2.0 & 0.15 & 21.53 & 0.21 & 21.53 & 0.00 & 7.76 & 0.00 & 14.92 & 0.00 & 21.53 & 0.00 & 21.53 \\
3.0 & 0.49 & 30.62 & 0.76 & 30.62 & 0.01 & 11.47 & 0.01 & 21.63 & 0.00 & 30.62 & 0.01 & 30.62 \\
4.0 & 1.28 & 38.73 & 2.07 & 38.73 & 0.02 & 15.07 & 0.02 & 27.86 & 0.01 & 38.73 & 0.02 & 38.73 \\
5.0 & 3.70 & 45.96 & 4.29 & 45.96 & 0.04 & 18.55 & 0.05 & 33.66 & 0.04 & 45.96 & 0.04 & 45.96 \\
6.0 & \textbf{8.61} & 52.41 & \textbf{7.51} & 52.41 & 0.07 & 21.93 & 0.12 & 39.04 & 0.15 & 52.41 & 0.16 & 52.41 \\
7.0 & \textbf{16.99} & 58.14 & \textbf{11.64} & 58.14 & 0.15 & 25.19 & 0.25 & 44.04 & 0.39 & 58.14 & 0.37 & 58.14 \\
8.0 & \textbf{28.76} & 63.23 & \textbf{16.61} & 63.23 & 0.22 & 28.36 & 0.51 & 48.68 & 0.93 & 63.23 & 0.90 & 63.23 \\
\bottomrule
\end{tabular}
\end{table*}

\begin{table*}[t]
\centering
\caption{Per-logical-qubit error rate $p_L = 1{-}(1{-}\text{LER})^{1/k}$ (\%) for the non-CSS PBB codes in Table~\ref{tab:threshold_noncss_full} ($X$-only noise).
All six codes achieve $p_L < p$ at every tested rate.
The two $n=360$ codes have identical per-qubit rates at displayed precision since both have $k=12$ and their $X$-only LER values agree to within rounding.}
\label{tab:per_qubit_noncss_full}
\tiny
\setlength{\tabcolsep}{2pt}
\begin{tabular}{r|c|c|c|c|c|c}
\toprule
$p$ (\%) & $\code{144,12,12}$ CSS ($k{=}12$) & $\code{144,12,12}$ PBB ($k{=}12$) & $\code{72,4,8}$ PBB ($k{=}4$) & $\code{108,8,10}$ PBB ($k{=}8$) & $\code{360,12,{\leq}20}$ PBB ($k{=}12$) & $\code{360,12,{\leq}24}$ PBB ($k{=}12$) \\
\midrule
0.2 & ${<}0.01$ & ${<}0.01$ & ${<}0.01$ & ${<}0.01$ & ${<}0.01$ & ${<}0.01$ \\
0.5 & ${<}0.01$ & ${<}0.01$ & ${<}0.01$ & ${<}0.01$ & ${<}0.01$ & ${<}0.01$ \\
0.8 & ${<}0.01$ & ${<}0.01$ & ${<}0.01$ & ${<}0.01$ & ${<}0.01$ & ${<}0.01$ \\
1.0 & ${<}0.01$ & ${<}0.01$ & ${<}0.01$ & ${<}0.01$ & ${<}0.01$ & ${<}0.01$ \\
1.5 & 0.01 & 0.01 & ${<}0.01$ & ${<}0.01$ & ${<}0.01$ & ${<}0.01$ \\
2.0 & 0.01 & 0.02 & ${<}0.01$ & ${<}0.01$ & ${<}0.01$ & ${<}0.01$ \\
3.0 & 0.04 & 0.06 & ${<}0.01$ & ${<}0.01$ & ${<}0.01$ & ${<}0.01$ \\
4.0 & 0.11 & 0.17 & ${<}0.01$ & ${<}0.01$ & ${<}0.01$ & ${<}0.01$ \\
5.0 & 0.31 & 0.36 & 0.01 & 0.01 & ${<}0.01$ & ${<}0.01$ \\
6.0 & 0.75 & 0.65 & 0.02 & 0.01 & 0.01 & 0.01 \\
7.0 & 1.54 & 1.03 & 0.04 & 0.03 & 0.03 & 0.03 \\
8.0 & 2.79 & 1.50 & 0.06 & 0.06 & 0.08 & 0.08 \\
\bottomrule
\end{tabular}
\end{table*}

\begin{table*}[t]
\centering
\caption{Code capacity simulation (depolarizing channel, decoded with iid X/Z BP-OSD prior): block LER (\%) for five non-CSS PBB codes and the CSS gross code baseline (100{,}000 shots).
Sampler: each qubit independently gets $X$, $Y$, or $Z$ with probability $p/3$.
Decoder: BP-OSD on the symplectic check matrix $[H_z \mid H_x]$ with iid bit-flip priors at marginal rate $2p/3$ on each of the $2n$ columns; this iid prior is mismatched with respect to the depolarizing channel (see main text Sec.~VI\,D), so the reported LERs are a lower bound on the optimal depolarizing threshold rather than the optimal threshold itself.
Bold entries indicate block LER~$> p$.
Under depolarizing noise, the CSS and PBB $\code{144,12,12}$ achieve nearly identical LER at all tested rates---a comparison that is decoder-dependent.
Both $n=360$ PBB codes achieve LER~$<p$ throughout: the top-FOM $\code{360,12,{\leq}24}$ reaches $4.52\%$ and the lower-FOM $\code{360,12,{\leq}20}$ reaches $4.97\%$ at $p=8\%$, with a small crossover in which the lower-FOM code is marginally better at low $p$ ($p \leq 4\%$) and the higher-FOM code is better at high $p$ ($p \geq 5\%$).}
\label{tab:threshold_depolarizing_full}
\small
\setlength{\tabcolsep}{3pt}
\begin{tabular}{r|cc|cc|cc|cc|cc|cc}
\toprule
& \multicolumn{2}{c|}{$\code{144,12,12}$~\cite{bravyi2024high} (CSS)} & \multicolumn{2}{c|}{$\code{144,12,12}$ PBB} & \multicolumn{2}{c|}{$\code{72,4,8}$ PBB} & \multicolumn{2}{c|}{$\code{108,8,10}$ PBB} & \multicolumn{2}{c|}{$\code{360,12,{\leq}20}$ PBB} & \multicolumn{2}{c}{$\code{360,12,{\leq}24}$ PBB} \\
$p$ (\%) & LER & Unc. & LER & Unc. & LER & Unc. & LER & Unc. & LER & Unc. & LER & Unc. \\
\midrule
0.2 & 0.00 & 2.37 & 0.00 & 2.37 & 0.00 & 0.80 & 0.00 & 1.59 & 0.00 & 2.37 & 0.00 & 2.37 \\
0.5 & 0.00 & 5.84 & 0.01 & 5.84 & 0.00 & 1.99 & 0.00 & 3.93 & 0.00 & 5.84 & 0.01 & 5.84 \\
0.8 & 0.01 & 9.19 & 0.01 & 9.19 & 0.00 & 3.16 & 0.01 & 6.22 & 0.03 & 9.19 & 0.05 & 9.19 \\
1.0 & 0.03 & 11.36 & 0.02 & 11.36 & 0.01 & 3.94 & 0.01 & 7.73 & 0.05 & 11.36 & 0.08 & 11.36 \\
1.5 & 0.09 & 16.59 & 0.07 & 16.59 & 0.01 & 5.87 & 0.03 & 11.39 & 0.15 & 16.59 & 0.19 & 16.59 \\
2.0 & 0.13 & 21.53 & 0.14 & 21.53 & 0.05 & 7.76 & 0.07 & 14.92 & 0.27 & 21.53 & 0.35 & 21.53 \\
3.0 & 0.37 & 30.62 & 0.34 & 30.62 & 0.25 & 11.47 & 0.18 & 21.63 & 0.70 & 30.62 & 0.81 & 30.62 \\
4.0 & 0.71 & 38.73 & 0.67 & 38.73 & 0.69 & 15.07 & 0.49 & 27.86 & 1.24 & 38.73 & 1.27 & 38.73 \\
5.0 & 1.27 & 45.96 & 1.40 & 45.96 & 1.80 & 18.55 & 1.05 & 33.66 & 1.89 & 45.96 & 1.82 & 45.96 \\
6.0 & 2.69 & 52.41 & 2.64 & 52.41 & 3.73 & 21.93 & 2.32 & 39.04 & 2.59 & 52.41 & 2.40 & 52.41 \\
7.0 & 5.15 & 58.14 & 5.22 & 58.14 & 6.63 & 25.19 & 4.56 & 44.04 & 3.58 & 58.14 & 3.22 & 58.14 \\
\textbf{8.0} & \textbf{9.43} & 63.23 & \textbf{9.29} & 63.23 & \textbf{10.47} & 28.36 & \textbf{8.38} & 48.68 & 4.97 & 63.23 & 4.52 & 63.23 \\
\bottomrule
\end{tabular}
\end{table*}

\section{Utility stratification of the catalog}
\label{sm:utility}

Not all 465 distinct codes are useful for error correction.
We stratify the combined catalog by practical relevance.

Among the 97 CSS codes, the MILP-verified distance distribution spans $d = 2$ (the $A = B$ and some univariate codes) through $d = 14$.
Among the Campaigns~1--3 trinomial portion of the catalog, 13 $(n, k, d)$ parameter triples have $d \geq 6$ (the $x/y$-swap and select mixed-monomial families); Campaign~4's mixed-monomial search expands this to 41 distinct triples with $d \geq 6$ across the full CSS catalog.
The CSS codes with $\FOM \geq 6.0$---a threshold above which codes offer meaningful error-correction advantage over the uncoded rate---number 5 parameter triples at Campaigns~1--3 block lengths (including the $\code{288,24,12}$ at $\FOM = 12.0$, which is a direct sum of two gross codes; see main text Sec.~VI\,A) and 29 additional $(n,k,d)$ triples reached only by Campaign~4 (34 distinct CSS triples in total at $\FOM \geq 6$).

In contrast, all 368 non-CSS PBB codes have $d \geq 6$ by construction: Campaign~5's adaptive pipeline rejected $d \leq 4$ codes during evolution.
Of these, 108 distinct codes have $\FOM \geq 6.0$, 53 have $\FOM \geq 8.0$, and 25 match or exceed $\FOM = 12.0$.

The most practically relevant codes---those with $d \geq 8$ and $\FOM \geq 6.0$---span 50 distinct $(n,k,d)$ parameter triples across both catalogs (30 CSS, 24 non-CSS, with 4 shared between the two catalogs); a representative subset is consolidated in main text Table~V.

\section{Additional comparison and novelty details}
\label{sm:comparison}

\subsection{Novelty assessment}
\label{sm:novelty}

The published inventory of weight-6 BB codes at the block lengths we study ($n = 144$, $288$, $360$) is small: one code each from Bravyi et al.~\cite{bravyi2024high}.
Weight-8 codes have been reported at $n = 144$ by Symons et al.~\cite{symons2025covering} ($\code{144,14,14}$, $\FOM = 19.1$, covering graphs) and Liang and Chen~\cite{liang2025selfdual} ($\code{144,6,14}$, $\FOM = 8.2$, self-dual).
Among the most comprehensive systematic searches is Liang et al.~\cite{liang2025generalized}, who enumerate all weight-6 generalized toric codes of a canonical 3-term form $f = 1{+}x{+}x^a y^b$ for $n \leq 400$; at $n = 288$ they report only the $\code{288,12,18}$.
Our $x/y$-swap codes do \emph{not} fall within their canonical form (which requires one monomial to be the identity and a second to be $x$), though a subset may be equivalent under coordinate transformations of the underlying $\ZZ_\ell \times \ZZ_m$ group (e.g., $x \mapsto x^a$ for $a$ coprime to $\ell$, or swapping $x$ and $y$ when $\ell = m$).
Lin and Pryadko~\cite{lin2024abelian} provide a comprehensive enumeration of 2BGA codes (including nonabelian groups) with weight~$\leq 8$ for $n \leq 100$; this database is complementary to ours but does not include the block lengths considered here.

To our knowledge, two of our 99 CSS equivalence classes match previously reported codes: the gross code $\code{144,12,12}$ and the $\code{360,12,{\leq}24}$ (both~\cite{bravyi2024high}; Campaign~4 rediscovered these through mixed-monomial representations).
The remaining 97 classes do not appear in prior work, including all codes at non-standard factorizations---$(18,8)$, $(16,9)$, $(24,6)$ for $n = 288$; $(15,12)$ for $n = 360$---and all high-rate codes with $k \geq 16$ at $n = 288$.

We note that unreported codes may exist in unpublished search databases (e.g., from industrial groups); our novelty claim is based on a review of 25+ papers and 6 public code repositories through March~2026, and we have not had access to unpublished search data.

\subsection{Weight-8 comparison}
\label{sm:weight8}

Weight-8 codes have eight nonzero entries per row of $H_X$ and $H_Z$, compared to six for weight-6 codes.
Under standard depth-optimal syndrome-extraction circuits, this corresponds to eight controlled-NOT (CNOT) layers per syndrome-extraction cycle versus six (a 33\% increase in circuit depth per cycle), placing weight-8 codes in a different hardware-complexity regime.
At $n = 144$, the weight-8 $\code{144,14,14}$~\cite{symons2025covering} achieves $\FOM = 19.1$, $1.6\times$ our highest weight-6 $\FOM = 12.0$.
Whether the FOM advantage compensates for the increased per-cycle depth depends on the noise model and circuit schedule; a full circuit-level comparison is left to future work.

\section{Campaign details}
\label{sm:campaigns}

\subsection{Evolution dynamics}

\paragraph{Campaign~1 (Gemini~3 Flash Preview, 100 iterations, ${\sim}$US\$15 LLM inference).}
The population of 100 programs across 3 islands reached 79 of 100 occupied MAP-Elites niches at convergence.
The highest-scoring program (combined score 327.0) was found at iteration~83, at generation depth~5, growing from 4~strategies and 230~lines (seed) to 7~strategies and 292~lines with the addition of a univariate strategy.
Nine codes were subjected to full verification.

\paragraph{Campaign~2 (3-model ensemble, 251 iterations, ${\sim}$US\$25 LLM inference).}
The seed solution incorporated five Campaign~1 codes into its known-code table.
Nine improvements were recorded over the seed, with the largest gain ($+39.3$ points) at iteration~3.
Highest combined score: 348.6 ($+6.6\%$ over Campaign~1).
Manually stopped after 251 iterations.

\paragraph{Campaign~3 (3-model ensemble, 500 iterations, ${\sim}$US\$50 LLM inference).}
Population 1{,}000 across 5 islands; highest-scoring program (score 354.4) at iteration~462.
The last ${\sim}275$ iterations produced only one marginal improvement ($+1.2\%$), indicating saturation.
145 unique codes not found in Campaign~1.

\paragraph{Campaign~4 (Ansatz, 300 iterations, ${\sim}$US\$47 LLM inference).}
24 parallel evaluation workers on a 64-core server.
All 1{,}188 codes passing evolutionary fitness filters were MILP-verified; top~39 re-verified with extended budgets (3{,}000\,s per logical, 72{,}000\,s total); top~6 subjected to 300{,}000-trial BP-OSD attacks.
45 distinct codes from structural families not present in~\cite{bravyi2024high}.
Campaign~4 independently rediscovered the gross code through six mixed-monomial representations.

\paragraph{Campaign~5 (PBB, 500 iterations, ${\sim}$US\$100 LLM inference).}
352 productive evaluations of 500 total iterations; 18{,}588 candidate 4-tuples evaluated across 7 lattices.
149 PBB catalog entries deep-verified by MILP (timeouts up to 14{,}400\,s per logical, up to a 60-worker pool); 63 exact outcomes, 33 downward distance corrections identified.
The lattice selection was constrained by MILP runtime: the symplectic formulation is ${\sim}4\times$ slower per logical than CSS, limiting practical verification to the $m \leq 6$ lattice set and excluding the larger-$m$ CSS lattices $(12,12)$ and $(15,12)$.

\paragraph{Total cost.}
The five evolutionary campaigns sum to ${\sim}$US\$237 in LLM inference (per-campaign breakdown above) over ${\sim}140$\,h of wall-clock time (sum of the campaign rows in main text Table~I: $5 + 9.5 + 21 + 92 + 11 = 138.5$\,h); this budget includes the in-loop MILP verification performed by each campaign's evaluator, of which Campaign~4's 92\,h itself absorbs ${\sim}20$\,h ($72{,}000$\,s) of top-code re-verification at $3{,}000$\,s per logical.
The headline ${\sim}$US\$400 figure cited in the abstract additionally covers the ablation arms of main text Sec.~VI\,F (random search at $10^3$ and $10^4$/lattice plus the GA/GA-G arms over 5~seeds each, ${\sim}$US\$70 combined) and exploratory/failed runs (re-tries on flaky API responses, abandoned config sweeps, ${\sim}$US\$90); see the caption of main text Table~I for the same itemization.
The Campaign~5 deep MILP audit on 149 PBB catalog entries (timeouts up to $14{,}400$\,s per logical, up to a 60-worker pool on a 64-core server) is a post-evolution verification step run on a separate machine and is reported \emph{on top of} the ${\sim}140$\,h evolution wall-clock figure rather than within it.
Across all components, MILP verification---not LLM inference---is the dominant compute cost.

\subsection{Fitness trajectory}

Figure~\ref{fig:sm_trajectory} shows the fitness trajectories of Campaigns~1--3 (trinomial search).
The $y$-axis shows the combined score: the sum of highest credible $\FOM = kd^2/n$ values across all evaluated lattices, with trust filtering on $d/\sqrt{n}$ (main text Sec.~V\,D).
The combined score uses BP-OSD distance estimates that are severe overestimates for high-$k$ codes; consequently, the fitness trajectory primarily reflects $k$-optimization rather than true FOM improvement.
The deterministic $k$-only metric $\Sigma_k$ (main text Sec.~VI\,F) provides a robust assessment of search effectiveness.

\begin{figure*}[t]
\centering
\includegraphics[width=\textwidth]{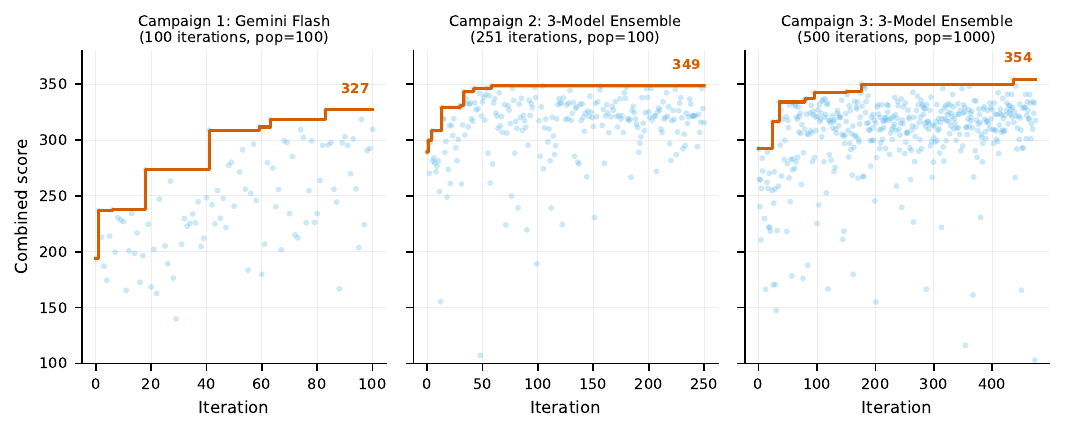}
\caption{Fitness trajectories of Campaigns~1--3.
\emph{Left:} Campaign~1 (Gemini~3 Flash, 100 iterations).
\emph{Center:} Campaign~2 (3-model ensemble, 251 iterations, early-stopped).
\emph{Right:} Campaign~3 (3-model ensemble, 500 iterations, population 1{,}000).
Blue dots: per-iteration combined scores; red step line: running maximum.
All campaigns start from the same seed (score~194).
Campaign~3 saturates after ${\sim}225$ iterations.
Campaign~4 used a different evaluation pipeline and is not shown.}
\label{fig:sm_trajectory}
\end{figure*}

\subsection{Cross-campaign analysis}

The ensemble campaigns (2 and 3) collectively produced 145 unique verified codes not found in Campaign~1, including the highest-$k$ trinomial codes: the $\code{360,40,2}$ ($k = 40$), the $\code{288,32,4}$ ($k = 32$), and the $A = B$ $\code{144,32,2}$ ($k = 32$; $d = 2$ by main text Appendix~D, Theorem~1).
Campaign~4's mixed-monomial generalization later surpasses these encoding dimensions, reaching $k = 54$ via factored-product structures.
Campaign~4 added 45 distinct codes from structural families not present in prior catalogs, with only one polynomial-pair overlap with the Campaign~1--3 catalog: the $\code{288,24,12}$ at $(12,12)$.

Campaigns~2--3 differed from Campaign~1 in three dimensions simultaneously: mutation model, iteration count, and population size.
Campaign~2 (ensemble, pop=100, 251 iterations) achieved $\Sigma_k = 464$, \emph{below} Campaign~1's $704$ despite more iterations and model diversity---suggesting that ensemble mutations may produce less focused variation at small population sizes.
Only Campaign~3 (pop=1{,}000, 500 iterations) matched Campaign~1's $\Sigma_k$ and yielded additional codes.
The most parsimonious explanation is that population size and iteration count are the primary drivers; a controlled experiment varying only model configuration would be needed to isolate any ensemble contribution.

\section{Ablation methodology}
\label{sm:ablation}

This section supplements the ablation study presented in the main text (Sec.~VI\,F and Appendix~B) with methodological details.

\paragraph{Arms.}
Eight arms evaluated on the same 8 Stage-2 lattices:
(1) hand-crafted seed generator;
(2) uniform random trinomials at $10^3$/lattice (single seed);
(3) uniform random at $10^4$/lattice (5 independent seeds, 80{,}000 total per seed);
(4) genetic algorithm (GA) with tournament selection (size~5), single-point polynomial crossover, exponent mutation, 10\% elitism, population~200, $10^4$ evaluations per lattice (5~seeds);
(5) GA operating on the same ansatz representation as the LLM (``GA-G''), with abstract syntax tree (AST)-level mutations on integer literals, loop bounds, and strategy blocks, tournament selection (size~5), strategy-block crossover, population~100, 30~generations, ${\sim}2{,}800$ evaluations per seed (5~seeds, ${\sim}14{,}000$ total);
and (6--8) the highest-fitness evolved ansätze from Campaigns~1--3.

\paragraph{GA-G mutation operators.}
The ansatz GA uses six AST-level operators:
(i) changing integer literals ($\pm 1$, $\pm 2$, or random);
(ii) modifying loop bounds;
(iii) duplicating a strategy block with perturbed constants;
(iv) removing a strategy block;
(v) splicing a new strategy from a template library ($x/y$-swap, perturbation, univariate);
(vi) swapping block order.
Crucially, GA-G \emph{cannot} create entirely new strategy branches: it shuffles and parameterizes existing code structure, whereas the LLM can invent qualitatively new patterns from scratch.

\paragraph{Degenerate code verification.}
Exhaustive verification of the exponent-tuple GA across all 5~seeds and 8~lattices (9{,}832 unique codes) confirms that at every lattice where Campaign~1 found codes, every GA code exceeding the per-lattice $\max\{k\}$ has $d \leq 2$ via explicit weight-2 logical operators ($k/n = 2/3$ in all 11~cases).
The per-lattice breakdown: 6~codes at $(12,6)$, 1~at $(6,12)$, 1~at $(12,12)$, and 3~at $(24,6)$.
One of these 11 has $A = B$, for which $d = 2$ is proven; the others share common polynomial factors or other algebraic degeneracies producing $k/n = 2/3$.

At $(16,9)$ and $(18,8)$---where all evolved ansätze fail---the exponent-tuple GA finds genuine $d \geq 3$ codes.
MILP verification of 10 codes with the highest $k$ ($k = 32$--$64$) at these lattices yields $d = 4$--$6$ (all proven exact) and $\FOM \leq 4.0$, a factor of $3$ below the LLM's $\FOM = 12.0$ and below every Bravyi et al.\ baseline.

GA-G reaches $k = n/2$ at six of eight lattices in every seed and at all eight across the union of seeds, producing codes with $d \leq 2$.
Its low variance ($\pm 22$ vs.\ $\pm 142$ for the exponent-tuple GA) reflects rapid convergence to $k = n/2$ attractors within ${\sim}10$ generations.

\section{MILP formulation details}
\label{sm:milp}

\subsection{CSS distance formulation}

Following Bravyi et al.~\cite{bravyi2024high}, we formulate the problem of finding a minimum-weight $X$-type logical operator for the $j$-th logical qubit as:
\begin{align}
\min \quad & \sum_{i=1}^{n} x_i \label{eq:milp_obj}\\
\text{s.t.} \quad & H_Z \mathbf{x} \equiv \mathbf{0} \pmod{2}, \label{eq:milp_commute}\\
& \langle \mathbf{x}, \bar{Z}_j \rangle \equiv 1 \pmod{2}, \label{eq:milp_anticommute}\\
& x_i \in \{0,1\}, \quad i = 1, \ldots, n, \label{eq:milp_binary}
\end{align}
where $\bar{Z}_j$ is the $j$-th independent $Z$-logical operator.
The mod-2 constraints~\eqref{eq:milp_commute}--\eqref{eq:milp_anticommute} are linearized using integer slack variables: for each row constraint $\sum_i a_i x_i \equiv b \pmod{2}$, we introduce $s \in \ZZ_{\geq 0}$ and replace with $\sum_i a_i x_i - 2s = b$.
The $Z$-type distance is computed analogously by swapping the roles of $H_X$ and $H_Z$; the code distance is $d = \min(d_X, d_Z)$.

We use HiGHS~\cite{huangfu2018highs} via \texttt{scipy.optimize.milp}~\cite{virtanen2020scipy}.
A solution is \emph{exact} when HiGHS reports MIP gap~$= 0$ (proven optimal); otherwise the incumbent solution provides a valid upper bound.
For BB codes, $d_X = d_Z$ because the polynomial involution $x \mapsto x^{-1},\, y \mapsto y^{-1}$ exchanges the roles of $H_X$ and $H_Z$; we compute both as a consistency check.

\subsection{Non-CSS symplectic formulation}

A Pauli operator on $n$ qubits has symplectic representation $(\mathbf{a} \mid \mathbf{b}) \in \FF_2^{2n}$, where $a_i, b_i \in \{0,1\}$ encode the X- and Z-content on qubit~$i$.
We write $\bar{L}_j \in \FF_2^{2n}$ for a representative of the $j$-th independent logical Pauli operator (obtained from \texttt{qldpc}'s symplectic Gaussian elimination), and $\langle \cdot, \cdot \rangle_{\text{symp}}$ for the symplectic inner product.
For non-CSS codes, minimum-weight logical operators are measured in \emph{symplectic} weight: $P = X^{\mathbf{a}} Z^{\mathbf{b}}$ has symplectic weight $|\{i : a_i \neq 0 \text{ or } b_i \neq 0\}|$.
The MILP formulation minimizes this weight:
\begin{align}
\min \quad & \sum_{i=1}^{n} w_i \\
\text{s.t.} \quad & H \cdot (\mathbf{a} \mid \mathbf{b})^\top \equiv \mathbf{0} \pmod{2}, \\
& \langle (\mathbf{a} \mid \mathbf{b}), \bar{L}_j \rangle_{\text{symp}} \equiv 1 \pmod{2}, \\
& w_i \geq a_i, \quad w_i \geq b_i, \quad w_i \leq a_i + b_i, \\
& a_i, b_i, w_i \in \{0,1\},
\end{align}
where $w_i = a_i \lor b_i$ is enforced by the standard linear encoding of the binary OR constraint $w_i \geq a_i$, $w_i \geq b_i$, $w_i \leq a_i + b_i$ (the convex hull of the four feasible $(a_i, b_i, w_i) \in \{0,1\}^3$ triples; tight at integer points), and $H \in \FF_2^{m \times 2n}$ is the symplectic-flipped stabilizer matrix: each row encodes a stabilizer $(s_X \mid s_Z)$ in the order $(s_Z \mid s_X)$ so that the ordinary matrix-vector product $H \cdot (\mathbf{a} \mid \mathbf{b})^\top$ computes the column of symplectic inner products and $H \cdot (\mathbf{a} \mid \mathbf{b})^\top \equiv \mathbf{0}$ enforces pairwise commutation.
We use the X-first convention for the symplectic vector $(\mathbf{a} \mid \mathbf{b}) \in \FF_2^{2n}$ (so $a_i$ is the X-content and $b_i$ the Z-content on qubit~$i$); the row-flip $(s_Z \mid s_X)$ is dictated by this choice, since the symplectic inner product of $(s_X \mid s_Z)$ with $(\mathbf{a} \mid \mathbf{b})$ is $s_X \!\cdot \mathbf{b} + s_Z \!\cdot \mathbf{a}$, which equals the dot product of the flipped row $(s_Z \mid s_X)$ with $(\mathbf{a} \mid \mathbf{b})$. Some references adopt the opposite (Z-first) convention; the row ordering of $H$ flips accordingly. For PBB codes specifically, this commutativity reduces to the explicit condition $A C^\top + B D^\top$ symmetric over $\FF_2$ (main text Sec.~III\,B; \texttt{evaluation/pbb\_code.py}).
At fixed code parameters $n$, the symplectic formulation uses $3n$ binary variables (one each for $a_i$, $b_i$, $w_i$) versus $n$ for the CSS X-distance MILP, and we measure ${\sim}4\times$ longer per-logical solve time on matched codes.

\subsection{Validation and optimality audits}

The MILP formulation was validated on the Bravyi et al.\ baselines $\code{72,12,6}$ and $\code{144,12,12}$~\cite{bravyi2024high}, both confirmed exactly with HiGHS proving optimality (MIP gap~$= 0$) for every logical operator.
MILP solves in seconds for low-$d$ codes and minutes for $d = 12$ codes.

For codes with $d = 2$ (including all $A = B$ codes), the solver finds weight-2 feasible solutions in under one second; combined with the lower bound $d \geq 2$ from the trinomial column-weight structure, this establishes $d = 2$ exactly.
For $d \leq 4$ univariate codes, all MILP instances achieve proven optimality within seconds.

For the strongest distance results in the CSS catalog ($d = 12$ and $d = 14$), we conducted a dedicated optimality audit with 300\,s per-logical timeouts:
the gross code $\code{144,12,12}$ (24 logicals, 13\,min total),
the $\code{288,24,12}$ (48 logicals, 29\,min),
and the $\code{288,16,12}$ (32 logicals, 80\,min) all achieved proven optimality (MIP gap~$= 0$) for every logical operator; these distances are therefore exact.
For the $\code{360,16,{\leq}14}$ ($n = 360$), the solver found weight-14 solutions for all 32 logicals and proved optimality for 7; the remaining 25 returned incumbents at weight~14 or~16 within the 300\,s timeout.
Since the minimum weight across all logicals is~14, $d \leq 14$ is established; the 7 logicals with proven weight-14 optima confirm that no lower distance is missed on those operators, and no logical returned a weight below~14.
The full per-logical audit data are included in the public repository.

For Campaign~4, all 1{,}188 codes passing evolutionary fitness filters were MILP-verified with 120--600\,s per-logical timeouts.
The top~39 codes were re-verified with extended budgets (3{,}000\,s per logical, 72{,}000\,s total).
For Campaign~5 non-CSS codes, 149 PBB catalog entries underwent deep MILP verification (timeouts up to 14{,}400\,s per logical, up to a 60-worker pool on a 64-core server), producing 63 exact outcomes and identifying 33 downward distance corrections.

\bibliography{references}